\documentclass[a4paper,11pt]{article} 
\usepackage[top = 2.5cm, bottom = 2.5cm, left = 2.5cm, right =2.5cm]{geometry} 

\usepackage[T1]{fontenc}
\usepackage[utf8]{inputenc}
\usepackage{multirow} 
\usepackage{booktabs} 
\usepackage{graphicx} 
\usepackage{setspace}
\usepackage{color}
\usepackage{bm}

\usepackage{enumitem}
\usepackage{tikz}
\usepackage{etoolbox}
\newcommand{\circled}[2][]{\tikz[baseline=(char.base)]
    {\node[shape = circle, draw, inner sep = 0.5pt]
    (char) {\phantom{\ifblank{#1}{#2}{#1}}};%
    \node at (char.center) {\makebox[0pt][c]{#2}};}}
\robustify{\circled}

\usepackage{float}

\usepackage{fancyhdr}
\usepackage{amsmath}

\usepackage{amssymb}
\usepackage{bbold}
\usepackage{amsthm}
\numberwithin{equation}{section}
\usepackage{bbm}
\theoremstyle{definition}
\usepackage[capitalize]{cleveref}
\usepackage{natbib}
\setcitestyle{authoryear,open={(},close={)}}

\DeclareMathOperator*{\argmax}{arg\,max}

\newtheorem{thm}{Theorem}[]
\newtheorem{theorem}[thm]{Theorem}
\newtheorem{thm3}{Theorem}[]
\newtheorem{definition}[thm3]{Definition}
\newtheorem{thm1}{Theorem}[]
\newtheorem{lemma}[thm1]{Lemma}
\newtheorem{thm7}{Theorem}[]
\newtheorem{corollary}[thm7]{Corollary}
\newtheorem{thm4}{Theorem}[]
\newtheorem{proposition}[thm4]{Proposition}

\newtheorem{thm2}{Theorem}[]
\newtheorem{assumption}[thm2]{Assumption}

\newtheorem{thm5}{Theorem}[]
\newtheorem{claim}[thm5]{Claim}
\newtheorem{thm6}{Theorem}[]
\newtheorem{example}[thm6]{Example}

\title{Flexible Moral Hazard Problems with Adverse Selection}
\author{Siwen Liu\thanks{Department of Economics, University of Bonn (siwen.liu@uni-bonn.de).}}
\date{October , 2025}

\begin{document}

\maketitle
\begin{abstract}
     We study a moral hazard problem with adverse selection in which a risk-neutral agent can directly control the output distribution and possesses private information about the production environment. The principal designs a menu of contracts, subject to limited liability, to motivate the agent to exert effort. Departing from classical models, the principal not only induces specific levels of effort by designing the "power" of the contracts but also regulates the supports of the implemented output distributions through shaping the "range" of the contracts. We characterize the set of output distributions that can be implemented in this environment. Our results show that it is either optimal for the principal to provide a single full-range contract, or the optimal low-type contract range excludes some high outputs, or the optimal high-type contract range excludes some low outputs. We provide sufficient and necessary conditions on when a single full-range contract is optimal under convex effort functions, and show that this condition remains sufficient with general effort functions.
\end{abstract}

\section{Introduction}
This paper studies a contract design problem that involves both moral hazard and adverse selection. It focuses on situations in which a principal (she) uses contracts to motivate an agent (he) to undertake costly hidden actions (moral hazard), while the agent simultaneously possesses private information about the production environment (adverse selection). The goal is to investigate the optimal menu of contracts for the principal in such situations. This paper employs the flexible moral hazard framework, based on the work of \cite*{georgia22}, and analyzes a model in which the agent can directly control the output distribution. 

In the model, the risk-neutral agent is capable of choosing an arbitrary distribution over the output space, incurring a cost determined by the distribution chosen. The principal can provide contracts that make payments to the agent contingent on the realized output. The contracts must satisfy limited liability, which requires that all payments be non-negative. The agent possesses private information about his cost function: the high-type (low-type) agent faces lower (higher) costs and lower (higher) cost margins. Under adverse selection, the principal designs a menu of contracts from which the agent can choose, aiming to maximize her expected profit.

The agent's ability to directly control output distributions contrasts with the classical moral hazard model, where the agent is limited to choosing an effort level that parameterizes the output distribution. This formulation captures real-world scenarios where the agent's action space is sufficiently rich to allow the implementation of any output distribution.

We focus on scenarios where the agent's cost of choosing an output distribution is moment-based: it depends solely on the expectation of a function over this distribution. This function, interpreted as the effort function, specifies the effort needed to produce each output with certainty. The corresponding expectation is interpreted as the \textit{aggregate effort} needed to generate the distribution. Aggregate effort plays a central role in the decision problems of both parties. In the absence of adverse selection, we show that the principal's profit-maximizing problem reduces to selecting an aggregate effort level to implement. Under adverse selection, how the aggregate effort is distributed across different output levels becomes important in the design problem. 

A prominent feature of our model is that the principal is equipped with two instruments to design contracts. The first instrument, the \textit{power} of a contract, measures the steepness of the contract and determines the level of aggregate effort the agent exerts. A contract with a higher power generally induces the agent to work more. The second instrument, the \textit{range} of a contract, denotes the set of outputs for which the contract provides maximal production incentives. The support of any output distribution implemented by the contract is necessarily a subset of its range. Under adverse selection, the power instrument motivates different types of agents to exert aggregate effort levels preferred by the principal, while the range instrument facilitates screening, confining agents to produce outputs within designated contract ranges. 

Contract range is a key design element in the model. It captures the principal's trade-off between enhancing production value and reducing information rents. Expanding the contract ranges broadens the set of output distributions that can be implemented, thereby increasing the expected value of production. Conversely, narrowing the contract ranges allows the principal to tailor contracts more specifically to each agent type, making them less appealing to the other type. In this way, the principal reduces the screening costs incurred. 

We begin by characterizing the set of output distributions that can be implemented under adverse selection. The characterization result highlights the interplay between the aggregate efforts and the supports of the output distributions in determining implementability (Theorem \ref{IMP-pp2}). In particular, we identify cases in which implementability is guaranteed only when the supports of the distributions overlap in a limited way.

We analyze the structure of the principal's optimal menus. In any optimal menu, the low-type contract's power must be no larger than that of the high-type contract (Proposition \ref{OM-pp1}). Moreover, the range of the low-type contract can always be extended to the leftmost point to include all the low outputs, and the range of the high-type contract can be similarly extended to the rightmost point to include all the high outputs, without harming the menu's feasibility (Proposition \ref{OM-pp2}). It implies that granting the low-type (high-type) agent maximal incentives to produce low-cost, low-return (high-cost, high-return) outputs does not undermine feasibility. Finally, both contract ranges from an optimal menu can be extended to full without violating feasibility if and only if the contracts share the same power (Theorem \ref{OM-pp3}). 

 Theorem \ref{OM-pp3} implies that either it is optimal for the principal to provide a single full-range contract, or, at the optimal menu, some high outputs must be excluded from the low-type contract range, or it's necessary to exclude some low outputs from the high-type contract range. The optimality of a single full-range contract means that there is no need for the principal to differentiate between the agents, nor a necessity to exclude any output levels from the contract ranges. By contrast, when differentiation becomes essential, the exclusion of certain outputs is necessary to reduce the agent's payoff from misreporting. Given the same contract, the high-type agent intends to produce higher outputs than the low-type agent. Thus, excluding high outputs from the low-type contract can effectively reduce the high-type agent's incentives to misreport, and excluding low outputs from the high-type contract achieves the reverse. These findings provide insight into the structure of optimal task allocations between heterogeneous employees in practical settings.

We next consider the case where the effort function is convex. In this setting, the agent's cost of choosing an output distribution rises with its risk. We show that the optimal menu features contracts with equal powers if and only if, in the corresponding pure moral hazard problems, the optimal low-type contract has weakly lower power than the optimal high-type contract (Theorem \ref{EF-th1}). Moreover, this condition is both sufficient and necessary for when the optimum of the principal's problem can be reached by providing a single full-range contract. Furthermore, we examine distortions in aggregate efforts and welfare outcomes at optimal menus under adverse selection, relative to the pure moral hazard and first-best benchmarks (Proposition \ref{EF-pp1}).

We further demonstrate that, with any general effort function, providing a menu consisting of a single full-range contract is optimal if providing this contract is also optimal in the principal's problem where the effort function is replaced by its lower convex envelope (Theorem \ref{OSC-th1}). This finding implies that the condition in Theorem \ref{EF-th1} is also sufficient for the optimum of the principal's problem to be reached by a single full-range contract.  

Finally, we generalize Theorem \ref{OM-pp3} to two broader settings. The first extends the agent's type space to a finite set. The second considers an arbitrary type space and an arbitrary structure on cost functions, subject to the restriction that the effort function is affine. This is the case where the output distribution enters the cost functions through its mean. For both settings, we show that the optimal menu either consists of a single full-range contract or includes some contracts with truncated ranges (Proposition \ref{NT-pp1} and Proposition \ref{OSG-pp1}). Thus, whenever the principal employs contracts with varying powers at an optimal menu, she must tailor their ranges to impose strict restrictions on the supports of the output distributions chosen by the agent. \\
\\
\textbf{Related Literature}: The study of principal-agent problems featuring moral hazard can be traced back to seminal works such as \cite{mirrlees} and \cite{holmstrom79}. These works, along with a substantial body of subsequent research, analyze the design of optimal incentive contracts through which a principal motivates an agent to exert costly effort. Most of the papers in this literature assume that the agent selects an effort level, which in turn parameterizes the resulting output distribution. By contrast, \cite{holm_milgrom_87} introduces a moral hazard setting in which the agent directly chooses the output distribution itself, rather than selecting an effort level that indirectly shapes it. The cost of choosing any output distribution is specified by a cost function. This formulation grants the agent high flexibility in determining the output distribution. We refer to this approach as the flexible moral hazard framework. 

Recent explorations of the flexible moral hazard framework include \cite{hebert18} and \cite{krahmer2025}, which analyze security design problems where a seller issues a security and directly shapes the probability distribution of the asset's return. \cite*{barron20} considers a moral hazard problem where the agent can costlessly add mean-preserving noise to the output distribution. \cite{mattsson23} and \cite{Bonham21} both study contract design problems where the agent can arbitrarily alter the output distribution, incurring a cost based on the Kullback-Leibler divergence between the initial and altered distributions. 

Our study builds on \cite{georgia22}, which develops a first-order approach to the flexible moral hazard problems with general cost functions that are smooth and convex. \cite{georgia22} provides a first-order characterization of when an output distribution can be implemented by a contract, as well as a necessary first-order condition for the principal's optimal contract. Extending \cite{georgia22}, we incorporate adverse selection into the moral hazard problem and study how asymmetric information influences the design of contracts.

 This paper is also related to the literature on moral hazard problems with adverse selection (\cite{laffont1986}, \cite{baron1987}, \cite{picard1987}, \cite{melumad}, \cite{faynzi}). More recently, \cite{gottlieb2017} studies a moral hazard problem with binary outputs, binary effort levels, and two-dimensional private information, highlighting issues of exclusions and distortions in the optimal menus. \cite{gottlieb22} shows that when the agent is risk-neutral and protected by limited liability, under an assumption that restricts the set of possible output distributions, it is always optimal for the principal to provide a single contract. \cite{castro21} extends this result to settings with a risk-averse agent under additional assumptions. \cite*{chade21} investigates conditions under which moral hazard and adverse selection can be decoupled, thereby reducing the moral hazard problem with adverse selection to a pure screening problem. While all the aforementioned studies are conducted under the classical moral hazard framework, our study employs the flexible moral hazard framework.
 
 \cite{krahmer2024} explores a screening problem where the agent can invest in the distribution of his private type. The study shows that the optimal contracts feature finite options under moment-based cost functions. Our study differs by focusing on a setting in which the agent is initially endowed with private information and subsequently chooses a hidden action (adverse selection followed by moral hazard). In contrast, \cite{krahmer2024} examines a setting where the agent first invests in his type, which then becomes the input to the subsequent screening problem (moral hazard followed by adverse selection). 

 Another closely related work is \cite*{castro25}, which looks into a problem that involves both flexible moral hazard and adverse selection. \cite{castro25} provides a characterization result on incentive-compatible menus and utilizes this result in several optimal contract design settings. Their approach reduces the problem into a pure screening problem by focusing on on-path incentive compatibility constraints and subsequently verifying whether all off-path deviations can be prevented. In contrast, our model cannot make use of their approach because of the imposed limited liability constraints and the complexity of the off-path deviations. Instead, we leverage the analytical tractability of the flexible moral hazard framework to directly deal with the double deviation issues that arise from the joint presence of moral hazard and adverse selection.\\
 \\
 This paper is organized as follows. In Section \ref{MBCF-model}, we set up the model and formulate the principal's problem. In Section \ref{MBCF-imp}, we present the characterization of implementable output distributions. Section \ref{MBCF-om} demonstrates the results on optimal contract powers and optimal contract ranges. Section \ref{MBCF-cef} focuses on the problem with convex effort functions. Section \ref{MBCF-os} explores the optimality of single full-range contracts. In Section \ref{MBCF-general}, we extend our adverse selection setting beyond binary types. Finally, we provide some concluding remarks in Section \ref{MBCF-con}.
\section{Model}
\label{MBCF-model}
  We consider a moral hazard problem in which a risk-neutral principal (she) hires a risk-neutral agent (he) to take costly actions to produce outputs, under a setting where both parties share symmetric information about the production environment. The agent can directly choose an output distribution $\mu \in \Delta([\underline{x},\overline{x}])$, where $[\underline{x},\overline{x}]\subset \mathbb{R}_{+}$. An output realization $x \in [\underline{x},\overline{x}]$ indicates the monetary value of the output to the principal.
   
   The agent has a binary type $t \in \{0,1\}$ that affects his cost function. It is known to both the principal and the agent that $t = 0$ occurs with probability $p_0\in (0,1)$, while $t=1$ occurs with probability $p_1 = 1-p_0$. However, the realization of $t$ is private to the agent. Let $C_{t}:\Delta([\underline{x},\overline{x}])\to \mathbb{R}_{+}$ denote the cost function corresponding to type $t \in \{0,1\}$. We assume that the cost functions are "moment-based" and take the following form 
\begin{equation*}
C_t(\mu) = K_t\big(\int c(x)\mu(dx)\big),\quad \forall \mu \in \Delta([\underline{x},\overline{x}]),\quad\forall t\in \{0,1\}.
\end{equation*}

We assume that $c:[\underline{x},\overline{x}] \to \mathbb{R}_{+}$ is differentiable and strictly increasing, and $K_t: [\underline{c},\overline{c}]\to \mathbb{R}_{+}$ ($t\in\{0,1\}$) is strictly increasing, strictly convex and twice differentiable, where $\underline{c}:= c(\underline{x})$ and $\overline{c} := c(\overline{x})$. We standardize $\underline{c}$ to be 0. According to \cite{georgia22}, the monotonicity of $c(\cdot)$ guarantees that the cost function $C_t(\cdot)$ ($t\in\{0,1\}$) is increasing in the sense of first-order stochastic dominance (if distribution $\mu'$ first-order stochastically dominates $\mu$ then $C_t(\mu')\geq C_t(\mu)$). The convexity of $K_t(\cdot)$ guarantees that $C_t(\cdot)$ is convex.
 
 We can interpret $c(\cdot)$ as the effort function where $c(x)$ is the amount of effort needed to produce $x$ with probability 1. We refer to $\int c(y)\mu(dy)$ as the \textit{aggregate effort} needed to generate output distribution $\mu \in \Delta([\underline{x},\overline{x}])$. Function $K_t(\cdot)$ indicates how the total cost depends on the aggregate effort given type $t\in \{0,1\}$. For convenience, sometimes we use $k_t(\cdot)$ to denote $K_t'(\cdot)$.
 
Note that the agent's effort function is consistently $c(\cdot)$ irrespective of his type, although the cost's dependency on aggregate effort varies according to type. We assume that $K_0(\underline{c}) = K_1(\underline{c}) = 0$ and $K_0(\alpha) > K_1(\alpha)$ for all $\alpha\in (\underline{c},\overline{c}]$.\footnote{The main results in this paper still hold in generalized settings where we only assume that $\underline{c} \geq 0$ and $K_0(\underline{c})\geq K_1(\underline{c})\geq 0$, while imposing a limited liability constraint such that the payment from the principal to the agent specified by any contract is no smaller than a constant $K\geq 0$.} It is less costly for a type-1 agent to produce outputs than a type-0 agent, given a certain level of aggregate effort.

 The agent's choice of output distribution is not contractible. However, the principal can provide a contract contingent on the realized output to motivate the agent to produce. A contract $s:[\underline{x},\overline{x}] \to \mathbb{R}$ is a measurable function that specifies the payment $s(x)$ from the principal to the agent when the realized output is $x\in [\underline{x},\overline{x}]$. We further require that $s(\cdot)$ satisfies limited liability, i.e., $s(x)\geq 0$ for all $x\in [\underline{x},\overline{x}]$.

The principal optimally designs a menu of contracts to maximize her expected profit. According to \cite{myerson82}, we can focus on the design of a direct incentive-compatible mechanism $(s_0(\cdot),s_1(\cdot),\mu_0,\mu_1)$, where $s_0(\cdot)$, $s_1(\cdot)$ are contracts allocated to the agent if he reports his type to be 0 and 1 respectively, and $\mu_0$, $\mu_1$ are the output distributions recommended to the agent corresponding to the reported type. 

The timeline of the model is as follows. First, the agent observes his private type $t\in \{0,1\}$. Then, the principal provides a menu of contracts (with recommended output distributions) $(s_0(\cdot),s_1(\cdot),\mu_0,\mu_1)$ to the agent\footnote{We can allow the agent to choose between accepting the menu or rejecting the menu and receiving his outside option with zero utility. Because the limited liability constraint ensures that the agent's utility from the menu is always non-negative, the agent is always willing to accept. Therefore, this step is omitted.}. The agent reports his type $t$ and gets the corresponding contract $s_{t}(\cdot)$. Next, the agent chooses an output distribution. Finally, the output is realized and the contract is carried out.

\subsection{Baseline Problem}
\label{MBCF-BP}
This section introduces the baseline problem without adverse selection. In the baseline problem, the agent's type $t$ is observed by the principal. The principal aims to design a contract to maximize her profit, knowing that the agent's cost function is given by $C_t(\mu) = K_t\big(\int c(y)\mu(dy)\big)$ for all $\mu \in \Delta([\underline{x},\overline{x}])$.

We introduce the notion of concavification that will be useful in the analysis that follows. For any function $g:[\underline{c},\overline{c}]\to \mathbb{R}$, we say $\widehat{g}:[\underline{c},\overline{c}]\to \mathbb{R}$ is the concavification of $g$ if $\widehat{g}$ is the smallest concave function that is everywhere weakly greater than $g$: 
\begin{equation*}
    \widehat{g}(\alpha): = \inf\{u(\alpha)\mid u\text{ is concave and }u\geq g\text{ over }[\underline{c},\overline{c}]\}.
\end{equation*}
\textbf{Agent's Problem:} Given a bounded and measurable contract $s(\cdot)$, the type-$t$ agent chooses an output distribution to maximize his expected payoff:
 \begin{equation}
 \max_{\mu\in \Delta([\underline{x},\overline{x}])} \int s(x)\mu(dx) - C_t(\mu).
 \label{PRE-A}
 \end{equation}

 We say contract $s(\cdot)$ implements distribution $\mu$ if $\mu$ solves problem \eqref{PRE-A}. Define $\theta:[\underline{c},\overline{c}]\to [\underline{x},\overline{x}]$ the inverse function of $c(\cdot)$. We have the following result on the aggregate effort corresponding to the output distribution that is implemented by a contract $s(\cdot)$. 
 \begin{lemma}
    Take an optimal solution $\mu^{*}$ to problem \eqref{PRE-A} and take the corresponding aggregate effort $\alpha^{*}: = \int c(y)\mu^{*}(dy)$. Then it holds that 
    \begin{equation}
    \alpha^{*} = \argmax_{\alpha\in [\underline{c},\overline{c}]} \big[\widehat{s\circ \theta}(\alpha) - K_t(\alpha)\big].
    \label{PRE-aalpha}
    \end{equation}
\label{PRE-lm2}
\end{lemma}

 The idea of Lemma \ref{PRE-lm2} is illustrated as follows. To solve \eqref{PRE-A}, we first consider problem \eqref{PRE-Aalpha} where the agent focuses on choosing an optimal distribution among those with an aggregate effort of $\alpha$.
\begin{equation}
\max_{\mu \in \Delta([\underline{x},\overline{x}])}  \int s(x)\mu(dx) - K_t(\alpha) \quad
s.t.  \int c(x)\mu(dx) = \alpha. 
\label{PRE-Aalpha}
\end{equation}

We can reformulate \eqref{PRE-Aalpha} as an optimization problem in the effort space rather than the output space by a change of variables. We have that problem \eqref{PRE-Aalpha} is equivalent to 
\begin{equation}
\max_{\gamma \in \Delta[\underline{c},\overline{c}]}  \int s\circ \theta (c)\gamma(dc) - K_t(\alpha) \quad s.t. \int c\gamma(dc) = \alpha,
\label{PRE-Agamma}
\end{equation}
where $s\circ \theta:[\underline{c},\overline{c}]\to \mathbb{R}$ is the composition of $s(\cdot)$ and $\theta(\cdot)$. $s\circ \theta(\cdot)$ can be regarded as the contract written in terms of efforts instead of outputs ($s\circ \theta(c)$ is the payment to the agent if the realized output costs effort $c$). In problem \eqref{PRE-Agamma}, we choose an effort distribution with aggregate effort $\alpha$ that maximizes the agent's payoff. Since $\alpha$ is fixed in \eqref{PRE-Agamma}, any optimal solution $\gamma$ to \eqref{PRE-Agamma} maximizes $\int s\circ \theta(c)\gamma(dc)$ while satisfying $\int c\gamma(dc) = \alpha$. Hence, the optimal value of \eqref{PRE-Agamma} is given by $\widehat{s\circ\theta}(\alpha) - K_t(\alpha)$, ,  where $\widehat{s\circ \theta}(\cdot)$ is the concavification of $s\circ \theta(\cdot)$ on interval $[\underline{c},\overline{c}]$\footnote{This concavification technique has been used in similar contexts (see, e.g., \cite{BP2011}). The expression $\widehat{s\circ \theta}(\alpha) - K_t(\alpha)$ is the optimal value of problem \eqref{PRE-Agamma}, provided that the maximum exists. If the maximum does not exist, $\widehat{s\circ \theta}(\alpha) - K_t(\alpha)$ denotes the supremum payoff the agent can approximate.}.

To maximize his expected payoff, the agent chooses the optimal level of aggregate effort $\alpha$ to maximize $\widehat{s\circ \theta}(\alpha) - K_t(\alpha)$, as stated in the following lemma. This maximizer is unique since $\widehat{s\circ \theta}(\alpha) - K_t(\alpha)$ is strictly concave.\\
\\
\noindent\textbf{Principal's Problem:} The principal optimally designs a contract to maximize her profit and recommends to the agent an output distribution $\mu$ that the proposed contract can implement. Suppose that the principal would like to implement distribution $\mu$ by contract $s(\cdot)$. According to \cite{georgia22}, it is necessary that there exists a constant $m\in \mathbb{R}$ such that\footnote{The first equality in \eqref{PRE-FOC} holds $\mu$-almost everywhere. We omit the discussion of this technicality in the rest of the paper since it doesn't affect the analysis.} 
\begin{equation}
 s(x) \left\{
 \begin{array}{ll}
  = K_t'\big(\int c(y)\mu(dy)\big)\cdot c(x) + m    &\text{ if }x\in \text{supp}(\mu), \\
  \leq K_t'\big(\int c(y)\mu(dy)\big)\cdot c(x) + m     &\text{ if } x \notin \text{supp}(\mu).
 \end{array}
 \right.
 \label{PRE-FOC}
\end{equation}

Here, the term $K_t'\big(\int c(y)\mu(dy)\big)c(\cdot)$ is the Gateaux derivative of $C_t(\cdot)$ at $\mu$. The right-hand side of \eqref{PRE-FOC} quantifies the marginal cost associated with adjustments to the distribution around $\mu$, while the left-hand side of \eqref{PRE-FOC} is the contract payment, indicating the marginal gain from such adjustments. \eqref{PRE-FOC} presents how a contract is pinned down by the distribution it aims to implement within the support of the targeted distribution.

Under \eqref{PRE-FOC} and $\underline{c} = 0$, limited liability implies $m\geq 0$. Since $m$ is a constant term in the contract that doesn't affect the implemented distribution, we must have $m=0$ at the optimal contract. Hence, the principal's problem can be formulated as
\begin{equation}
\begin{aligned}
    \max_{s(\cdot), \mu \in \Delta([\underline{x},\overline{x}])} & \int (x-s(x))\mu(dx) \\
    s.t. \quad & s(x) = K_t'\big(\int c(y)\mu(dy)\big)\cdot c(x), \text{ for }x \in \text{supp}(\mu)\\
     & 0 \leq s(x)\leq K_t'\big(\int c(y)\mu(dy)\big)\cdot c(x), \text{ for }x \notin \text{supp}(\mu).
\end{aligned}
\label{PRE-P}
\end{equation}

Note that when we are stating the principal's problem, we assume that the agent chooses the recommended distribution as long as there doesn't exist a strictly better alternative, even though there can be other output distributions that yield the same payoff to him. 

Similar to Lemma \ref{PRE-lm2}, we have the following result on the aggregate effort corresponding to the optimal output distribution that solves the principal's problem. We define $\Theta(\cdot):= \hat{\theta}(\cdot)$ as the concavification of $\theta(\cdot)$ on the interval $[\underline{c},\overline{c}]$. 

\begin{lemma}
Take an optimal solution $s^{*}(\cdot)$, $\mu^{*}$ to problem \eqref{PRE-P} and take the corresponding aggregate effort $\alpha^{*}:= \int c(y)\mu^{*} (dy)$. Then it holds that 
\begin{equation}
    \alpha^{*} \in \argmax_{\alpha \in [\underline{c},\overline{c}]} \big[ \Theta(\alpha) - K_t'(\alpha)\cdot \alpha\big].
\label{PRE-sol}
\end{equation}
\label{PRE-lm1}
\end{lemma}

The idea of Lemma \ref{PRE-lm1} is illustrated as follows. Fix a pair of $s(\cdot)$ and $\mu$ feasible to problem \eqref{PRE-P}. If the aggregate effort of $\mu$ is equal to $\alpha$, i.e., $\int c(y)\mu(dy) = \alpha$, the agent's expected payoff from $s(\cdot)$ and $\mu$ is given by 
\begin{equation*}
\begin{aligned}
\int s(x)\mu(dx) - C_t(\mu) &=\int K_t'\big(\int c(y)\mu(dy)\big)\cdot c(x)\mu(dx) -K_t\big(\int c(y)\mu(dy)\big)\\
&= K_t'(\alpha)\cdot\alpha - K_t(\alpha),
\end{aligned}
\end{equation*}
and the principal's profit can be written as
\begin{equation*}
\begin{aligned}
    \int (x-s(x))\mu(dx) &= \int x \mu(dx) - K_t'(\alpha)\cdot \alpha. 
\end{aligned}
\end{equation*}

Among all the output distributions that induce an aggregate effort of $\alpha \in [\underline{c},\overline{c}]$, the profit-maximizing one is the optimal solution to
\begin{equation}
\max_{\mu \in \Delta([\underline{x},\overline{x}])}  \int x \mu (dx) - K_t'(\alpha) \cdot \alpha 
\quad s.t. \int c(x)\mu(dx) = \alpha.
\label{PRE-alpha}
\end{equation}

Similarly, \eqref{PRE-alpha} is equivalent to the following problem in the effort space:
\begin{equation}
\max_{\gamma \in \Delta([\underline{c},\overline{c}])} \int \theta(c)\gamma(dc) - K_t'(\alpha)\cdot \alpha \quad s.t. \int c\gamma(dc) = \alpha. 
\label{PRE-gamma}
\end{equation}

In problem \eqref{PRE-gamma}, we choose an effort distribution with aggregate effort $\alpha$ that maximizes the principal's profit. Then, any distribution with a mean $\alpha$ that attains the concavification value $\Theta(\alpha)$ is optimal to problem $\eqref{PRE-gamma}$. Given that the principal focuses on inducing an aggregate effort of $\alpha$, $\Theta(\alpha) - K_t'(\alpha)\cdot \alpha$ is her maximal attainable profit. The second term $K_t'(\alpha)\cdot \alpha$ is the total payment to the agent, while the first term $\Theta(\alpha)$ is the efficient production value generated when the aggregate effort $\alpha$ is distributed over outputs in the most productive way. The principal chooses an aggregate effort level to maximize this maximal attainable profit, as stated in the Lemma \ref{PRE-lm1}. 

Consequently, we can approach the principal's problem \eqref{PRE-P} in the following way. We first take the optimal level of aggregate effort $\alpha^{*}$ that solves \eqref{PRE-sol}. Given $\alpha^{*}$, we take the optimal effort distribution $\gamma^{*}$ from \eqref{PRE-gamma}. Take $\mu^{*}$ as the pushforward measure of $\gamma^{*}$ such that $\mu^{*} = \theta \# \gamma^{*}$ (i.e., for every Borel set $B \subseteq [\underline{x},\overline{x}]$, there is $\mu^{*}(B) = \gamma^{*}(\theta^{-1}(B))$). Then $\mu^{*}$ is a probability measure over $[\underline{x},\overline{x}]$ and is the optimal distribution for the principal to recommend to the agent. Take any contract $s^{*}(\cdot)$ that satisfies the constraints in \eqref{PRE-P} with $\mu^{*}$. Then $s^{*}(\cdot), \mu^{*}$ is an optimal solution to \eqref{PRE-P}. 

Given $\alpha^{*}$ that is optimal to \eqref{PRE-sol}, if $\theta(\alpha^{*}) = \Theta(\alpha^{*})$, a point mass distribution with support only at $\alpha$ solves \eqref{PRE-gamma}. If $\theta(\alpha^{*})\neq \Theta(\alpha^{*})$, there exists an optimal solution to \eqref{PRE-gamma} that is supported on only two points\footnote{\eqref{PRE-gamma} is an optimization problem with a linear objective and a linear constraint. Its optimum is attainable by some extreme point of the feasible set. By \cite{winkler}, any such extreme point takes the form of a discrete probability measure that assigns strictly positive mass to at most two points. \label{MBCF-EP}}. We can conclude that the maximal principal's profit can be attained by inducing an output distribution supported on either one or two points. 

For $t\in \{0,1\}$, when the solution to \eqref{PRE-sol} is unique, we define $\alpha_t^{MH}:= \argmax_{\alpha\in [\underline{c},\overline{c}]}\big[\Theta(\alpha) - K_t'(\alpha)\cdot \alpha\big]$ as the \textit{pure moral hazard solution} to the principal's problem.\\
\\
\textbf{First-best Contract: }The first-best problem for the principal is given by
\begin{equation}
    \max_{\mu\in \Delta([\underline{x},\overline{x}])} \int x\mu(dx) - K_t\big(\int c(y)\mu(dy)\big).
\label{PRE-FB}
\end{equation}

Following the previous idea, the optimal first-best output distribution corresponds to an aggregate effort that solves $\max_{\alpha \in [\underline{c},\overline{c}]} \big[ \Theta(\alpha) - K_t(\alpha)\big]$.

\cite{georgia22} provides first-order conditions on the optimal solutions to the agent's problem \eqref{PRE-A} and the principal's problem \eqref{PRE-P}. Those conditions can be used to verify whether a candidate solution is optimal, but do not directly provide the optimal solutions. Unlike their findings, Lemma \ref{PRE-lm2} and \ref{PRE-lm1} offer direct methods to solve the problems with moment-based cost functions. 

In the given formulation, aggregate effort is central to the decision problems of both parties, with output distributions chosen to attain the concavification value corresponding to the optimal aggregate effort. Lemma \ref{PRE-lm2} also provides an alternative way to understand why the first-order approach is always valid in our setting. Given an arbitrary contract $s(\cdot)$, according to Lemma \ref{PRE-lm2}, the agent would act as if he is provided with a contract $\widehat{s\circ \theta}(\cdot)$ that is concave in the aggregate effort. This concavification can be done because the agent is able to choose any arbitrary output distribution. Hence, the agent is faced with a convex decision problem, regardless of the shape of the contract he is provided with. This concavification idea is similar to \cite{barron20}, which shows that concave contracts should be provided when the agent can add arbitrary mean-preserving noise to his output costlessly.

The derivations of Lemma \ref{PRE-lm2} and \ref{PRE-lm1} depend on the simple form of the first-order condition \eqref{PRE-FOC} when the cost function is moment-based. Beyond this technical benefit, moment-based cost functions accommodate a wide range of cost structures relevant to practical scenarios. For instance, an affine effort function $c(\cdot)$ represents situations where the cost of selecting a distribution depends solely on its mean. When $c(\cdot)$ is convex (concave), the cost associated with choosing a distribution rises (falls) with its risk level. 
\subsection{Feasible Menu and Principal's Problem}
In the problem with adverse selection, the principal aims to maximize her profit by optimally designing a menu $(s_0(\cdot),s_1(\cdot),\mu_0,\mu_1)$ that satisfies the \textit{moral hazard (MH)}, \textit{limited liability (LL)} and \textit{incentive compatibility (IC)} constraints. Moral hazard requires that the output distribution $\mu_t$ ($t\in \{0,1\}$) can be implemented by contract $s_t(\cdot)$. Limited liability mandates that both contracts are non-negative. Incentive compatibility ensures that it is optimal for the agent to truthfully report his type.

Rather than directly formulating the principal's problem, we will reformulate the set of feasible menus by decomposing a contract into several design elements. We will focus on designing these elements instead of designing the contracts directly.\\
\\
\textbf{Design Elements:} If a non-negative contract $s_t(\cdot)$ induces the type-$t$ ($t\in\{0,1\}$) agent to choose the distribution $\mu_t$, whose corresponding aggregate effort is equal to $\alpha_t = \int c(y)\mu_t(dy)$ ($\alpha_t$ is unique due to Lemma \ref{PRE-lm2}), then the first order condition establishes that there exists a constant $m_t\in \mathbb{R}_{+}$, such that \begin{equation}
s_t(x) \leq K_t'(\alpha_t)\cdot c(x) + m_t, \quad \forall x\in [\underline{x},\overline{x}],
\label{MBC-eq1}
\end{equation}
and the inequality is binding for any $x\in \text{supp}(\mu)$.

Suppose we consider the set of outputs where \eqref{MBC-eq1} is binding. These are the outputs where the principal provides maximal incentives for the agent to produce. Any output distribution with an aggregate effort $\alpha_t$ can be implemented by $s_t(\cdot)$ if and only if the support of the distribution is a subset of this set. We can alter $s_t(\cdot)$ by setting the contract value outside of this set to zero without harming the implementation ability of the contract. We will focus on contracts of this kind when analyzing the principal's problem. 

Formally, we restrict ourselves to providing a contract $s_t(\cdot)$ ($t\in \{0,1\}$), which takes the following form, to the type-$t$ agent.\footnote{In Appendix \ref{MBCF-app-model}, Proposition \ref{MBCF-app-pp1}, we establish that this restriction is without loss of generality for the profit-maximizing principal under adverse selection.} 
\begin{equation}
s_t(x) = \mathbb{1}_{\{x\in R_t\}}\big[K_t'(\alpha_t)\cdot c(x) + m_t\big], \quad \forall x\in [\underline{x},\overline{x}].
\label{MBC-eq2}
\end{equation}

In \eqref{MBC-eq2}, $R_t$ is the \textit{contract range}, which is the set of outputs where maximal incentives are provided to the agent. We require $R_t$ to be a measurable subset of $[\underline{x},\overline{x}]$. Outside of the range $R_t$, the contract $s_t(\cdot)$ takes zero value. Within the range $R_t$, $s_t(\cdot)$ coincides with $K_t'(\alpha_t) \cdot c(\cdot) + m_t$, where $\alpha_t$ is the \textit{aggregate effort} it aims to implement, and $m_t$ is the \textit{constant payment}. We refer to the contract range $R_t$, aggregate effort $\alpha_t$, and constant payment $m_t$ as the design elements. 

Correspondingly, $s_t\circ \theta(\cdot)$ is given by
\begin{equation}
s_t\circ \theta(\alpha) = \mathbb{1}_{\{\alpha \in c(R_t)\}}\big[K_t'(\alpha_t)\cdot \alpha + m_t\big],\quad \forall \alpha \in [\underline{c},\overline{c}]. 
\label{MBC-eq3}
\end{equation}

We present how $s_t\circ \theta(\cdot)$ is pinned down by those design elements in Figure \ref{MBC-design-el}. $s_t\circ \theta(\cdot)$ is depicted in the blue line. The contract range $R_t$ specifies the part of $s_t\circ\theta(\cdot)$ where it coincides with a straight line, whose slope $K_t'(\alpha_t)$ is determined by the aggregate effort $\alpha_t$, while its intercept is determined by the constant payment $m_t$. We say $K_t'(\alpha_t)$ is the \textit{power} of contract $s_t(\cdot)$, which measures the steepness of the contract within the contract range. 
\begin{figure}[H]
    \centering
    \includegraphics[width=0.5\linewidth]{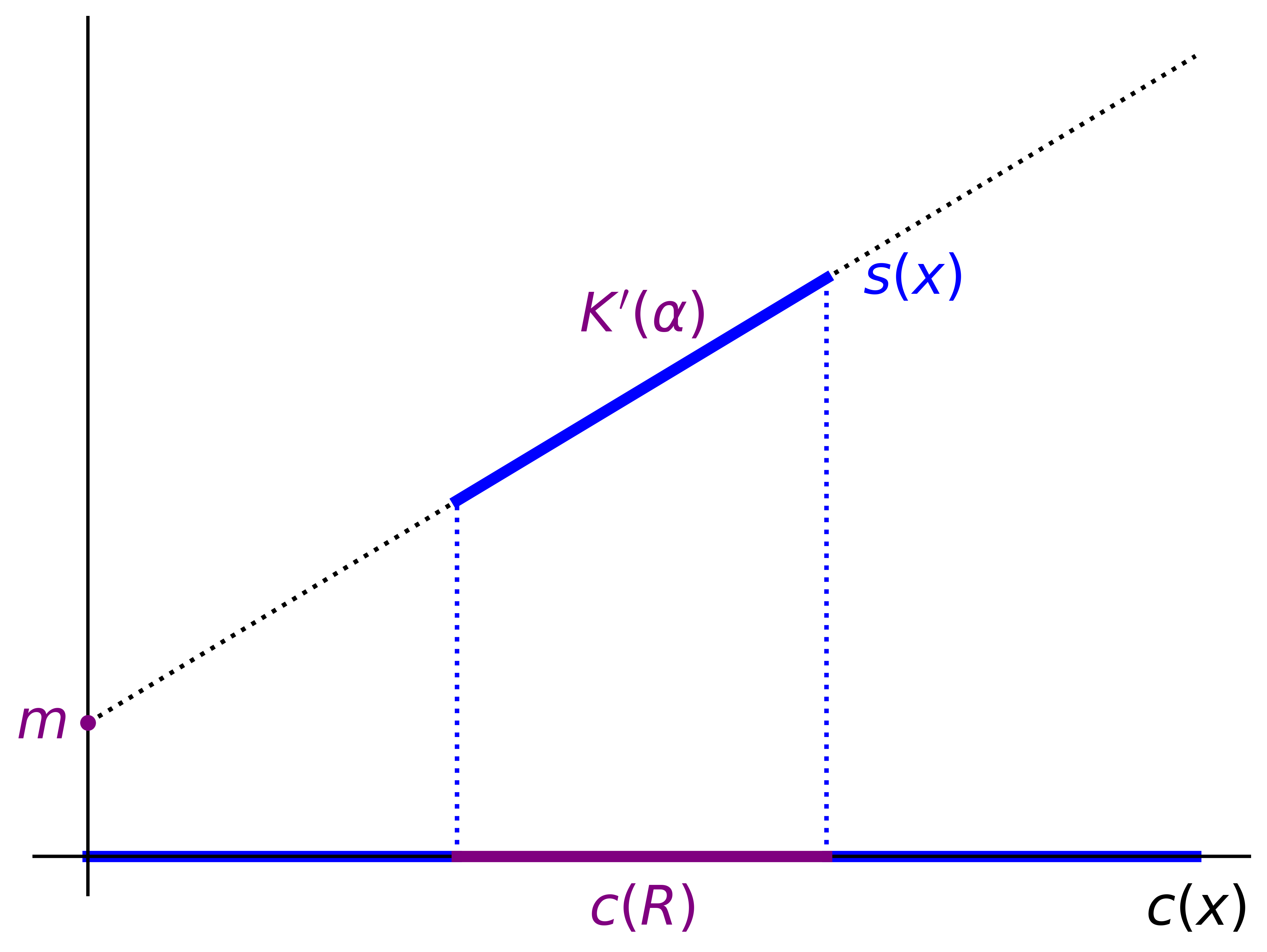}
    \caption{The form of $s_t\circ \theta(\cdot)$ as defined in \eqref{MBC-eq3}.}
    \label{MBC-design-el}
\end{figure}
When the range $R_t$ is expanded, the part of $s_t\circ \theta(\cdot)$ that coincides with the straight line increases. When $R_t$ is the full range, i.e., $R_t = [\underline{x},\overline{x}]$, $s_t\circ \theta(\cdot)$ becomes the straight line itself. We say contract $s_t(\cdot)$, as formulated in \eqref{MBC-eq2}, is \textit{full-range} when its corresponding $R_t$ is equal to $[\underline{x},\overline{x}]$.

Given a contract as formulated in \eqref{MBC-eq2}, with a type-$t$ agent, it can implement any output distribution that induces an aggregate effort of $\alpha_t$ and has a support that is a subset of $R_t$, provided such an output distribution exists and $m_t\geq 0$. Instead of designing a contract directly, we focus on designing $\alpha_t$, $m_t$, $R_t$, and consider the contract pinned down by those elements by \eqref{MBC-eq2}. 

For the rest of this paper, we use $(a, m, R,\mu)$ to denote a menu designed by the principal. $a = (\alpha_0,\alpha_1)\in \mathbb{R}^2$ represents the aggregate effort $\alpha_t$ she intends to implement for the type-$t$ ($t\in \{0,1\}$) agent. $m = (m_0,m_1)\in \mathbb{R}^2$ and $R = (R_0, R_1)\in (2^{[\underline{x},\overline{x}]})^2$ specify the constant payment $m_t$ and the contract range $R_t$ of the type-$t$ contract. $\mu = (\mu_0, \mu_1)\in \Delta([\underline{x},\overline{x}])^2$ refers to the output distribution $\mu_t$ recommended to the type-$t$ agent.\\
\\
\textbf{Feasible Menu: }A menu $(a,m,R,\mu)$ is \textit{feasible} if there exist $s_0,s_1:[\underline{x},\overline{x}]\to \mathbb{R}$ that satisfy the following constraints:
\begin{enumerate}
\item \textit{Moral hazard (MH)}
\begin{equation}
\begin{aligned}
&s_t\circ \theta(\alpha) = \mathbb{1}_{\{\alpha\in c(R_t)\}}\big[K_t'(\alpha_t)\cdot \alpha+m_t\big], \quad \forall \alpha\in [\underline{c},\overline{c}],\quad\forall t\in \{0,1\}.\\
& \int c(x)\mu_t(dx) = \alpha_t, \quad\text{supp}(\mu_t) \subseteq R_t,\quad\forall t\in \{0,1\}.
\end{aligned}
\label{MBC-MH}
\tag{MH}
\end{equation}
The first two constraints in \eqref{MBC-MH} illustrate how our contracts are pinned down by the design elements. The third constraint requires that the corresponding aggregate effort of the distribution $\gamma_t$ ($t\in \{0,1\}$) is equal to $\alpha_t$ and that the support of $\gamma_t$ is a subset of $c(R_t)$. This constraint ensures that the output distribution corresponding to $\gamma_t$ can be implemented by contract $s_t(\cdot)$, given that $m_t\geq 0$. 
\item \textit{Limited liability (LL)}
\begin{equation}
    m_0\geq 0\text{ and }m_1\geq 0.
\label{MBC-LL}
\tag{LL}
\end{equation}
To see why the LL constraints take this form, if we have $m_t\geq 0$, it guarantees that the contract $s_t(\cdot)$ as defined in \eqref{MBC-MH} is non-negative. On the other hand, if $\underline{x} \in R_t$, then limited liability requires that $0 \leq s_t(\underline{x}) = m_t$. If $\underline{x}\notin R_t$, since we want to ensure $s_t(\cdot)$ to implement an output distribution whose support is a subset of $R_t$, it is necessary that $0 \leq s_t(\underline{x}) \leq m_t$.\footnote{In the current formulation, \eqref{MBC-MH} and \eqref{MBC-LL} need to be combined to ensure that the contract $s_t(\cdot)$ ($t\in \{0,1\}$) pinned down by the design elements $\alpha_t$, $m_t$ and $R_t$ is non-negative and can actually implement an output distribution that induces aggregate effort $\alpha_t$ and has a support that is a subset of $R_t$.}
\item \textit{Incentive compatibility (IC)}
\begin{equation}
 K_0'(\alpha_0)\cdot \alpha_0 - K_0(\alpha_0) + m_0\geq \max_{\alpha\in [\underline{c},\overline{c}]}\big[\widehat{s_{1}\circ \theta}(\alpha) - K_0(\alpha) \big],
 \tag{IC0}
 \label{MBC-IC0}
\end{equation}
and
\begin{equation}
 K_1'(\alpha_1)\cdot \alpha_1 - K_1(\alpha_1)+m_1\geq \max_{\alpha\in [\underline{c},\overline{c}]}\big[\widehat{s_0\circ \theta}(\alpha) - K_1(\alpha) \big].
 \tag{IC1}
 \label{MBC-IC1}
\end{equation}
On the LHS of each IC constraint is the agent's payoff from truthfully reporting his type $t$ and taking the recommended action $\mu_t$ whose aggregate effort is equal to $\alpha_t$. On the RHS is the agent's maximal payoff from misreporting his type and obtaining the contract for the other type, whose form is given by Lemma \ref{PRE-lm2}.  
\end{enumerate}

Define the set of feasible menus by
\begin{equation}
\begin{aligned}
    \mathcal{F}:=  \Big\{ & (a,m,R,\mu): \alpha_0,\alpha_1 \in [\underline{c},\overline{c}],\quad \mu_0,\mu_1 \in \Delta([\underline{x},\overline{x}]),\\
    &R_0,R_1\subseteq [\underline{x},\overline{x}],
    \quad m_0,m_1\in \mathbb{R},\quad \exists s_0,s_1: [\underline{x},\overline{x}]\to \mathbb{R}\text{ such that }\\
    &\eqref{MBC-MH}, \eqref{MBC-LL},\eqref{MBC-IC0},\eqref{MBC-IC1}\text{ are all satisfied}.\Big\}
\end{aligned}
\label{MBC-fea}
\end{equation}
\textbf{Principal's problem: }The principal's objective is to choose a feasible menu that maximizes her expected profit. The principal's problem is given by 
\begin{equation}
\begin{aligned}
  \max_{(a,m,R,\mu)\in \mathcal{F}} \quad&p_0 \big[\int x\mu_0(dx) - K_0'(\alpha_0)\cdot \alpha_0 - m_0\big] \\
  \quad+ &p_1\big[\int x\mu_1(dx) - K_1'(\alpha_1)\cdot \alpha_1 - m_1\big].
\end{aligned}
\tag{P}
\label{MBC-P}
\end{equation}

In our current formulation of the principal's problem, the ranges of contracts $R$ and the output distributions $\mu$ are treated as separated design elements. Note that $\mu$ is only involved in the third constraint from \eqref{MBC-MH}. At an optimal menu $(a^{*},m^{*},R^{*},\mu^{*})$, for $t\in \{0,1\}$, $\mu_t^{*}$ is chosen among all $\mu_t\in \Delta([\underline{x},\overline{x}])$ such that $\int c(x)\mu_t(dx) = \alpha_t^{*}$ and $\text{supp}(\mu_t)\subseteq R_t^{*}$ to maximize $\int x\mu_t(dx)$. We straightforwardly have that $\mu_t^{*}$ is chosen to maximize a linear functional, subject to a moment constraint. It implies that the optimum of \eqref{MBC-P} can be reached by recommending output distributions that are supported on at most two points\footnote{The reason is similar to that in footnote \ref{MBCF-EP}}. 

The existence of an optimal menu to \eqref{MBC-P} is also guaranteed. To optimally solve problem \eqref{MBC-P}, we can restrict ourselves to menus where the effort distributions are supported on at most two singletons and contract ranges are equal to the supports of the output distributions. Then solving \eqref{MBC-P} is equivalent to maximizing a continuous function over a compact set in $\mathbb{R}^d$ with some $d\in \mathbb{N}$, which ensures the existence of an optimal solution. We formalize the idea in Appendix \ref{MBCF-app-model}, Proposition \ref{MBCF-app-pp2}.\\
\\
\textbf{Remark: }Any output distribution can clearly be implemented by a contract whose range coincides with the support of that distribution. One could alternatively formulate \eqref{MBC-P} by setting the contract range to equal the support of the output distribution it is designed to implement. However, in our current formulation, we distinguish between contract ranges and the supports of output distributions because we are also interested in studying the maximal contract ranges, which indicate to what extent the contract ranges can be extended without harming the feasibility of the menu.

The maximal contract ranges are shaped by the trade-off between raising production efficiency and reducing information rents, which is regarded as an important feature of this model. To see that, we first consider how the ranges of the contracts influence the agent's incentives. Take a feasible menu $(a,m,R,\mu)\in \mathcal{F}$. For $t\in \{0,1\}$, define $\underline{r}_t : = \inf c(R_t)$ and $\overline{r}_t := \sup c(R_t)$. 
Then $\widehat{s_t\circ \theta}$ can be written as 
\begin{equation}
\widehat{s_t\circ \theta}(\alpha) = \left\{
\begin{array}{ll}
\big(K_t'(\alpha_t) \cdot \underline{r}_t + m_t\big)\cdot\frac{\alpha}{\underline{r}_t}    &  \text{ if }\underline{c}\leq \alpha < \underline{r}_t, \\
K_t'(\alpha_t)\cdot \alpha + m_t  &\text{ if }\underline{r}_t \leq \alpha \leq \overline{r}_t, \\
\big(K_t'(\alpha_t) \cdot \overline{r}_t + m_t\big)\cdot\frac{\overline{c}-\alpha}{\overline{c}-\overline{r}_t} & \text{ if }\overline{r}_t < \alpha \leq \overline{c}. 
\end{array}
\right.
\label{MBC-conv}
\end{equation}

The concavification $\widehat{s_t\circ \theta}(\cdot)$ ($t\in \{0,1\}$) is piece-wise linear and consists of up to three segments. If $\underline{c},\overline{c}\in c(R_t)$, then $\widehat{s_t\circ \theta}$ is linear on the whole interval $[\underline{c},\overline{c}]$. If the range $R_t$ is reduced while $\alpha_t$ is fixed, the interval $[\underline{r}_t,\overline{r}_t]$ narrows, thereby reducing $\widehat{s_t\circ \theta}(\cdot)$ and diminishing the agent's incentives to misreport his type. As illustrated in Figure \ref{MBC-shrink-mis}, when the contract range of $s_t(\cdot)$ is reduced and the interval for the second line segment shifts from $[\underline{r}_t,\overline{r}_t]$ to $[\underline{r}_t',\overline{r}_t']$, $\widehat{s_t\circ \theta}(\cdot)$ transitions from the solid purple line to the solid blue line, resulting in a point-wise reduction.
\begin{figure}[H]
    \centering
    \includegraphics[width=0.5\linewidth]{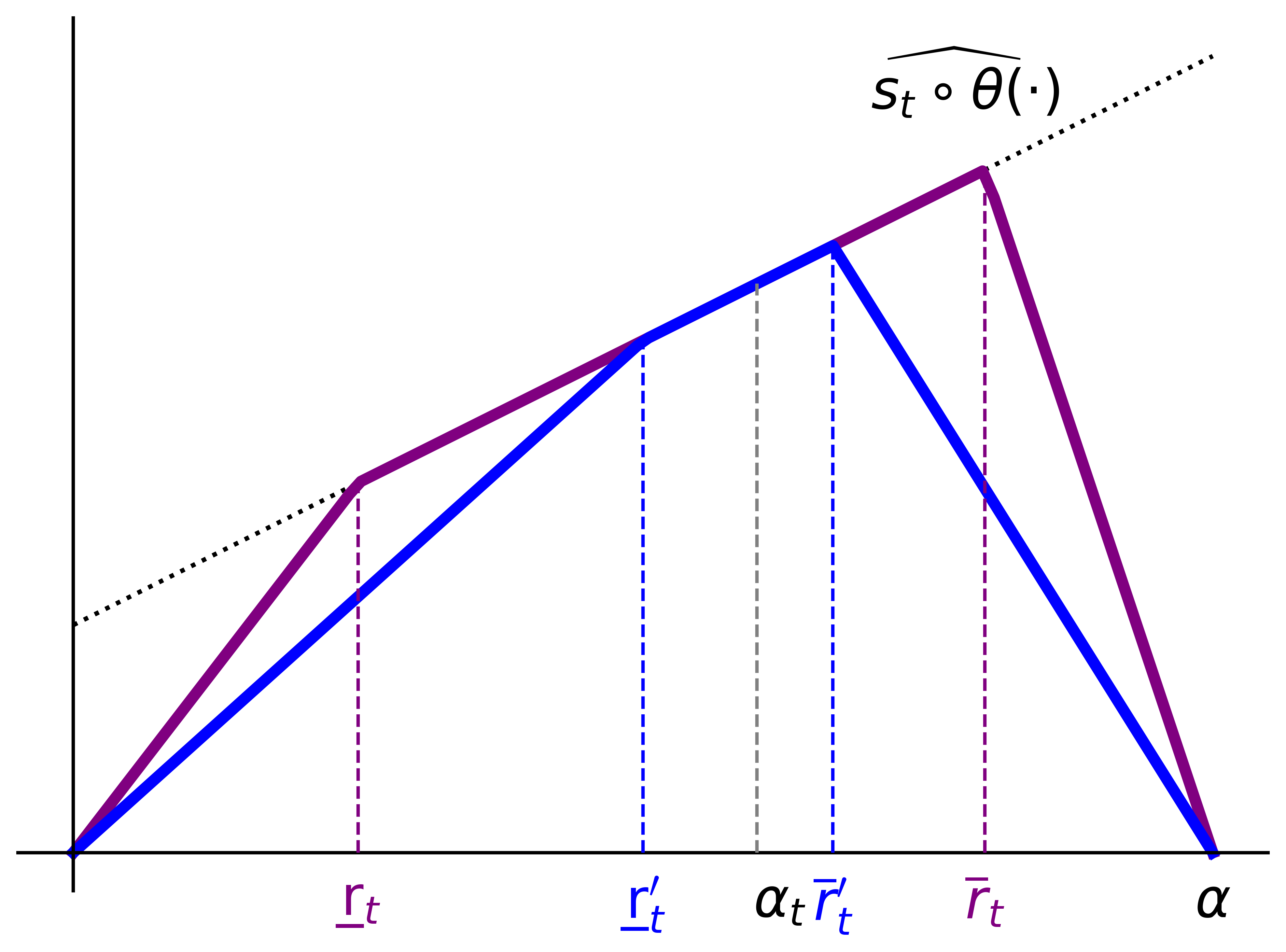}
    \caption{The change in $\widehat{s_t\circ \theta(\cdot)}$ when the contract range is reduced.}
    \label{MBC-shrink-mis}
\end{figure}
While the reduction in the contract range can diminish the agent's incentives to misreport, it also restricts the set of output distributions that the principal can implement, potentially reducing the attainable production value. When faced with a small contract range, the agent is only incentivized to produce outputs within a limited subset, potentially restricting his choice and reducing production efficiency.

 Studying the maximal contract ranges sheds light on the driving force behind the forms of optimal contracts. When the maximal contract range is equal to the full output space, it indicates that full flexibility can be granted to the agent to reach production efficiency. In contrast, if the maximal contract range is truncated, it reflects the principal's need to limit flexibility in order to reduce the screening costs.
\section{Implementation}
\label{MBCF-imp}
In this section, we study the set of feasible menus and characterize the output distributions that can be implemented by the principal in the presence of adverse selection. We begin by introducing two notions of implementability. 
\begin{definition}
A pair of aggregate efforts $a=(\alpha_0,\alpha_1) \in \mathbb{R}^2$ is implementable if there exist $m \in \mathbb{R}^2$, $R\in (2^{[\underline{x},\overline{x}]})^2$ and $\mu \in \Delta([\underline{x},\overline{x}])^2$ such that $(a,m,R,\mu)\in \mathcal{F}$.
\label{IMP-def1}
\end{definition}
\begin{definition}
A pair of output distributions $\mu = (\mu_0,\mu_1)\in \Delta([\underline{x},\overline{x}])^2$ is implementable if there exist $a\in \mathbb{R}^2$, $m \in \mathbb{R}^2$ and $R\in (2^{[\underline{x},\overline{x}]})^2$ such that $(a,m,R,\mu)\in \mathcal{F}$. 
\label{IMP-def2}
\end{definition}
From now on, we impose the following assumption on the cost functions.
\begin{assumption}
$K_0'(\alpha) > K_1'(\alpha)$ for all $\alpha \in (\underline{c},\overline{c}]$. 
\label{OM-as1}
\end{assumption}
Assumption \ref{OM-as1} shows that the type-0 agent has a higher marginal cost with regard to the aggregate effort. Hence, the type-0 agent requires a higher payment to be incentivized to increase effort. Assumption \ref{OM-as1} holds if $K_1(\alpha) = \eta \cdot K_0(\alpha)$ for all $\alpha \in [\underline{c},\overline{c}]$ with $\eta \in (0,1)$. 

Under Assumption \ref{OM-as1}, we have a simple characterization of the implementable aggregate efforts, stated in the following proposition. 
\begin{proposition}
Under Assumption \ref{OM-as1}, a pair of aggregate efforts $(\alpha_0,\alpha_1) \in [\underline{c},\overline{c}]^2$ is implementable if and only if
\begin{equation*}
    \alpha_0 \leq \alpha_1.
\end{equation*}
\label{IMP-pp1}
\end{proposition}
The necessity of $\alpha_0 \leq \alpha_1$ in Proposition \ref{IMP-pp1} follows directly from the on-path IC constraints. If a menu $(a,m,R,\mu)$ is feasible, it must satisfy these constraints, which prevent each type of agent from misreporting his type and adopting the action recommended to the other type: 
\begin{equation*}
\begin{aligned}
    K_0'(\alpha_0)\cdot\alpha_0 - K_0(\alpha_0) + m_0 &\geq K_1'(\alpha_1)\cdot \alpha_1 - K_0(\alpha_1) + m_1,\\
    K_1'(\alpha_1)\cdot \alpha_1 - K_1(\alpha_1) + m_1 &\geq K_0'(\alpha_0)\cdot \alpha_0 - K_1(\alpha_0) + m_0.
\end{aligned}
\end{equation*}

By combining the on-path IC constraints, we obtain that $K_0(\alpha_1) - K_0(\alpha_0)\geq K_1(\alpha_1) - K_1(\alpha_0)$. Under Assumption \ref{OM-as1}, this inequality directly implies that $\alpha_0 \leq \alpha_1$. Conversely, the sufficiency of $\alpha_0\leq \alpha_1$ can be established by constructing a feasible menu for each pair of $(\alpha_0,\alpha_1)$ satisfying this condition. In particular, one can construct a menu that minimizes the expected payment from the principal while implementing the desired aggregate efforts. The construction of these payment-minimizing menus is provided in Lemma \ref{IMP-lm1}.
\begin{lemma}
Take any pair of aggregate efforts $a=(\alpha_0,\alpha_1)\in [\underline{c},\overline{c}]^2$ such that $\alpha_0 \leq \alpha_1$. Take $R':=(\{\theta(\alpha_0)\},\{\theta(\alpha_1)\})$ and $\mu': = (\delta_{\theta(\alpha_0)},\delta_{\theta(
\alpha_1)})$. If $K_0'(\alpha_0)\leq K_1'(\alpha_1)$, define
\begin{equation*}
\begin{aligned}
    m_0' &:= \big[K_1'(\alpha_1)\cdot k_0^{-1}\big(K_1'(\alpha_1)\big) - K_0\big(k_0^{-1}\big(K_1'(\alpha_1)\big)\big)\big]-\big[K_0'(\alpha_0)\cdot \alpha_0 - K_0(\alpha_0)\big],\\
    m_1' &:= 0.
\end{aligned}  
\end{equation*}
If $K_0'(\alpha_0)> K_1'(\alpha_1)$, define 
\begin{equation*}
\begin{aligned}
    m_0':&= 0,\\
    m_1':&= \max\Big\{0, \big[K_0'(\alpha_0)\cdot \alpha_0 - K_1(\alpha_0)\big] - \big[K_1'(\alpha_1)\cdot \alpha_1 - K_1(\alpha_1)\big]\Big\}.
\end{aligned}
\end{equation*}
Take $m':=(m_0',m_1')$. Under Assumption \ref{OM-as1}, the menu $(a,m',R',\mu')$ minimizes the principal's expected payment among all feasible menus that implement $(\alpha_0,\alpha_1)$. 
\label{IMP-lm1}
\end{lemma}
\begin{proof}
See Appendix \ref{MBCF-app-imp}.   
\end{proof}
Lemma \ref{IMP-lm1} demonstrates that, to implement $(\alpha_0,\alpha_1)$ with minimum contract payments, the principal can restrict attention to menus in which the contract ranges are reduced to singletons. This is because narrowing the contract ranges weakens the agents' incentives to misreport and thus reduces the screening costs. In these payment-minimizing menus, the principal provides a higher constant payment to the agent with stronger incentives to misreport. If $K_0'(\alpha_0) \leq K_1'(\alpha_1)$, then the type-0 agent has greater incentives to misreport, and his constant payment is determined by the binding \eqref{MBC-IC0} constraint, while the type-1 agent obtains zero constant payment. Conversely, if $K_0'(\alpha_0) >K_1'(\alpha_1)$, then the type-1 agent has greater incentives to misreport, and his constant payment is pinned down by either the binding \eqref{MBC-IC1} or \eqref{MBC-LL} constraint. 

Proposition \ref{IMP-pp1} characterizes the set of implementable aggregate efforts. However, the principal's profit not only depends on the total amount of aggregate effort, but also on how the aggregate effort is distributed across different output levels. To enhance production efficiency, the principal may wish to implement distributions whose supports extend beyond singletons. Consequently, the optimal menus need not be restricted to those with singleton contract ranges. To better understand the principal's problem, it is therefore necessary to examine the implementation of output distributions. 

Consider any pair of output distributions $(\mu_0,\mu_1) \in \Delta([\underline{x},\overline{x}])^2$ with corresponding aggregate efforts $(\alpha_0,\alpha_1)$. Proposition \ref{IMP-pp1} implies that a necessary condition for $(\mu_0,\mu_1)$ to be implementable is $\alpha_0 \leq \alpha_1$. We show that if $(\alpha_0, \alpha_1)$ further satisfies $K_0'(\alpha_0)\leq K_1'(\alpha_1)$, then $(\mu_0,\mu_1)$ is always implementable. In contrast, if $K_0'(\alpha_0) > K_1'(\alpha_1)$, the implementability of $(\mu_0,\mu_1)$ additionally depends on the values of $\overline{\alpha}_0 := c\big(\max \text{supp}(\mu_0)\big)$, the effort corresponding to the highest output in the support of $\mu_0$, and $\underline{\alpha}_1 := c\big(\min \text{supp}(\mu_1)\big)$, the effort corresponding to the lowest output in the support of $\mu_1$. In this case, there exists an increasing function $\check{\alpha}_0:[\underline{c},\overline{c}]\to [\underline{c},\overline{c}]$ such that $(\mu_0,\mu_1)$ is implementable if and only if $\overline{\alpha}_0$ does not exceed $\check{\alpha}_0(\underline{\alpha}_1)$\footnote{The shape of $\check{\alpha}_0(\cdot)$ depends on the targeted $(\alpha_0,\alpha_1)$. We omit $(\alpha_0,\alpha_1)$ from the notation of $\check{\alpha}_0(\cdot)$ for simplicity.}. The exact expression of $\check{\alpha}_0(\cdot)$ is provided in Section \ref{APP-IMP-OD}, Appendix \ref{MBCF-app-imp}.
\begin{theorem}
Under Assumption \ref{OM-as1}, a pair of output distributions $(\mu_0,\mu_1)\in \Delta([\underline{x},\overline{x}])^2$ with corresponding aggregate efforts $(\alpha_0,\alpha_1)$ is implementable if and only if one of the following two conditions is satisfied
\begin{enumerate}[label=\circled{\arabic*}]
    \item $\alpha_0\leq \alpha_1$ and $K_0'(\alpha_0)\leq K_1'(\alpha_1)$.
    \item $\alpha_0\leq \alpha_1$, $K_0'(\alpha_0)> K_1'(\alpha_1)$ and $\overline{\alpha}_0 \leq \check{\alpha}_0(\underline{\alpha}_1)$. 
\end{enumerate}
\label{IMP-pp2}
\end{theorem}
\begin{proof}
See Appendix \ref{APP-IMP-OD}. 
\end{proof}
Consider the case where $(\mu_0,\mu_1)$ satisfies that $\alpha_0 \leq \alpha_1$ and $K_0'(\alpha_0) > K_1'(\alpha_1)$. When $(\mu_0,\mu_1)$ corresponds to a large $\overline{\alpha}_0$ and a small $\underline{\alpha}_1$, the supports of $\mu_0$ and $\mu_1$ overlap substantially. This makes incentive compatibility more difficult to achieve since it pushes up the payoffs from off-path deviations (the right-hand side of \eqref{MBC-IC0} and \eqref{MBC-IC1}). The fact that only one side of the support boundaries matters here follows from Assumption 1, a point we will elaborate in Proposition \ref{OM-pp2} in the next section. For each possible value of $\underline{\alpha}_1$, $\check{\alpha}_0(\underline{\alpha}_1)$ provides an upper bound on the implementable $\overline{\alpha}_0$. A larger $\underline{\alpha}_1$ implies a more restricted support of $\mu_1$, which allows the support of $\mu_0$ to be more spread out, leading to a larger $\check{\alpha}_0(\underline{\alpha}_1)$. 

We illustrate the upper bound $\check{\alpha}_0(\cdot)$ in Figure \ref{IMP-fig1} as the green line. The function initially takes a constant that is smaller than $\underline{c}$, then becomes convex, and eventually flattens out at $\overline{c}$. Theorem \ref{IMP-pp2} shows that that $(\mu_0,\mu_1)$ is implementable as long as $(\underline{\alpha}_1,\overline{\alpha}_0)$ lies in the region below $\check{\alpha}_0(\cdot)$. Corollary \ref{IMP-co1} establishes that the region above $\check{\alpha}_0(\cdot)$ is always non-empty. 
\begin{figure}[H]
    \centering
    \includegraphics[width=0.6\linewidth]{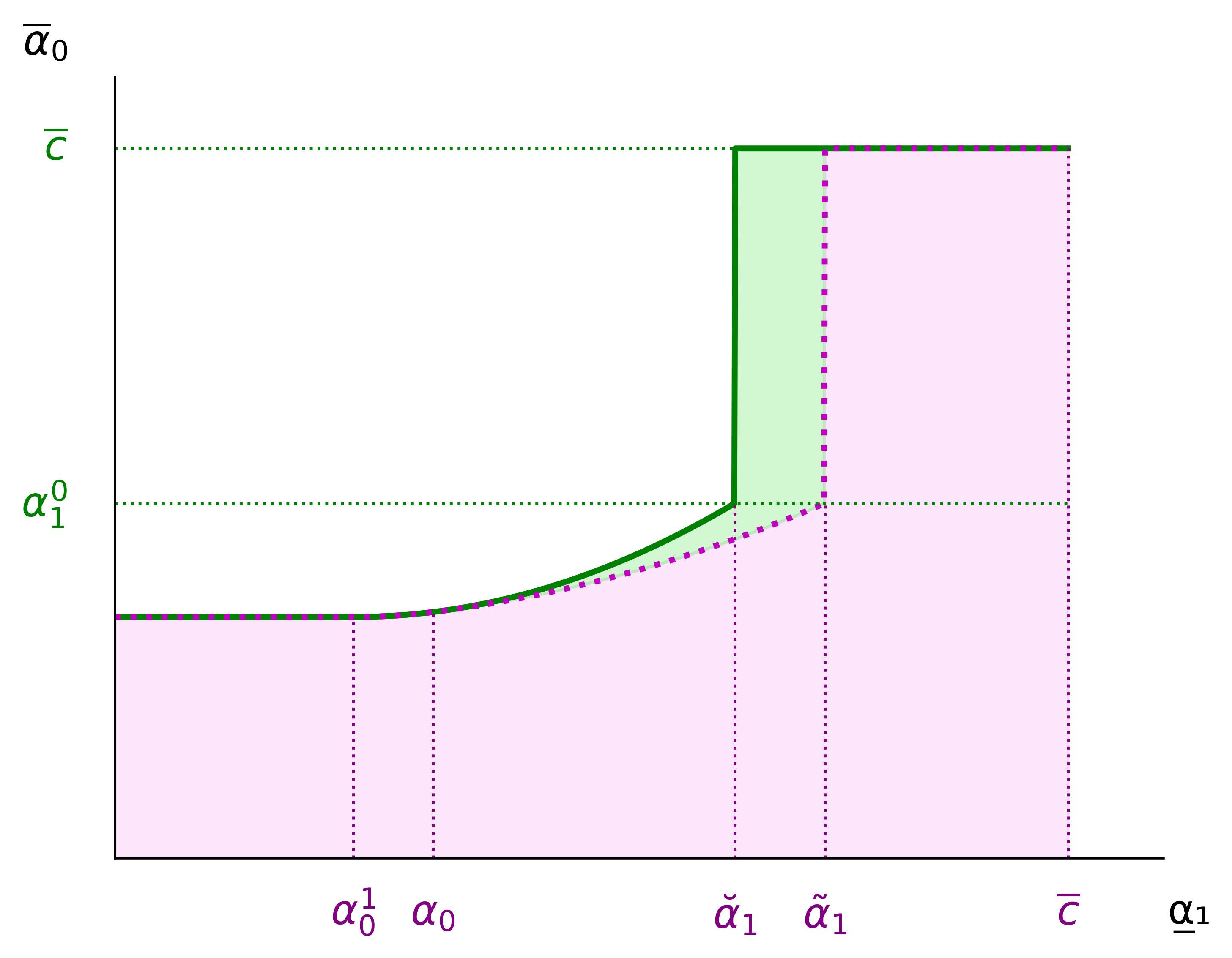}
    \caption{The implementable region of $(\underline{\alpha}_1,\overline{\alpha}_0)$ given $(\alpha_0,\alpha_1)$ satisfying $\alpha_0\leq \alpha_1$ and $K_0'(\alpha_0)>K_1'(\alpha_1)$.}
    \label{IMP-fig1}
\end{figure}

\begin{corollary}
Under Assumption \ref{OM-as1}, a pair of output distributions $(\mu_0,\mu_1)$ with corresponding aggregate efforts $(\alpha_0,\alpha_1)$ is not implementable if $K_0'(\alpha_0)>K_1'(\alpha_1)$, $\overline{\alpha}_0 = \overline{c}$, and $\underline{\alpha}_1 = \underline{c}$. 
\label{IMP-co1}
\end{corollary}

Theorem \ref{IMP-pp2} implies that if a pair of aggregate efforts satisfies condition \circled{1}, then any pair of output distributions inducing $(\alpha_0,\alpha_1)$ is implementable. If $(\alpha_0,\alpha_1)$ satisfies the first two inequalities in condition \circled{2}, the implementability of the corresponding output distributions also depends on the extent to which their supports overlap. Under Assumption \ref{OM-as1}, an increase in the contract power benefits the type-1 agent more than the type-0 agent, as it leads to a larger increase in the effort exerted. Consequently, it is more difficult to implement the case in which the type-1 agent obtains a lower-powered contract than the case in which he receives a higher-powered contract. In the former situation, it becomes necessary to restrict the supports of the output distributions to ensure implementability. 

We also examine the payment-minimizing menu that implements a given pair of output distributions $(\mu_0,\mu_1)$. As in Lemma \ref{IMP-lm1}, we restrict attention to menus in which the contract ranges coincide with the supports of the respective output distributions. When $(\mu_0,\mu_1)$ satisfies condition \circled{1}, the payment-minimizing menu features a zero constant payment to the type-1 agent, while the constant payment for the type-0 agent is determined by the binding \eqref{MBC-IC0} constraint. 

In contrast, when $(\mu_0,\mu_1)$ satisfies condition \circled{2}, the structure of the payment-minimizing constant payments depends on $\overline{\alpha}_0$ and $\underline{\alpha}_1$. In some cases, similar to Lemma \ref{IMP-lm1}, the type-0 agent receives zero constant payment, while the type-1 agent's constant payment is determined by either the binding \eqref{MBC-IC1} or \eqref{MBC-LL} constraint. In other cases, however, both types of agents receive strictly positive constant payments. In Figure \ref{IMP-fig1}, the former case corresponds to the purple region, and the latter to the green region. The boundary separating these two regions (the dotted purple line in Figure \ref{IMP-fig1}) is denoted by the function $\tilde{\alpha}_0: [\underline{c},\overline{c}] \to  [\underline{c},\overline{c}]$. The exact forms of the payment-minimizing menus and $\tilde{\alpha}_0(\cdot)$ are given in Appendix \ref{APP-IMP-RCP}.

The boundary $\tilde{\alpha}_0(\cdot)$ coincides with $\check{\alpha}_0(\cdot)$ except over an interval in the middle. When the pair $(\underline{\alpha}_1,\overline{\alpha}_0)$ associated with $(\mu_0,\mu_1)$ falls in the region between $\tilde{\alpha}_0(\cdot)$ and $\check{\alpha}_0(\cdot)$, it is necessary to provide strictly positive constant payments to both types of agents to implement $(\mu_0,\mu_1)$ (see Corollary \ref{APP-IMP-pp3}, Appendix \ref{MBCF-app-imp}). In this region, offering the type-0 agent a zero constant payment causes a violation of \eqref{MBC-IC0}. However, increasing the constant payments uniformly to both types relaxes the \eqref{MBC-IC0} constraint, since its right-hand side increases less than its left-hand side. 
\section{Optimal Menus}
\label{MBCF-om}
Before analyzing the optimal menus for problem \eqref{MBC-P}, we impose an additional assumption on $K_0(\cdot)$ and $K_1(\cdot)$.
\begin{assumption} $\Theta(\alpha)- K_t'(\alpha)\cdot \alpha$ is strictly single-peaked in $\alpha \in [\underline{c},\overline{c}]$ for $t\in \{0,1\}$. \label{OM-as3} 
\end{assumption}
In Assumption \ref{OM-as3}, the expression $\Theta(\alpha) - K_t'(\alpha)\cdot \alpha$ ($t\in \{0,1\}$) represents the maximal profit that the principal can obtain from a type-$t$ agent exerting an aggregate effort of $\alpha$ when the type is observable (Lemma \ref{PRE-lm1}). Assumption \ref{OM-as3} requires that there exists a unique optimal level of aggregate effort that solves the principal's pure moral hazard problem, with the property that the further the induced effort deviates from this level, the lower the principal's profit. This requirement is fulfilled, for example, if $K_t'(\alpha)\cdot \alpha$ is strictly convex in $\alpha$. Within Assumption \ref{OM-as3}, the pure moral hazard solution $(\alpha_0^{MH},\alpha_1^{MH})$ is always well-defined.

Under Assumption \ref{OM-as1} and \ref{OM-as3}, we show that in any optimal menu, the contract power of the type-0 contract is never smaller than that of the type-1 contract.  
\begin{proposition}
\label{OM-pp1}
Take an optimal menu $(a^{*},m^{*},R^{*},\mu^{*})\in \mathcal{F}$ that solves \eqref{MBC-P}. Under Assumption \ref{OM-as1} and \ref{OM-as3}, it holds that
\begin{equation*}
    K_0'(\alpha_0^{*}) \geq K_1'(\alpha_1^{*}).
\end{equation*}
\end{proposition}
\begin{proof} 
See Appendix \ref{MBCF-app-om}.
\end{proof}
The intuition behind Proposition \ref{OM-pp1} is as follows. Suppose instead that $K_0'(\alpha_0^{*}) < K_1'(\alpha_1^{*})$ at the optimal menu $(a^{*},m^{*},R^{*},\mu^{*})$. The principal's reluctance to raise the type-0 agent's aggregate effort beyond $\alpha_0^{*}$ indicates that $\alpha_0^{*}$ is already at a high level. In fact, one can show that in this case $\alpha_0^{*}$ is at the first-best level, which in turn implies that $\alpha_1^{*}$ lies above the first-best level. Hence, the principal can always benefit from slightly reducing the aggregate effort exerted by the type-1 agent, which is achieved by lowering the power of the type-1 contract.

Proposition \ref{OM-pp1} sheds light on the contract powers at the optimal menus. As a next step, we now turn to the contract ranges at the optimal menus. We examine to what extent the contract ranges at an optimal menu can be expanded without harming the feasibility of the menu. As discussed in Section \ref{MBCF-model}, these maximal ranges can capture the principal's trade-off between increasing production value and diminishing the agent's incentives to misreport. We formally define the maximal contract ranges as follows.

\begin{definition}
Given a feasible menu $(a,m,R,\mu)\in \mathcal{F}$, the range $R_t$ ($t\in \{0,1\}$) is \textit{maximal} if for any $E \subseteq [\underline{x},\overline{x}]$ satisfying $R_t\subsetneq E$, it holds that
\begin{equation*}
    K_{t'}'(\alpha_{t'})\cdot \alpha_{t'} -K_{t'}(\alpha_{t'}) + m_{t'} < \max_{\alpha\in [\underline{c},\overline{c}]}\big[\widehat{s_t^{E}\circ\theta}(\alpha) - K_t(\alpha)\big],
\end{equation*}
where $t'\in \{0,1\}$, $t'\neq t$, and $s_t^{E}\circ\theta(\cdot)$ is given by $s_t^{E}\circ \theta(\alpha) = \mathbb{1}_{\{\alpha\in c(E)\}}\big[K_t'(\alpha_t)\cdot \alpha+m_t\big]$ for all $\alpha\in [\underline{c},\overline{c}]$.
\label{OM-df1}
\end{definition}
According to Definition \ref{OM-df1}, a contract range is maximal at a given menu if any further expansion of the range would violate the IC constraint. When $a$ and $m$ are fixed, whether a range is maximal is independent of the contract range assigned to the other type. Therefore, when expanding the contract ranges to their maximal levels, the order of expansion is irrelevant.  

We first show that the type-0 contract range can always be extended to the leftmost point, yielding a right-truncated interval, while the type-1 contract range can always be extended to the rightmost point, yielding a left-truncated interval, as stated in Proposition \ref{OM-pp2}. Note that this result holds for any feasible menu, not only the optimal ones. 
\begin{proposition}
Take any feasible menu $(a, m, R, \mu)\in \mathcal{F}$ and take $\overline{x}_0:= \sup R_0$ and $\underline{x}_1:= \inf R_1$. Define
\begin{equation*}
{R}_0' := [\underline{x},\overline{x}_0] \text{ and }
R'_1 := [\underline{x}_1,\overline{x}].
\end{equation*}
Under Assumption \ref{OM-as1}, $(a,m,R',\mu)$ with $R' := (R_0',R_1')$ is also feasible.
\label{OM-pp2}
\end{proposition}
\begin{proof}
See Appendix \ref{MBCF-app-om}.
\end{proof} 
Proposition \ref{OM-pp2} builds on the idea that, for any feasible menu $(a,m,R,\mu)$, one can extend $R_0$ leftward and $R_1$ rightward without affecting feasibility. Suppose we extend the range $R_1$ of the type-1 contract $s_1(\cdot)$ rightward to include higher outputs. Because such high outputs are prohibitively costly for the type-0 agent, the extension does not make the type-1 contract more attractive to the type-0 agent. Similarly, if we extend the range $R_0$ of the type-0 contract $s_0(\cdot)$ leftward, the type-1 agent, who tends to produce higher outputs, would not find the extended type-0 contract more appealing. This idea also explains why, in implementing output distributions, only one side of the support boundaries is relevant (Theorem \ref{IMP-pp2}).

Proposition \ref{OM-pp2} states that, for any feasible menu, it does no harm to provide the type-0 (type-1) agent with maximal incentives to produce low-cost, low-return (high-cost, high-return) outputs. Applied to the optimal menu, this means the principal can attain the optimal expected profit by offering the type-0 agent a contract whose range includes all output levels except some higher ones, and the type-1 agent a contract that excludes certain lower levels (where the exclusions may not be strict).

Proposition \ref{OM-pp2} shows that, for an optimal menu, each contract range can be extended in one direction without affecting feasibility. We now ask whether the ranges of the optimal contracts can be extended in both directions while still preserving feasibility. 

Theorem \ref{OM-pp3} demonstrates that, given an optimal menu, the ranges of both contracts can be extended to full without affecting feasibility if and only if the optimal contracts have the same power.
\begin{theorem}
Take an optimal menu $(a^{*},m^{*},R^{*},\mu^{*})\in \mathcal{F}$ that solves \eqref{MBC-P}. Define $R_F := ([\underline{x},\overline{x}],[\underline{x},\overline{x}])$. Under Assumption \ref{OM-as1} and \ref{OM-as3}, the menu $(\alpha^{*},m^{*},R_F,\mu^{*})$ is feasible if and only if 
$$K_0'(\alpha_0^{*}) = K_1'(\alpha_1^{*}). $$
\label{OM-pp3}
\end{theorem}

The sufficiency part of Theorem \ref{OM-pp3} is straightforward, since the IC constraints are automatically satisfied when only a single contract is offered. To establish the necessity part,  we demonstrate that implementing $(\mu_0^{*},\mu_1^{*})$ with two full-range contracts is impossible whenever $K_0'(\alpha_0^{*}) > K_1'(\alpha_1^{*})$. The reasoning parallels that of Corollary \ref{IMP-co1}: to implement $(\mu_0^{*},\mu_1^{*})$ in this case, the principal must offer the type-1 agent a strictly higher-powered contract. Because the type-1 agent benefits more from such a contract than the type-0 agent, implementation without restricting the contract ranges is infeasible.

Theorem \ref{OM-pp3} shows that, for an optimal menu $(a^{*},m^{*},R^{*},\mu^{*})$, if $K_0'(\alpha_0^{*}) = K_1'(\alpha_1^{*})$, then both contract ranges can be extended to full without affecting feasibility or optimality. In this case, the extended contracts coincide, and the principal's maximum profit can be achieved by a menu that only contains a single full-range contract. By contrast, if $K_0'(\alpha_0^{*}) > K_1'(\alpha_1^{*})$, then extending both contracts to full ranges is infeasible. One contract range must remain truncated: either the type-0 contract has a strictly right-truncated maximal range, or the type-1 contract has a strictly left-truncated maximal range. Any inclusion of outputs outside the maximal range renders the menu infeasible.

In our model, the principal has two instruments for designing contracts. The first is the power of a contract, which measures the contract's steepness and determines the agent's aggregate effort level. The second is the range of a contract, which specifies the set of outputs the agent is willing to produce and restricts how the aggregate effort is allocated. Theorem \ref{OM-pp3} shows that, in any optimal menu, if the principal uses the first instrument to create differences in contract powers, then the second instrument must also be employed to assist screening.
\section{Convex Effort Functions}
\label{MBCF-cef}
\subsection{Optimal Menus under Convex Effort Functions}
In this section, we focus on exploring the forms of optimal menus for \eqref{MBC-P} when the effort function $c(\cdot)$ is convex. Because $c(\cdot)$ is convex, its inverse $\theta(\cdot) = c^{-1}(\cdot)$ is therefore concave, which implies that $\Theta(\alpha) = \theta(\alpha)$ for all $\alpha \in [\underline{c},\overline{c}]$. In the baseline problem with a convex effort function, it is always optimal for the principal to recommend a point mass distribution to the agent. We show that this result continues to hold in the presence of adverse selection and formally present it in Lemma \ref{EF-lm1}. The analysis is straightforward: narrowing the ranges of the contracts does not impact the profit obtained by the principal and only reduces the agent's incentives to misreport.
\begin{lemma}
Assume that $c(\cdot)$ is convex. Take an optimal menu $(a^{*},m^{*},R^{*},\mu^{*})\in \mathcal{F}$ that solves \eqref{MBC-P}. It holds that $(a^{*},m^{*},R',\mu')$ also optimally solves \eqref{MBC-P}, where 
\begin{equation*}
R' := (\big\{\theta(\alpha_0^{*})\big\}, \big\{\theta(\alpha_1^{*})\big\})\text{ and } \mu': = (\delta_{\theta(\alpha_0^{*})},\delta_{\theta(\alpha_1^{*})}).
\end{equation*}
\label{EF-lm1} 
\end{lemma}
\begin{proof}
See Appendix \ref{MBCF-app-cef}.
\end{proof}

Lemma \ref{EF-lm1} suggests that, to identify optimal menus, it suffices to restrict attention to cases where both the contract ranges and the supports of the recommended effort distributions are singletons. This implies that it is optimal to adopt payment-minimizing menus as described in Lemma \ref{IMP-lm1}. By Proposition \ref{OM-pp1}, the principal can further restrict attention to menus in which the type-0 contract has greater power than the type-1 contract. Consequently, the principal's problem \eqref{MBC-P} can be reduced to the simpler problem \eqref{EF-eq1}, optimized over $\alpha_0$, $\alpha_1$, and $m_1$, as formalized in the following lemma.
\begin{lemma}
Assume that $c(\cdot)$ is convex. Define problem \eqref{EF-eq1} as below: 
\begin{subequations}
\begin{align}
\max_{\alpha_0,\alpha_1\in [\underline{c},\overline{c}], m_1\in \mathbb{R}} \quad & p_0\big[\Theta(\alpha_0) - K_0'(\alpha_0)\cdot \alpha_0\big] + p_1\big[\Theta(\alpha_1) - K_1'(\alpha_1)\cdot \alpha_1 - m_1\big]  \tag{$P_1$}\label{EF-eq1}\\
s.t. \quad & \alpha_0 \leq \alpha_1 \\
& K_0'(\alpha_0)\geq K_1'(\alpha_1)\\
& m_1 = \max\Big\{0,\big[K_0'(\alpha_0)\cdot \alpha_0 - K_1(\alpha_0)\big] \atop \qquad\qquad- \big[K_1'(\alpha_1)\cdot \alpha_1 - K_1(\alpha_1)\big]\Big\}.\label{EF-eq1c} 
\end{align}
\end{subequations}
Under Assumption \ref{OM-as1}, problem \eqref{EF-eq1} has the same optimal value as problem \eqref{MBC-P}. Take an optimal menu $(a^{*},m^{*},R^{*},\mu^{*})\in \mathcal{F}$ to \eqref{MBC-P}, it holds that $m_0^{*} = 0$, and that $(\alpha_0^{*},\alpha_1^{*},m_1^{*})$ solves \eqref{EF-eq1}. 
\label{EF-lm2}
\end{lemma}
\begin{proof}
See Appendix \ref{MBCF-app-cef}.
\end{proof}

Based on Lemma \ref{EF-lm2}, we establish a sufficient and necessary condition on when the optimal menu consists of contracts with the same level of power. 
\begin{theorem}
Assume that $c(\cdot)$ is convex. Take an optimal menu $(a^{*},m^{*},R^{*},\mu^{*})\in \mathcal{F}$ that solves \eqref{MBC-P}. Under Assumption \ref{OM-as1} and \ref{OM-as3}, it holds that
\begin{equation*}
    K_0'(\alpha_0^{*}) = K_1'(\alpha_1^{*})
\end{equation*}
if and only if $K_0'(\alpha_0^{MH})\leq K_1'(\alpha_1^{MH})$. 
\label{EF-th1}
\end{theorem}
\begin{proof}
See Appendix \ref{MBCF-app-cef}.
\end{proof}

Theorem \ref{EF-th1} shows that, with a convex effort function, whether $K_0'(\alpha_0^{*}) = K_1'(\alpha_1^{*})$ holds at the optimal menu $(a^{*},m^{*},R^{*},\mu^{*})$ depends entirely on the pure moral hazard solution $(\alpha_0^{MH},\alpha_1^{MH})$, which is well-defined under Assumption \ref{OM-as3}. The optimal contracts exhibit the same level of power if and only if, in the pure moral hazard problem, the principal would like to assign the type-0 agent with a weakly lower-powered contract, compared to the one assigned to the type-1 agent. Theorem \ref{EF-th1} depends on the fact that the principal's profit in the baseline problem is single-peaked in the aggregate effort. If the implemented aggregate effort levels do not align with the pure moral hazard solution, we can adjust the menu to strictly increase the profit obtained. This result is intuitive but rests on the simplicity of the principal's problem when the effort function is convex. 

In the case where $K_0'(\alpha_0^{MH}) < K_1'(\alpha_1^{MH})$, to implement $(\alpha_0^{MH}, \alpha_1^{MH})$ in the problem with adverse selection, the principal must pay a positive constant payment to the type-0 agent to ensure incentive compatibility. Because this transfer is excessively costly, the principal instead prefers to offer contracts with equal power, thereby eliminating the screening cost. By contrast, when $K_0'(\alpha_0^{MH}) > K_1'(\alpha_1^{MH})$, the implemented aggregate efforts might be distorted from this pure moral hazard solution in the presence of adverse selection, but the screening feature is maintained. 

From Theorem \ref{EF-th1}, when it holds that $K_0'(\alpha_0^{MH})\leq K_1'(\alpha_1^{MH})$, all optimal menus feature contracts with equal power. In this case, the contract ranges can be extended to full, resulting in a single contract through which the principal can attain her maximal profit. By contrast, if $K_0'(\alpha_0^{MH}) >  K_1'(\alpha_1^{MH})$, the optimal contracts have distinct powers at any optimal menu, and the maximal profit cannot be attained through a single full-range contract.

Theorem \ref{EF-th1} demonstrates that, under a convex effort function, it may be optimal for the principal to screen the agent by offering contracts with different powers. We provide an example in which such screening is optimal and use it to further illustrate why optimality of screening may occur in our setting.
\begin{example}
\label{EF-ex1}
Consider the problem with a linear effort function $c(x) = x$ for all $x\in [\underline{x},\overline{x}] = [0,\frac{1}{2}]$. Then we have $[\underline{c},\overline{c}] = [0,\frac{1}{2}]$. The cost functions are given by $K_0(\alpha) = \alpha^2$ and $K_1(\alpha) = \frac{2}{3}\alpha^3$ for all $\alpha \in [\underline{c},\overline{c}]$. We can verify that Assumption \ref{OM-as1} and \ref{OM-as3} are both satisfied. Furthermore, it holds that $K_0'(\alpha_0^{MH}) > K_1'(\alpha_1^{MH})$. By Theorem \ref{EF-th1}, the optimal menu to \eqref{MBC-P} in this setting features contracts with different powers.

We take the optimal aggregate efforts $(\alpha_0^{*},\alpha_1^{*})$ from the optimal menu to \eqref{MBC-P}. Let $(\alpha_0^{S},\alpha_1^{S})$ be the optimal aggregate efforts when the principal focuses on designing a single full-range contract. We have that $\alpha_0^{S} < \alpha_0^{*} <\alpha_0^{MH}$ and $\alpha_1^{S} > \alpha_1^{*} >\alpha_1^{MH}$. 

In Figure \ref{EF-binary-eg}, we plot $\Theta(\alpha_0) - K_0'(\alpha_0)\cdot \alpha_0$, $\Theta(\alpha_1) - K_1'(\alpha_1)\cdot \alpha_1$, and $-m_1$ where $m_1 = \max\Big\{0, \big[K_0'(\alpha_0)\cdot \alpha_0 - K_1(\alpha_0)\big] - \big[K_1'(\alpha_1)\cdot \alpha_1 - K_1(\alpha_1)\big]\Big\}$ while we move $(\alpha_0,\alpha_1)$ along the direction from $(\alpha_0,\alpha_1) = (\alpha_0^{S},\alpha_1^{S})$ to $(\alpha_0,\alpha_1) = (\alpha_0^{*},\alpha_1^{*})$ (the first two terms are shifted by constants in the plot). Note that the profit obtained by the principal from implementing $(\alpha_0,\alpha_1)$ in a cost-minimizing way is a weighted sum of those three terms. 

We have that $\Theta(\alpha_0) - K_0'(\alpha_0)\cdot \alpha_0$ and $\Theta(\alpha_1) - K_1'(\alpha_1)\cdot \alpha_1$ both increase when we move $(\alpha_0,\alpha_1)$ from $(\alpha_0^{S},\alpha_1^{S})$ to $(\alpha_0^{*},\alpha_1^{*})$, while $m_1$ consistently remains 0. If we move $(\alpha_0,\alpha_1)$ further, though the first two terms would also increase, $m_1$ is increased as well, making the implementation too costly for the principal. Figure \ref{EF-binary-eg} illustrates that starting from $(\alpha_0,\alpha_1) = (\alpha_0^{S},\alpha_1^{S})$, it is possible to move $\alpha_0$ and $\alpha_1$ towards the pure moral hazard solution without increasing the constant payment $m_1$, which strictly improves the principal's profit. Hence, the principal prefers screening (implementing $(\alpha_0^{*},\alpha_1^{*})$) to pooling (implementing $(\alpha_0^{S},\alpha_1^{S})$) in this problem. $\hfill\lhd$
\begin{figure}[h]
    \centering
    \includegraphics[width=0.5\linewidth]{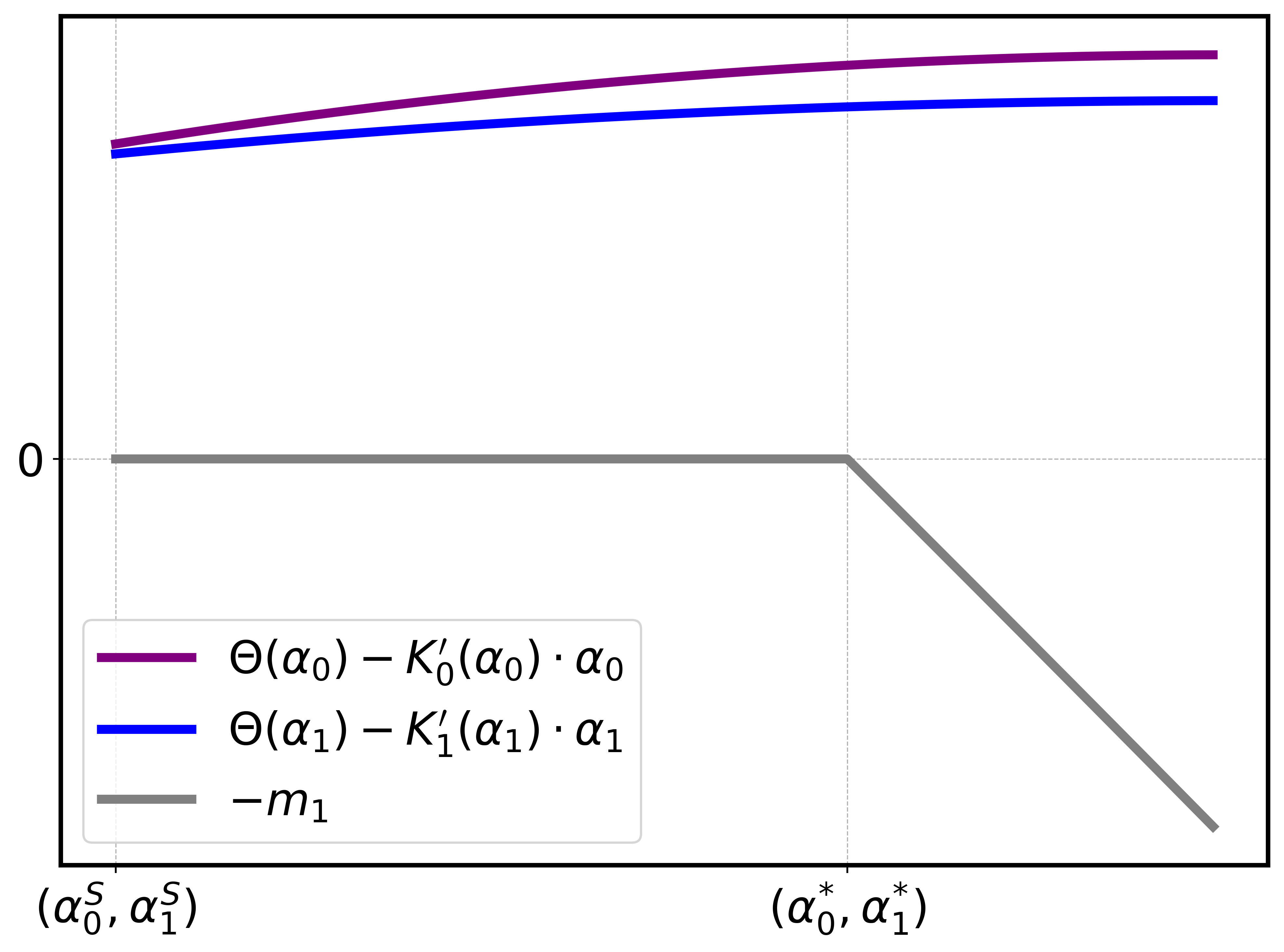}
    \caption{The plot of $\Theta(\alpha_0) - K_0'(\alpha_0)\cdot \alpha_0$, $\Theta(\alpha_1) - K_1'(\alpha_1)\cdot \alpha_1$ and $-m_1$ for Example \ref{EF-ex1}.}
    \label{EF-binary-eg}
\end{figure}
\end{example}

We argue that the optimality of screening in our model arises because the principal possesses large contracting power within the flexible moral hazard framework. Specifically, she can regulate the supports of the implemented output distributions by designing the ranges of contracts while maintaining the implemented aggregate efforts. To see that, consider a flexible moral hazard problem with a binary output space $\{\underline{x},\overline{x}\}$. The agent can choose an arbitrary output distribution from $\Delta(\{\underline{x},\overline{x}\})$. In this case, each output distribution is parameterized by the probability that the high output occurs, situating this problem within the classical moral hazard framework. According to \cite{gottlieb22}, it is optimal to provide a single contract in the adverse selection problem in this setting. This arises from the fact that if the principal intends to provide contracts with different powers, she must include a strictly positive constant term in the lower-powered contract. This payment proves too costly, leading the principal to prefer offering a single contract instead.

In our setting, the output space is an interval rather than a binary set. When the principal aims to provide a higher-powered type-0 contract and a lower-powered type-1 contract, she can limit the range of the type-0 contract to include only low outputs, while maintaining the implemented aggregate effort. Despite the higher power of the type-0 contract, it remains unattractive to the type-1 agent who prefers to produce higher outputs. This approach allows for effective screening of the agent at a lower cost, and in some cases, the screening cost can be negligible, as demonstrated in Example \ref{EF-ex1}. 
\subsection{Welfare Analysis under Convex Effort Functions}
In this section, we look into the welfare implications of the optimal menus under convex effort functions. Given $t\in \{0,1\}$ and $\alpha \in [\underline{c},\overline{c}]$, define $S_t(\alpha) := \Theta(\alpha) - K_t(\alpha)$ as the social surplus generated by the type-$t$ agent when an aggregate effort $\alpha$ is implemented (via a point mass output distribution). Let $(\alpha_0^{FB},\alpha_1^{FB})$ denote the optimal solution to the first-best problem such that $\alpha_t^{FB}: = \argmax_{\alpha\in [\underline{c},\overline{c}]}S_t(\alpha)$ for $t\in \{0,1\}$. 

We compare the implemented aggregate efforts $(\alpha_0^{*},\alpha_1^{*})$ at an optimal menu $(a^{*},m^{*},R^{*},\mu^{*})$ to \eqref{MBC-P} with the pure moral hazard solution $(\alpha_0^{MH},\alpha_1^{MH})$ and the first-best solution $(\alpha_0^{FB},\alpha_1^{FB})$. The findings are summarized in Proposition \ref{EF-pp1}. 
\begin{proposition}
Assume that $c(\cdot)$ is convex. Take an optimal menu $(a^{*},m^{*},R^{*},\mu^{*})\in \mathcal{F}$ to \eqref{MBC-P}. Under Assumption \ref{OM-as1} and \ref{OM-as3}, the following statements hold:
\begin{enumerate}[label=\circled{\arabic*}]
    \item If $K_0'(\alpha_0^{MH})\leq K_1'(\alpha_1^{MH})$, then it holds that $\alpha_0^{MH}\leq \alpha_0^{*}\leq \alpha_0^{FB}$, and $\alpha_1^{*}\leq \alpha_1^{MH}\leq \alpha_1^{FB}$.
    \item If $K_0'(\alpha_0^{MH}) > K_1'(\alpha_1^{MH})$, then it holds that $\alpha_0^{*}\leq \alpha_0^{MH}\leq \alpha_0^{FB}$, and $\alpha_1^{MH}\leq \alpha_1^{*} \leq \alpha_1^{FB}$. 
\end{enumerate}
\label{EF-pp1}
\end{proposition}
\begin{proof} 
See Appendix \ref{MBCF-app-cef}.
\end{proof}
Intuitively, compared to the pure moral hazard problem, the principal intends to align the powers of the provided contracts more closely to reduce the screening cost in the presence of adverse selection. In the case of $K_0'(\alpha_0^{MH}) \leq K_1'(\alpha_1^{MH})$, adverse selection leads to an increase in the power of the type-0 contract and a decrease in the power of the type-1 contract. Conversely, when $K_0'(\alpha_0^{MH}) > K_1'(\alpha_1^{MH})$, these adjustments occur in the opposite direction. The implemented aggregate efforts are adjusted accordingly, yet they would not exceed those of the first-best solution $(\alpha_0^{FB},\alpha_1^{FB})$.

From Proposition \ref{EF-pp1}, we can conduct the comparison between the agent's payoffs and the resulting social surpluses relative to the pure moral hazard benchmark.
\begin{corollary}
Assume that $c(\cdot)$ is convex. Take an optimal menu $(a^{*},m^{*},R^{*},\mu^{*})\in \mathcal{F}$ that solves \eqref{MBC-P}. Under Assumption \ref{OM-as1} and \ref{OM-as3}, compared to the pure moral hazard benchmark, the following statements hold:
\begin{enumerate}[label=\circled{\arabic*}]
\item If $K_0'(\alpha_0^{MH})\leq K_1'(\alpha_1^{MH})$, the type-0 agent's payoff and generated social surplus are weakly higher (i.e., $S_0(\alpha_0^{*})\geq S_0(\alpha_0^{MH})$), while the type-1 agent's payoff and generated social surplus are weakly lower (i.e., $S_1(\alpha_1^{*}) \leq S_1(\alpha_1^{MH})$).
\item If $K_0'(\alpha_0^{MH})>K_1'(\alpha_1^{MH})$, the type-0 agent's payoff and generated social surplus are weakly lower (i.e., $S_0(\alpha_0^{*})\geq S_0(\alpha_0^{MH})$), while the type-1 agent's payoff and generated social surplus are weakly higher (i.e., $S_1(\alpha_1^{*}) \leq S_1(\alpha_1^{MH})$). 
\end{enumerate}
\label{EF-co1}
\end{corollary}
\section{Optimality of Single Full-range Contracts}
\label{MBCF-os}
In this section, we explore when the optimum of the principal's problem \eqref{MBC-P} can be reached by providing a single full-range contract with a general effort function. We begin by formally defining what it means for a menu to consist of such a contract. 
\begin{definition}
A menu $(a,m,R,\mu)\in \mathcal{F}$ takes the form of \textit{a single full-range contract} if 
\begin{equation*}
K_0'(\alpha_0) = K_1'(\alpha_1) \text{ and }R=([\underline{x},\overline{x}],[\underline{x},\overline{x}]).
\end{equation*}
\end{definition}

The specialty of the single full-range contracts lies in that, when the optimal menu takes this form, it does not distinguish between different types of agents, and no output level needs to be excluded from the ranges of the contracts. Technically, when a single full-range contract is optimal, the principal's problem becomes trivial, as it reduces to solving for the optimal slope coefficient of a single contract. Moreover, it follows immediately that the recommended output distributions must be those that attain the concavification value $\Theta(\cdot)$. 

Theorem \ref{EF-th1} established a sufficient and necessary condition under which it is optimal for the principal to provide a single full-range contract when the effort function $c(\cdot)$ is convex. We can show that the condition proposed in Theorem \ref{EF-th1} is also sufficient to ensure that both contracts in the optimal menu have the same level of power when the effort function takes a general form. 

Theorem \ref{OSC-th1} shows that a menu consisting of a single full-range contract is optimal with an effort function $c(\cdot)$ if providing this single full-range contract is also optimal to the principal's problem where the effort function is given by $\Theta^{-1}(\cdot)$, which is the largest convex function that is everywhere weakly lower than $c(\cdot)$.   
\begin{theorem}
Take any $c:[\underline{x},\overline{x}]\to \mathbb{R}_{+}$ that is differentiable and strictly increasing. If a single full-range contract is optimal to \eqref{MBC-P} with effort function $\Theta^{-1}(\cdot)$, it is also optimal to \eqref{MBC-P} when the effort function is given by $c(\cdot)$.  
\label{OSC-th1}
\end{theorem}
\begin{proof} 
With the effort function $c(\cdot)$, the principal's problem is given by
\begin{equation}
\begin{aligned}
\max_{(a,m,R,\mu)\in \mathcal{F}} \quad&p_0 \big[\int x\mu_0(dx) - K_0'(\alpha_0)\cdot \alpha_0 - m_0\big] \\
  \quad+ &p_1\big[\int x\mu_1(dx) - K_1'(\alpha_1)\cdot \alpha_1 - m_1\big].
\end{aligned}
\label{OSC-eq1}
\end{equation}

When the effort function is given by $\Theta^{-1}(\cdot)$, the principal's problem is correspondingly given by
\begin{equation}
\begin{aligned}
\max_{(a,m,R,\mu)\in \mathcal{F}} \quad & p_0 \big[\Theta(\alpha_0) - K_0'(\alpha_0)\cdot \alpha_0 - m_0 \big] + p_1 \big[\Theta(\alpha_1) - K_1'(\alpha_1)\cdot \alpha_1 - m_1\big]. 
\end{aligned}
\label{OSC-eq2}
\end{equation}

Take a feasible menu $M = (a,m,R,\mu)\in \mathcal{F}$. Let $P(M)$ denote the value of the objective function of \eqref{OSC-eq1} given menu $M$ and let $\tilde{P}(M)$ denote the corresponding value of the objective function from \eqref{OSC-eq2}. 

For any feasible menu $M = (a,m,R,\mu) \in \mathcal{F}$, take $\gamma_t:= c\#\mu_t$ for $t\in \{0,1\}$. It holds that
\begin{equation}
\begin{aligned}
P(M) &= p_0 \big[\int \theta(c)\gamma_0(dc) - K_0'(\alpha_0)\cdot \alpha_0 - m_0\big] + p_1\big[\int \theta(c)\gamma_1(dc) - K_1'(\alpha_1)\cdot \alpha_1 - m_1\big] \\
& \leq p_0  \big[\Theta(\alpha_0) - K_0'(\alpha_0)\cdot \alpha_0 - m_0\big] + p_1\big[\Theta(\alpha_1)- K_1'(\alpha_1)\cdot \alpha_1 - m_1\big] \\
& = \tilde{P}(M),
\end{aligned}
\label{OSC-eq3}
\end{equation}
where the first inequality follows from that $\int c\gamma_1(dc) = \alpha_1$. From \eqref{OSC-eq3}, we have that the optimal value of \eqref{OSC-eq1} is no larger than that of \eqref{OSC-eq2}. 

Suppose there exists $a^{*}\in [\underline{c},\overline{c}]^2$, $m^{*}\in \mathbb{R}^2$ and $\mu^{\tilde{P}}\in \Delta([\underline{x},\overline{x}])^2$ such that the menu $M_{\tilde{P}}:=(a^{*},m^{*},R_F,\mu^{\tilde{P}})\in \mathcal{F}$ takes the form of a single full-range contract (i.e., $K_0'(\alpha_0^{*}) = K_1'(\alpha_1^{*})$), and optimally solves \eqref{OSC-eq2}. For $t\in \{0,1\}$, let $\mu_t^{P}$ denote the distribution that attains the concavification value $\Theta(\alpha_t^{*})$, i.e., it holds that $\int c(x) \mu_t^{P}(dx) = \alpha_t^{*}$ and $\int x\mu_t^{P}(dx) = \Theta(\alpha_t^{*})$. Take $\mu^{P}: = (\mu_0^{P},\mu_1^{P})$. Define $M_P := (a^{*},m^{*},R_F,\mu^P)$. We have that $M_P$ is feasible since it is a menu involving full-range contracts. Furthermore, the value of the objective function of \eqref{OSC-eq1} given menu $M_P$ is given by
\begin{equation*}
\begin{aligned}
P(M_P) &= p_0 \big[\int x\mu_0^{P}(dx) - K_0'(\alpha_0^{*})\cdot \alpha_0^{*} - m_0^{*}\big] + p_1\big[\int x\mu_1^{P}(dc) - K_1'(\alpha_1^{*})\cdot \alpha_1^{*} - m_1^{*}\big] \\
&= p_0  \big[\Theta(\alpha_0^{*}) - K_0'(\alpha_0^{*})\cdot \alpha_0^{*} - m_0^{*}\big] + p_1\big[\Theta(\alpha_1^{*})- K_1'(\alpha_1^{*})\cdot \alpha_1^{*} - m_1^{*}\big] \\
&= \tilde{P}(M_{\tilde{P}}),
\end{aligned}
\end{equation*}
which is equal to the optimal value of problem \eqref{OSC-eq2}. Hence, $M_P$ optimally solves \eqref{OSC-eq1}. 
\end{proof}

Theorem \ref{OSC-th1} follows from the observation that, under the effort function $\Theta^{-1}(\cdot)$, the effort required for production is point-wise lower than under $c(\cdot)$. Intuitively, the principal's optimal profit in the former problem is at least as large as in the latter. Moreover, the principal can achieve the same amount of profit by providing the same single full-range contract in both problems. Therefore, if a single full-range contract is optimal in the former problem, it must be optimal in the latter. 

Alternatively, Theorem \ref{OSC-th1} can be viewed as reflecting the idea that truncated ranges are less beneficial for the principal when the effort function is non-convex. With a menu that takes the form of a single full-range contract and a potentially non-convex effort function $c(\cdot)$, the principal can reach the same profit level as would be obtainable under the convexified effort function $\Theta^{-1}(\cdot)$, since full-range contracts allow the principal to recommend any output distribution to the agent. By contrast, given a menu that involves contracts with truncated ranges, the profit that can be reached under $c(\cdot)$ may be smaller than under $\Theta^{-1}(\cdot)$. For instance, if $c(\cdot)$ is non-convex, the truncations of contract ranges limit the output distributions that can be implemented, thereby reducing the principal's achievable profit. 

According to Theorem \ref{EF-th1}, if it holds that $K_0'(\alpha_0^{MH})\leq K_1'(\alpha_1^{MH})$, then under Assumption \ref{OM-as1} and \ref{OM-as3}, we have $K_0'(\alpha_0^{*}) = K_1'(\alpha_1^{*})$ at the optimal menu $(a^{*},m^{*},R^{*},\mu^{*})$ that solves problem \eqref{MBC-P} with the effort function $\Theta^{-1}(\cdot)$. According to Theorem \ref{OM-pp3}, the menu $(a^{*},m^{*},R_F,\mu^{*})$ is also optimal. From Theorem \ref{OSC-th1}, there exists $\mu^P\in \Delta([\underline{x},\overline{x}])^2$ such that $(a^{*},m^{*},R_F,\mu^P)\in \mathcal{F}$ is optimal to problem \eqref{MBC-P} with the effort function $c(\cdot)$. Thereby, we've established Corollary \ref{OSC-co1}, which states that $K_0'(\alpha_0^{MH}) \leq K_1'(\alpha_1^{MH})$ is a sufficient condition for when a single full-range contract is optimal under a general effort function.
\begin{corollary}
Take any $c:[\underline{x},\overline{x}]\to \mathbb{R}_{+}$ that is differentiable and strictly increasing. Under Assumption \ref{OM-as1} and \ref{OM-as3}, if $K_0'(\alpha_0^{MH})\leq K_1'(\alpha_1^{MH})$, then there exists a menu $(a^{*},m^{*},R_F,\mu^P)\in \mathcal{F}$ that optimally solves \eqref{MBC-P} with effort function $c(\cdot)$ and satisfies that 
\begin{equation*}
    K_0'(\alpha_0^{*}) = K_1'(\alpha_1^{*}). 
\end{equation*}
\label{OSC-co1}
\end{corollary}

In Example \ref{OSC-ex1}, we apply Corollary \ref{OSC-co1} to a scenario where the cost functions of both types are quadratic in the aggregate effort, and show that it is always optimal to provide a single full-range contract as long as the effort function is concave. 
\begin{example}
Consider the case where $K_0(\alpha) = \beta_0\cdot\alpha^2$ and $K_1(\alpha) = \beta_1\cdot \alpha^2$ for all $\alpha \in [\underline{c},\overline{c}]$ with $\beta_0 > \beta_1 > 0$. Assumption \ref{OM-as1} is satisfied in this case. When the effort function $c(\cdot)$ is concave, we have that $\Theta(\cdot)$ is linear and $\Theta'(\alpha) = \xi$ for all $\alpha\in [\underline{c},\overline{c}]$, where $\xi>0$ is constant. 

It holds that
\begin{equation*}
\xi = K_0'(\alpha_0^{MH}) + K_0''(\alpha_0^{MH}) \cdot \alpha_0^{MH} = K_1'(\alpha_1^{MH}) + K_1''(\alpha_1^{MH})\cdot \alpha_1^{MH},
\end{equation*}
which indicates that $K_0'(\alpha_0^{MH}) = K_1'(\alpha_1^{MH})$. Applying Corollary \ref{OSC-co1}, we have that it is optimal to provide a single full-range contract in this case. $\hfill\lhd$
\label{OSC-ex1}
\end{example}

\subsection{Welfare Analysis}

In this section, we study the welfare implications of the optimal menus under general effort functions by comparing the optimal aggregate efforts at the optimal menus with the pure moral hazard solution $(\alpha_0^{MH},\alpha_1^{MH})$ and the first-best solution $(\alpha_0^{FB},\alpha_1^{FB})$. 

When the optimal menu takes the form of a single full-range contract, we can establish the following result on the optimal aggregate efforts.
\begin{proposition}
If $(a^{*},m^{*},R^{*},\mu^{*})$ optimally solves \eqref{MBC-P} and takes the form of a single full-range contract, under Assumption \ref{OM-as1} and \ref{OM-as3}, the following statements hold:
\begin{enumerate}[label=\circled{\arabic*}]
    \item With $K_0'(\alpha_0^{MH}) \leq K_1'(\alpha_1^{MH})$, it holds that $\alpha_0^{MH}\leq \alpha_0^{*} \leq \alpha_0^{FB}$, and $\alpha_1^{*}\leq \alpha_1^{MH}\leq \alpha_1^{FB}$. 
    \item With $K_0'(\alpha_0^{MH}) > K_1'(\alpha_1^{MH})$, it holds that $\alpha_0^{*}\leq \alpha_0^{MH}$, and $\alpha_1^{*}\geq \alpha_1^{MH}$.
\end{enumerate}
\label{WA-pp1}
\end{proposition}
\begin{proof}
    See Appendix \ref{MBCF-app-os}.
\end{proof}
Proposition \ref{WA-pp1} demonstrates that, if $K_0'(\alpha_0^{MH})\leq K_1'(\alpha_1^{MH})$, the introduction of adverse selection increases the type-0 agent’s aggregate effort and payoff while decreasing those of the type-1 agent, relative to the pure moral hazard benchmark. Conversely, if $K_0'(\alpha_0^{MH})>K_1'(\alpha_1^{MH})$, the effects are reversed. Additionally, in the former case, both types exert aggregate efforts that fall short of the first-best level, implying that the social surplus generated by the type-0 agent weakly increases, while that of the type-1 agent weakly decreases under adverse selection.

On the contrary, when it is not possible for the principal to attain the maximal profit by providing a single full-range contract, we obtain the following result regarding the optimal aggregate efforts. 
\begin{proposition}
If the optimum of \eqref{MBC-P} can not be reached by a single full-range contract, under Assumption \ref{OM-as1} and \ref{OM-as3}, for any $(a^{*},m^{*},R^{*},\mu^{*}) \in \mathcal{F}$ that solves \eqref{MBC-P}, it holds that \circled{1} $\alpha_0^{*}\leq \alpha_0^{FB}$, and \circled{2} $\alpha_1^{*}\geq \alpha_1^{MH}$. 
\label{WA-pp2}
\end{proposition}
\begin{proof}
See Appendix \ref{MBCF-app-os}.
\end{proof}
Proposition \ref{WA-pp2} shows that when single full-range contracts are not optimal, the optimal type-0 aggregate effort is distorted downward relative to the first-best benchmark, while the optimal type-1 aggregate effort is distorted upward relative to the pure moral hazard benchmark. The first statement follows from the fact that if $\alpha_0^{*} > \alpha_0^{FB}$, the principal could slightly reduce the power of the type-0 contract and adjust the constant payment accordingly to keep the type-0 agent's payoff unchanged. This adjustment would raise the social surplus generated by the type-0 agent and, in turn, increase the principal's profit from this agent. The second statement follows from the idea that if $\alpha_1^{*} < \alpha_1^{MH}$, the principal could increase the profit from the type-1 agent by slightly increasing the power of the type-1 contract. 
\section{Generalizations}
\label{MBCF-general}
In this section, we generalize Theorem \ref{OM-pp3} to settings with either a finite type space or mean-based cost functions.
\subsection{Adverse Selection with Finite Types}
We consider a setting where the agent has a private type $t\in \{1,2,...,N\}$ ($N\in \mathbb{N}$ and $N>2$). It is common knowledge that type $t$ occurs with probability $p_t>0$ ($\sum_{t=1}^{N}p_t = 1$). The agent's type only affects his cost function. Specifically, the cost function of a type-$t$ agent is given by 
\begin{equation*}
    C_t(\mu) = K_t\big(\int c(y)\mu(dy)\big),\quad \forall \mu \in \Delta([\underline{x},\overline{x}]).
\end{equation*}

The principal designs a menu $(a,m,R,\mu)$, where $a \in [\underline{c},\overline{c}]^{N}$ is the vector of aggregate efforts, $ m\in \mathbb{R}^{N}$ denotes the constant payments in the contracts, $R\in (2^{[\underline{x},\overline{x}]})^{N}$ is the vector of contract ranges, and $\mu \in \Delta([\underline{c},\overline{c}])^N$ denotes the output distributions recommended to all types of agents. A menu is feasible if it satisfies moral hazard constraints, incentive compatibility constraints, and limited liability constraints. Let $\mathcal{F}_N$ denote the set of feasible menus for this N-types setting. The principal's problem is given by
\begin{equation}
    \max_{(a,m,R,\mu)\in \mathcal{F}_N}\quad \sum_{t=1}^{N} p_t\big[\int x\mu_t(dx) - K_t'(\alpha_t)\cdot \alpha_t - m_t\big].
\tag{P-N}
\label{NT-P-N}
\end{equation}

Analogous to Assumption \ref{OM-as1}, we impose the following assumption in this setting. 
\begin{assumption}
$K_1'(\alpha) > K_2'(\alpha)>...>K_N'(\alpha)$ for all $\alpha \in [\underline{c},\overline{c}]$.  
\label{NT-as1}
\end{assumption}

In Proposition \ref{NT-pp2}, we generalize Proposition \ref{OM-pp2} in this setting and show that at any feasible menu, it is possible to extend the contract range of the lowest-type (type-1) contract to the leftmost point and extend that of the highest-type (type-$N$) contract to the rightmost point without harming the feasibility of the menu. This implies that the principal can always provide the lowest-type (highest-type) agent with sufficient incentives to produce low-cost, low-return (high-cost, high-return) outputs while still ensuring the implementability of the menu. The proof of Proposition \ref{NT-pp2} is similar to that of Proposition \ref{OM-pp2} and is hence omitted.
\begin{proposition}
\label{NT-pp2}
Take any feasible menu $(a,m,R,\mu)\in \mathcal{F}_N$, and take $\overline{x}_1:= \sup R_1$ and $\underline{x}_N:= \inf R_N$. Define
\begin{equation*}
    R_1':= [\underline{x},\overline{x}_1]\text{ and }R_N':=[\underline{x}_N, \overline{x}].
\end{equation*}
Under Assumption \ref{NT-as1}, $(a,m,R',\mu)$ with $R'= (R_1',R_2,...,R_{N-1},R_N')$ is also feasible.
\end{proposition}

Next, we extend Theorem \ref{OM-pp3} to the current setting in the following proposition.
\begin{proposition}
Take an optimal menu $(a^{*},m^{*},R^{*},\mu^{*})\in \mathcal{F}_N$ that solves \eqref{NT-P-N}. Define $R_F:= [\underline{x},\overline{x}]^{N}$. Under Assumption \ref{NT-as1}, the menu $(a^{*},m^{*},R_F,\mu^{*})$ is feasible if and only if 
\begin{equation*}
    K_i'(\alpha_i^{*}) = K_j'(\alpha_j^{*}),\quad \forall i,j\in \{1,2,...,N\}. 
\end{equation*}
\label{NT-pp1}
\end{proposition}
\begin{proof}
See Appendix \ref{MBCF-app-general}.
\end{proof}
Proposition \ref{NT-pp1} shows that, at an optimal menu $(a^{*},m^{*},R^{*},\mu^{*})$, all contract ranges can be extended to full without violating feasibility if and only if all contracts in the menu have the same power. We establish this by proving that if an optimal menu only involves full-range contracts, then their powers must be identical. The proof proceeds by contradiction: suppose an optimal menu involves only full-range contracts but features different powers. Then there must exist a set of higher-type contracts whose powers are strictly greater than those of some lower-type contracts. In that case, the principal can strictly improve her profit by lowering the powers of the higher-type contracts, contradicting the optimality of the menu. 
\subsection{Mean-based Cost Functions}
In this section, we explore flexible moral hazard problems with adverse selection, focusing on scenarios where the effort functions are affine. While imposing restrictions on the shape of effort functions, we allow for a more general adverse selection setting than the model described in section \ref{MBCF-model}. In this case, the cost of the agent to choose an output distribution is determined by the mean of this distribution.
As in section \ref{MBCF-model}, the agent is able to choose an arbitrary output distribution over $[\underline{x},\overline{x}]\subset \mathbb{R}_{+}$. The effort function $c:[\underline{x},\overline{x}]\to \mathbb{R}_{+}$ is affine and strictly increasing, i.e., there exist $c>0$ and $d\in \mathbb{R}$ such that $c(x) = cx+d$ for $x\in [\underline{x},\overline{x}]$. Similarly, we use $\underline{c}$ and $\overline{c}$ to denote the lowest and the highest effort. 

The agent has a private type $t\in T\subset \mathbb{R}$, which follows a probability distribution $\phi\in \Delta(T)$ and impacts the agent's cost of choosing output distributions. For a type-$t$ ($t\in T$) agent, the cost of choosing $\mu \in \Delta([\underline{x},\overline{x}])$ is given by $C_t(\mu) = K_t\big(\int c(y)\mu(dy)\big)$. We assume that for all $t\in T$, $K_t(\cdot)$ is strictly increasing, strictly convex, and twice differentiable. Apart from this, we do not impose any further restriction on the type set $T$ and the cost function $\{K_t(\cdot)\}_{t\in T}$. 
The timeline of the problem would be the same as the previous one. The principal optimally designs a menu $\{s_t(\cdot),\mu_t\}_{t\in T}$ to maximize her expected profit, where $s_t(\cdot)$ is the contract provided to the agent if he reports his type as $t$, and $\mu_t$ is the output distribution recommended to him. The principal's problem is given by:
\begin{equation}
\begin{aligned}
    \max_{\{s_t(\cdot),\mu_t\}_{t\in T}} \quad& \int_T \int \big(x-s_t(x)\big)\mu_t(dx)\phi(dt) \\
    s.t.\quad & (t,\mu_t) \in \argmax_{t'\in T,\atop\mu \in \Delta([\underline{x},\overline{x}])} \int s_{t'}(x)\mu(dx) - K_t\big(\int c(y)\mu(dy)\big), \quad \forall t \in T \\
    & s_t(x)\geq 0, \quad \forall x\in [\underline{x},\overline{x}],\quad\forall t \in T.
\end{aligned}
\label{OSG-GP}
\end{equation}

Similarly, we say a contract has a full range if it provides maximal incentives on the whole output space for the agent to produce. We formally define full-range contracts in this setting as follows.  
\begin{definition}
Take a contract $s:[\underline{x},\overline{x}]\to \mathbb{R}$. We say $s(\cdot)$ has \textit{full range} with regard to type $t$ if there exists $m\in \mathbb{R}$, such that
\begin{equation*}
 s(x)= K_t'(\alpha_t^{*})\cdot c(x) + m,\quad\forall x\in [\underline{x},\overline{x}],
\end{equation*}
where $\alpha_t^{*} =\argmax_{\alpha\in [\underline{c},\overline{c}]}\big[\widehat{s\circ \theta}(\alpha) - K_t(\alpha)\big]$. 
\end{definition}

In Proposition \ref{OSG-pp1}, we establish a result akin to Theorem \ref{OM-pp3}, demonstrating that if the principal screens the agent by varying the powers of contracts in the optimal menu, she must simultaneously limit the ranges of some contracts to assist the screening. Consequently, either the optimal menu takes the form of a single full-range contract or includes some contracts with truncated ranges.
\begin{assumption}
For all $t\in T$, it holds that
\begin{equation*}
K_t'(\underline{c}) < \frac{1}{c}\text{ and } K_t'(\overline{c}) > \frac{1}{c}.
\end{equation*}
\label{OSG-as1}
\end{assumption}
\begin{proposition}
Under Assumption \ref{OSG-as1}, take an optimal menu $\{s_t^{*}(\cdot),\mu_t^{*}\}_{t\in T}$ that solves problem \eqref{OSG-GP}, one of the following two statements holds: 
\begin{enumerate}[label=\circled{\arabic*}]
   \item There exist $T_0 \subset T$ with $\phi(T_0) = 0$ and a contract $\tilde{s}:[\underline{x},\overline{x}]\to \mathbb{R}$ such that $s_t^{*}(\cdot) = \tilde{s}(\cdot)$, and $\tilde{s}(\cdot)$ has full range with regard to type $t$, for all $t\in T\backslash T_0$. 
   \item There exists $T_1 \subset T$ with $\phi(T_1) >0$ such that $s_t^{*}(\cdot)$ does not have full range with regard to type $t$, for all $t\in T_1$.
\end{enumerate}
\label{OSG-pp1}
\end{proposition}
\begin{proof}
See Appendix \ref{MBCF-app-general}.
\end{proof}

Proposition \ref{OSG-pp1} demonstrates that, given an optimal menu to \eqref{OSG-GP}, either it consists of a single contract that has full range with regard to each type of agent, or for a positive probability measure of types of agents, they are assigned with contracts that do not possess full range. The result can be established by showing that if we restrict ourselves to menus including full-range contracts only, then it is necessary that the optimal menu takes the form of a single full-range contract. Assumption \ref{OSG-as1} requires that each type of agent would optimally choose an interior aggregate effort if he is provided with a selling-the-firm contract (i.e., a contract that takes the form of $s(x) = x - \rho$, where $\rho$ is a constant). 

In the proof of Proposition \ref{OSG-pp1}, we first establish that starting from an arbitrary feasible menu that consists of full-range contracts only, it is possible to weakly improve the principal's profit by replacing the menu with a single full-range contract. The argument is analogous to that in \cite{gottlieb22} by which they establish the optimality of a single contract. While their reasoning relies on the assumption imposed on the set of possible output distributions, ours is based on the linearity of full-range contracts. We then show that if the original menu does not take the form of a single full-range contract, the principal can strictly benefit from the replacement. This step is enabled by Assumption \ref{OSG-as1}.

\section{Conclusion}
\label{MBCF-con}
This paper studies the optimal menu for a moral hazard problem with adverse selection, where each type of agent directly controls the output distribution under a moment-based cost function. We characterize the set of implementable output distributions and analyze the structure of optimal menus. When optimal contracts feature distinct powers, either some high outputs must be excluded from the low-type contract, or some low outputs must be excluded from the high-type contract. Otherwise, it is optimal for the principal to provide a single full-range contract. We establish a sufficient and necessary condition for the optimality of a single full-range contract under convex effort functions and show that this condition is also sufficient with general effort functions.

In our model, the contract range is defined as the set of outputs that gives the agent the strongest production incentives. Any output outside this range yields zero payment to the agent. This reflects real-world situations in which producing such outputs is not specified in the contract, and therefore, the agent receives no payment for doing so. Our results provide guidance on which contractual terms should be deliberately omitted from contracts when there are heterogeneous agents. At the same time, we show that the principal can sometimes benefit from providing menus with different contract ranges. This captures real-world settings where employers extract more private information by segmenting tasks into distinct sets and allowing employees to select the set they prefer. 

Although monotonicity is often considered a desirable property of contracts in various applications, our model does not require contracts to be monotonically increasing. Nevertheless, this restriction can be imposed without altering our results by extending the tail end of each contract to the highest contract value within its range. By doing so, both contracts are flat in the tail end, thereby failing to provide sufficient production incentives for the agents. Consequently, the agent, regardless of his type, would neglect the tail ends of both contracts, and his optimal responses to those contracts are unaffected. 

\bibliography{main.bib}

\appendix
\section*{Appendix}
\setcounter{section}{0}
\section{Supplementary Results for Section \ref{MBCF-model}}
\label{MBCF-app-model}
Before stating Proposition \ref{MBCF-app-pp1}, we first define the MH, LL, and IC constraints for a menu that takes the original form $(s_0(\cdot),s_1(\cdot),\mu_0,\mu_1)$.
\begin{definition}
A menu $(s_0(\cdot),s_1(\cdot),\mu_0,\mu_1)$ is \textit{feasible} if \begin{enumerate}
\item \textit{MH constraints}: $s_t(\cdot)$ implements $\mu_t$ with regard to the type-$t$ agent for $t\in \{0,1\}$.
\item \textit{LL constraints}: $s_t(x)\geq 0$ for all $x\in [\underline{x},\overline{x}]$ and $t\in \{0,1\}$.
\item \textit{IC constraints}: It is optimal for the agent to truthfully report his type:
\begin{subequations}
\begin{align}
\int s_0(x)\mu_0(dx) - C(\mu_0)\geq \max_{\alpha\in[\underline{c},\overline{c}]}\big[\widehat{s_1\circ\theta}(\alpha)- K_0(\alpha)\big], \label{MBCF-app-eq7a}\\
\int s_1(x)\mu_1(dx) - C(\mu_1)\geq \max_{\alpha\in[\underline{c},\overline{c}]}\big[\widehat{s_0\circ\theta}(\alpha)- K_1(\alpha)\big]. \label{MBCF-app-eq7b}
\end{align}
\end{subequations}
\end{enumerate} 
\end{definition}
\begin{proposition}
Given a feasible menu $(s_0(\cdot),s_1(\cdot),\mu_0,\mu_1)$ that satisfies MH, LL, IC constraints. There exist contracts $\widetilde{s_0}(\cdot)$ and $\widetilde{s_1}(\cdot)$ that take the following form
\begin{equation}
\widetilde{s_t}(x) = \mathbb{1}_{\{x\in R_t\}}\big[K_t'(\alpha_t)\cdot c(x) + m_t\big], \quad \forall t\in \{0,1\},
\label{MBCF-app-eq8}
\end{equation}
where $\alpha_0, \alpha_1 \in [\underline{c},\overline{c}]$, $m_0,m_1\in \mathbb{R}_{+}$, and $R_0, R_1$ are measurable subsets of $[\underline{x},\overline{x}]$ such that the menu $(\widetilde{s_0}(\cdot),\widetilde{s_1}(\cdot),\mu_0,\mu_1)$ is feasible, and it leads to the same payoffs to the principal and both types of agents as $(s_0(\cdot),s_1(\cdot),\mu_0,\mu_1)$.   
\label{MBCF-app-pp1}
\end{proposition}
\begin{proof}
Take $\alpha_t = \int c(y)\mu_t(dy)$ for $t\in \{0,1\}$. Since $s_t(\cdot)$ implements distribution $\mu_t$ with regard to the type-$t$ agent for $t\in \{0,1\}$, according to the first-order condition, there exists constant $m_t\in \mathbb{R}_{+}$ such that,
\begin{equation*}
 s_t(x) \left\{
 \begin{array}{ll}
  = K_t'\big(\alpha_t\big)\cdot c(x) + m_t    &\text{ if }x\in \text{supp}(\mu_t), \\
  \leq K_t'\big(\alpha_t\big)\cdot c(x) + m_t     &\text{ otherwise}.
 \end{array}
 \right.
 \label{PRE-FOCA}
\end{equation*}

Define $R_t:= \{x\in [\underline{x},\overline{x}]:s_t(x) = K_t'(\alpha_t)\cdot c(x)+m_t\}$. Then $\widetilde{s_t}(\cdot)$, pinned down by \eqref{MBCF-app-eq8}, satisfies that $\widetilde{s_t}(x) = s_t(x) = K_t'(\alpha_t)\cdot c(x)+m_t$ for all $x\in \text{supp}(\mu_t)$ and that $\widetilde{s_t}(x) \leq s_t(x)\leq K_t'(\alpha_t)\cdot c(x)+m_t$ otherwise. Hence, $\mu_t$ can be implemented by $\widetilde{s_t}(\cdot)$. And by $m_t\in \mathbb{R}_{+}$, $\widetilde{s_t}(x)\geq 0$ for all $x\in [\underline{x},\overline{x}]$. Hence, $(\widetilde{s_0}(\cdot),\widetilde{s_1}(\cdot),\mu_0,\mu_1)$ satisfies the MH and LL constraints. It suffices to show that it satisfies the IC constraints. It holds that 
\begin{equation*}
\begin{aligned}
\int \widetilde{s_0}(x)\mu_0(dx) - C(\mu_0) &= \int s_0(x)\mu_0(dx) - C(\mu_0) \\
&\geq \max_{\alpha\in[\underline{c},\overline{c}]}\big[\widehat{s_1\circ\theta}(\alpha)- K_0(\alpha)\big]\\
&\geq \max_{\alpha\in[\underline{c},\overline{c}]}\big[\widehat{\widetilde{s_1}\circ\theta}(\alpha)- K_0(\alpha)\big],
\end{aligned}
\end{equation*}
where the first line follows from that $\widetilde{s_0}(x) = s_0(x)$ for all $x\in \text{supp}(\mu_0)$, the second line follows from \eqref{MBCF-app-eq7a}, and the last line follows from $\widetilde{s_1}(x)\leq s_1(x)$ for all $x\in [\underline{x},\overline{x}]$. Hence, the type-0 agent would truthfully report his type under $(\widetilde{s_0}(\cdot),\widetilde{s_1}(\cdot),\mu_0,\mu_1)$. Similarly, we can establish that the type-1 agent also prefers truthful reporting in this case. 

Hence, the menu $(\widetilde{s_0}(\cdot),\widetilde{s_1}(\cdot),\mu_0,\mu_1)$ is feasible. And it leads to the same payoff to all parties as $(s_0(\cdot),s_1(\cdot),\mu_0,\mu_1)$ since it implements the same output distributions $\mu_0,\mu_1$ and contract $\widetilde{s_t}(\cdot)$ coincides with $s_t(\cdot)$ on the support of $\mu_t$ for $t\in \{0,1\}$. 
\end{proof}
\begin{proposition}
There exists an optimal solution to \eqref{MBC-P}.
\label{MBCF-app-pp2}
\end{proposition}
\begin{proof}
To optimally solve problem \eqref{MBC-P}, we can restrict ourselves to menus where the output distributions are supported on at most two singletons and the contract ranges are equal to the supports of the output distributions. For $t\in \{0,1\}$, let $\{\underline{x}_t,\overline{x}_t\}$ denote the support of the type-$t$ distribution, and let $q_t$ denote the probability of $\overline{x}_t$ occurs in the corresponding distribution. With this restriction, the principal optimally designs
\begin{equation}
\alpha_t\in [\underline{c},\overline{c}],\quad m_t\in \mathbb{R}_{+},\quad \underline{x}_t,\overline{x}_t\in [\underline{x},\overline{x}],\quad q_t\in [0,1],\quad \forall t \in \{0,1\},
\label{MBCF-app-eq1}
\end{equation}
to maximize 
\begin{equation}
 p_0\big[(1-q_0)\underline{x}_0 + q_0\overline{x}_0 - K_0'(\alpha_0)\cdot \alpha_0 - m_0\big] + p_1\big[(1-q_1)\underline{x}_1 + q_1\overline{x}_1 - K_1'(\alpha_1)\cdot \alpha_1 - m_1\big],
\label{MBCF-app-eq2}
\end{equation}
subject to
\begin{subequations}
\begin{align}
    (1-q_t)c(\underline{x}_t) + q_t c(\overline{x}_t) = \alpha_t,\quad \forall t\in \{0,1\}, \label{MBCF-app-eq3a} \\
    K_0'(\alpha_0)\cdot \alpha_0 - K_0(\alpha_0) + m_0 \geq \max_{\alpha\in [\underline{c},\overline{c}]}\big[\widehat{s_1\circ \theta}(\alpha) - K_0(\alpha)\big],\label{MBCF-app-eq3b}\\
    K_1'(\alpha_1)\cdot \alpha_1 - K_1(\alpha_1) + m_1 \geq \max_{\alpha\in [\underline{c},\overline{c}]}\big[\widehat{s_0\circ \theta}(\alpha) - K_1(\alpha)\big],\label{MBCF-app-eq3c}
\end{align}
\end{subequations}
where $s_t\circ \theta(\alpha) = \mathbb{1}_{\{\alpha \in \{c(\underline{x}_t),c(\overline{x}_t)\}\}}\big[K_t'(\alpha_t)\cdot \alpha + m_t\big]$ for $t\in \{0,1\}$. Suppose we have an optimal solution to problem described through \eqref{MBCF-app-eq1} -- \eqref{MBCF-app-eq3c}, given by $\{\alpha_t^{*}, m_t^{*},\underline{x}_t^{*},\overline{x}_t^{*},q_t^{*}\}_{t\in\{0,1\}}$. Then the menu $(a^{*},m^{*},R^{*},\mu^{*})$, with $a^{*} = (\alpha_0^{*},\alpha_1^{*})$, $m^{*} = (m_0^{*},m_1^{*})$, $R^{*} = (\{\underline{x}_0^{*},\overline{x}_0^{*}\},\{\underline{x}_1^{*},\overline{x}_1^{*}\})$ and $\mu^{*} = (\mu_0^{*},\mu_1^{*})$ where $\mu_t^{*}$ is a distribution that takes $\underline{x}_t^{*}$ with probability $1-q_t^{*}$ and takes $\overline{x}_t^{*}$ with probability $q_t^{*}$, is an optimal solution to \eqref{MBC-P}. Hence, it suffices to show that the problem \eqref{MBCF-app-eq1} -- \eqref{MBCF-app-eq3c} employs an optimal solution.

We first show that while solving \eqref{MBCF-app-eq1} -- \eqref{MBCF-app-eq3c}, we can restrict ourselves to $m_0 \in [0,\frac{1}{p_0}\overline{x}]$ and $m_1 \in [0,\frac{1}{p_1}\overline{x}]$. Suppose $m_0 > \frac{1}{p_0}\overline{x}$ at a candidate solution. The principal's profit from the type-1 agent is bounded above by $\overline{x}$, while the profit from the type-0 agent is bounded above by $\overline{x} - m_0$. Hence, the principal's expected profit from this candidate solution is no larger than $p_0(\overline{x}-m_0) + p_1\overline{x} = \overline{x} - p_0\cdot m_0 < 0$. In this case, the principal would strictly prefer to offer zero contracts to both types of agents and get an expected profit of $\underline{x}\geq 0$. Hence, any candidate solution with $m_0 > \frac{1}{p_0}\overline{x}$ can not be optimal. Similarly, any solution with $m_1 > \frac{1}{p_1}\overline{x}$ is also not optimal. We can replace \eqref{MBCF-app-eq1} with:
\begin{equation}
\alpha_t\in [\underline{c},\overline{c}],\quad m_t\in [0,\frac{1}{p_t}\overline{x}],\quad \underline{x}_t,\overline{x}_t\in [\underline{x},\overline{x}],\quad q_t\in [0,1],\quad \forall t \in \{0,1\},
\label{MBCF-app-eq4}
\end{equation}
and establish the existence of an optimal solution to the problem described by \eqref{MBCF-app-eq2} -- \eqref{MBCF-app-eq4}. Since the objective function \eqref{MBCF-app-eq2} is continuous in the variables considered. It suffices to show that the set outlined in \eqref{MBCF-app-eq3a} -- \eqref{MBCF-app-eq4} is closed, and thereby compact.

The RHS of \eqref{MBCF-app-eq3c} is given by
\begin{equation*}
\max_{\alpha\in[\underline{c},\overline{c}]}\big[\widehat{s_0\circ \theta}(\alpha) - K_1(\alpha)\big] = \max_{\alpha \in [c(\underline{x}_0),c(\overline{x}_0)]}\big[K_0'(\alpha_0)\cdot \alpha - K_1(\alpha) + m_0\big]. 
\end{equation*}
Since $K_0'(\alpha_0)\cdot \alpha - K_1(\alpha)+m_0$ is continuous in $\alpha_0$, $m_0$ and $\alpha$, and the correspondence $(\underline{x}_0,\overline{x}_0)\Rightarrow [c(\underline{x}_0),c(\overline{x}_0)]$ is continuous and has non-empty compact values. By Berge's Maximum Theorem (Th 17.31, \cite{guide2006}), the RHS of \eqref{MBCF-app-eq3c} is continuous in $\alpha_0$, $m_0$, $\underline{x}_0$ and $\overline{x}_0$.

Consider the RHS of \eqref{MBCF-app-eq3b}, for $\underline{x}_1 > \underline{x}$ ($\Longleftrightarrow c(\underline{x}_1) > \underline{c}$), it is given by
\begin{equation*}
\begin{aligned}
    \max_{\alpha\in [\underline{c},\overline{c}]}\big[\widehat{s_1\circ\theta}(\alpha) - K_0(\alpha)\big] = \max\Big\{&\max_{\alpha\in[\underline{c},c(\underline{x}_1)]}\big[\big(K_1'(\alpha_1)+ \frac{m_1}{c(\underline{x}_1)}\big)\cdot \alpha - K_0(\alpha)\big],\\
    &\max_{\alpha\in [c(\underline{x}_1),c(\overline{x}_1)]}\big[K_1'(\alpha_1)\cdot \alpha - K_0(\alpha)+m_1\big]\Big\},
\end{aligned}
\end{equation*}
and for $\underline{x}_1 = \underline{x}$ ($\Longleftrightarrow c(\underline{x}_1) = \underline{c}$), it is given by 
\begin{equation*}
 \max_{\alpha\in [\underline{c},\overline{c}]}\big[\widehat{s_1\circ\theta}(\alpha) - K_0(\alpha)\big] = \max\Big\{m_1,\max_{\alpha\in [c(\underline{x}_1),c(\overline{x}_1)]}\big[K_1'(\alpha_1)\cdot \alpha - K_0(\alpha)+m_1\big] \Big\}.
\end{equation*}
Similarly to the previous arguments, we can establish that $\max_{\alpha\in [c(\underline{x}_1),c(\overline{x}_1)]}\big[K_1'(\alpha_1)\cdot \alpha - K_0(\alpha)+m_1\big]$ is continuous in $\alpha_1$, $m_1$, $\underline{x}_1$ and $\overline{x}_1$, and that $\max_{\alpha\in[\underline{c},c(\underline{x}_1)]}\big[\big(K_1'(\alpha_1)+ \frac{m_1}{c(\underline{x}_1)}\big)\cdot \alpha - K_0(\alpha)\big]$ is continuous in $\alpha_1$, $m_1$, $\underline{x}_1$ with $\underline{x}_1 > \underline{x}$. To establish the continuity of the RHS of \eqref{MBCF-app-eq3b}, we still need to show that $\lim_{n\to \infty}
\Big\{\max_{\alpha\in[\underline{c},c(\underline{x}_1^n)]}\big[\big(K_1'(\alpha_1^n)+ \frac{m_1^n}{c(\underline{x}_1^n)}\big)\cdot \alpha - K_0(\alpha)\big]\Big\} = m
_1$ for any sequence of $\{(\alpha_1^n, m_1^n,\underline{x}_1^n)\}$ ($\alpha_1^n \in [\underline{c},\overline{c}]$, $m_1^n \in [0,\frac{1}{p_1}\overline{x}]$, $\underline{x}_1^n \in (\underline{x},\overline{x}]$ for all $n\in \mathbb{N}$) that converges to $(\alpha_1,m_1,\underline{x}_1)$ with $\underline{x}_1 = \underline{x}$. 

Take a sequence $\{(\alpha_1^n, m_1^n,\underline{x}_1^n)\}$ as described above. For all $n\in \mathbb{N}$,
\begin{equation*}
\begin{aligned}
    &\big(K_1'(\alpha_1^n)+ \frac{m_1^n}{c(\underline{x}_1^n)}\big)\cdot \alpha - K_0(\alpha) \leq m_1^n + K_1'(\alpha_1^n)\cdot \alpha - K_0(\alpha),\quad \forall \alpha \in [\underline{c}, c(\underline{x}_1^n)],\\
   \Longrightarrow &\max_{\alpha\in [\underline{c},c(\underline{x}_1^n)]}\big[\big(K_1'(\alpha_1^n)+ \frac{m_1^n}{c(\underline{x}_1^n)}\big)\cdot \alpha - K_0(\alpha)\big]\leq m_1^n + \max_{\alpha\in [\underline{c},c(\underline{x}_1^n)]}\big[ K_1'(\alpha_1^n)\cdot \alpha - K_0(\alpha)\big],
\end{aligned}
\end{equation*}
and consequently, we have that
\begin{equation}
\begin{aligned}
    \lim_{n\to \infty}\Big\{\max_{\alpha\in [\underline{c},c(\underline{x}_1^n)]}\big[\big(K_1'(\alpha_1^n)+ \frac{m_1^n}{c(\underline{x}_1^n)}\big)\cdot \alpha - K_0(\alpha)\big]\Big\} &\leq \lim_{n\to\infty} m_1^n + \lim_{n\to \infty} \Big\{\max_{\alpha\in [\underline{c},c(\underline{x}_1^n)]}\big[ K_1'(\alpha_1^n)\cdot \alpha - K_0(\alpha)\big]\Big\}\\
    &= m_1 + K_1'(\alpha_1)\cdot \underline{c} - K_0(\underline{c})
    = m_1,
\end{aligned}
\label{MBCF-app-eq5}
\end{equation}

On the other hand, for all $n\in \mathbb{N}$,
\begin{equation*}
\max_{\alpha\in [\underline{c},c(\underline{x}_1^n)]}\big[\big(K_1'(\alpha_1^n)+ \frac{m_1^n}{c(\underline{x}_1^n)}\big)\cdot \alpha - K_0(\alpha)\big] \geq m_1^n + K_1'(\alpha_1^{n})\cdot c(\underline{x}_1^n) - K_0(c(\underline{x}_1^n)).
\end{equation*}

Hence, we have that
\begin{equation}
\begin{aligned}
\lim_{n\to\infty}\Big\{\max_{\alpha\in [\underline{c},c(\underline{x}_1^n)]}\big[\big(K_1'(\alpha_1^n)+ \frac{m_1^n}{c(\underline{x}_1^n)}\big)\cdot \alpha - K_0(\alpha)\big]\Big\}&\geq \lim_{n\to\infty}m_1^n + \lim_{n\to\infty}\Big\{ K_1'(\alpha_1^{n})\cdot c(\underline{x}_1^n) - K_0(c(\underline{x}_1^n))\Big\}\\
&= m_1.
\end{aligned}
\label{MBCF-app-eq6}    
\end{equation}

Combining \eqref{MBCF-app-eq5} and \eqref{MBCF-app-eq6}, we have that  $\lim_{n\to \infty}
\Big\{\max_{\alpha\in[\underline{c},c(\underline{x}_1^n)]}\big[\big(K_1'(\alpha_1^n)+ \frac{m_1^n}{c(\underline{x}_1^n)}\big)\cdot \alpha - K_0(\alpha)\big]\Big\} = m
_1$. Therefore, the RHS of \eqref{MBCF-app-eq3b} is continuous in $\alpha_1$, $m_1$, $\underline{x}_1$, $\overline{x}_1$. And the set outlined in \eqref{MBCF-app-eq3a} -- \eqref{MBCF-app-eq4} is a closed set. 
\end{proof}

\section{Supplementary Results and Proofs for Section \ref{MBCF-imp}}
\label{MBCF-app-imp}
\textbf{Proof of Lemma \ref{IMP-lm1}:} Given a pair of aggregate efforts $a =(\alpha_0,\alpha_1)$ satisfying that $\alpha_0\leq \alpha_1$, the payment-minimizing menu that implements $(\alpha_0,\alpha_1)$ solves the following problem:
\begin{equation}
\begin{aligned}
\min_{m,R,\mu}\quad&p_0\big[K_0'(\alpha_0)\cdot \alpha_0 + m_0\big] + p_1\big[K_1'(\alpha_1)\cdot \alpha_1 + m_1\big] \\
s.t.\quad&(a,m,R,\mu)\in \mathcal{F}.
\end{aligned}
\label{IMP-eq1}
\end{equation}

We first show that, if a menu $(a,m,R,\mu)$ is feasible, the menu $(a,m,R',\mu')$ is also feasible with $R' = (\{\theta(\alpha_0)\},\{\theta(\alpha_1)\})$ and $\mu' = (\delta_{\theta(\alpha_0)},\delta_{\theta(\alpha_1)})$. 

It straightforwardly holds that $\text{supp}(\delta_{\theta(\alpha_0)}) \subseteq \{\theta(\alpha_0)\}$, $\text{supp}(\delta_{\theta(\alpha_1)})\subseteq \{\theta(\alpha_1)\}$, $\int c(x)\delta_{\theta(\alpha_0)}(dx) = \alpha_0$, and $\int c(x)\delta_{\theta(\alpha_1)}(dx) = \alpha_1$. Take contracts $s_0'(\cdot)$, $s_1'(\cdot)$ such that $s_0'\circ \theta(\alpha) = \mathbb{1}_{\{\alpha = \alpha_0\}}\big[K_0'(\alpha_0)\cdot \alpha + m_0\big]$ and $s_1'\circ \theta(\alpha) = \mathbb{1}_{\{\alpha = \alpha_1\}}\big[K_1'(\alpha_1)\cdot \alpha + m_1\big]$. We only need to check whether it satisfies incentive compatibility. 

Take $s_0(\cdot)$ and $s_1(\cdot)$ as the contracts pinned down by the menu $(a,m,R,\mu)$, i.e., $s_0\circ \theta(\alpha)=\mathbb{1}_{\{\alpha \in c(R_0)\}}\big[K_0'(\alpha_0)\cdot \alpha + m_0\big]$, $s_1\circ \theta(\alpha)=\mathbb{1}_{\{\alpha \in c(R_1)\}}\big[K_1'(\alpha_1)\cdot \alpha + m_1\big]$. It holds that $s_0'\circ \theta(\alpha)\leq s_0\circ \theta(\alpha)$ and $s_1'\circ \theta(\alpha)\leq s_1\circ \theta(\alpha)$ for all $\alpha \in [\underline{c},\overline{c}]$. Since the menu $(a,m,R,\mu)$ satisfies incentive compatibility, we have that 
\begin{equation*}
\begin{aligned}
\max_{\alpha \in [\underline{c},\overline{c}]} \big[ \widehat{s_1'\circ \theta}(\alpha) - K_0(\alpha))\big] &\leq \max_{\alpha \in [\underline{c},\overline{c}]} \big[ \widehat{s_1\circ \theta}(\alpha) - K_0(\alpha))\big]\leq K_0'(\alpha_0)\cdot \alpha_0 - K_0(\alpha_0)+m_0,\\
\max_{\alpha \in [\underline{c},\overline{c}]} \big[ \widehat{s_0'\circ \theta}(\alpha) - K_1(\alpha))\big] &\leq \max_{\alpha \in [\underline{c},\overline{c}]} \big[ \widehat{s_0\circ \theta}(\alpha) - K_1(\alpha))\big]\leq K_1'(\alpha_1)\cdot \alpha_1 - K_1(\alpha_1) + m_1.
\end{aligned}
\end{equation*}

Hence, the menu $(a,m,R',\mu')$ also satisfies incentive compatibility and is thus feasible to problem \eqref{IMP-eq1}. Since $(a,m,R',\mu')$ leads to the same expected payment as $(a,m,R,\mu)$, we can w.l.o.g. focus on menus where the contract ranges and output distributions are fixed by $R'$ and $\mu'$ when solving \eqref{IMP-eq1}. We reformulate \eqref{IMP-eq1} as the following problem\footnote{The RHS of \eqref{IMP-eq2b} is equal to $m_1$ when $\alpha_1 = \underline{c}$. Similarly, the RHS of \eqref{IMP-eq2c} is equal to $m_0$ if $\alpha_0 = \underline{c}$.}:
\begin{subequations}
\begin{align}
  \min_{m_0,m_1}\quad&p_0\big[K_0'(\alpha_0)\cdot \alpha_0 + m_0\big] + p_1\big[K_1'(\alpha_1)\cdot \alpha_1 + m_1\big]\label{IMP-eq2a}\\
  s.t. \quad& K_0'(\alpha_0)\cdot \alpha_0 -K_0(\alpha_0)+m_0\geq \max_{\alpha\in[\underline{c},\alpha_1]}\big[\big(\frac{K_1'(\alpha_1)\cdot \alpha_1+m_1}{\alpha_1}\big)\cdot \alpha- K_0(\alpha)\big]\label{IMP-eq2b}\\
  &K_1'(\alpha_1)\cdot\alpha_1 - K_1(\alpha_1)+m_1 \geq \max_{\alpha\in [\underline{c},\alpha_0]}\big[\big(\frac{K_0'(\alpha_0)\cdot\alpha_0 +m_0}{\alpha_0}\big)\cdot \alpha - K_1(\alpha)\big]\label{IMP-eq2c}\\
  &m_0\geq 0,\quad m_1\geq 0.
\end{align}
\end{subequations}

Next, we prove the result by splitting into two cases based on the comparison between $K_0'(\alpha_0)$ and $K_1'(\alpha_1)$.\\
\\
\textbf{Case 1:} $K_0'(\alpha_0)\leq K_1'(\alpha_1)$.

We take $m_0'$ and $m_1'$ as stated and verify that $(m_0',m_1')$ is feasible to problem \eqref{IMP-eq2a}. We directly have that $m_1' \geq 0$ and $m_0' = \big[K_1'(\alpha_1)\cdot k_0^{-1}\big(K_1'(\alpha_1)\big) - K_0\big(k_0^{-1}\big(K_1'(\alpha_1)\big)\big)\big]-\big[K_0'(\alpha_0)\cdot \alpha_0 - K_0(\alpha_0)\big]\geq 0$, following from $K_0'(\alpha_0)\leq K_1'(\alpha_1)$. Furthermore, we have that $(m_0',m_1')$ satisfies \eqref{IMP-eq2b} if and only if 
\begin{equation}
\begin{aligned}
    m_0' &\geq \max_{\alpha\in[\underline{c},\alpha_1]}\big[K_1'(\alpha_1)\cdot \alpha - K_0(\alpha)\big] -\big[K_0'(\alpha_0)\cdot \alpha_0 - K_0(\alpha_0)\big]\\
    &= \big[K_1'(\alpha_1)\cdot k_0^{-1}\big(K_1'(\alpha_1)\big) - K_0\big(k_0^{-1}\big(K_1'(\alpha_1)\big)\big)\big]-\big[K_0'(\alpha_0)\cdot \alpha_0 - K_0(\alpha_0)\big],
\label{IMP-eq3}
\end{aligned}
\end{equation}
where the equality follows from that $\alpha_0\leq k_0^{-1}\big(K_1'(\alpha_1)\big)\leq \alpha_1$ (Assumption \ref{OM-as1} implies $k_0^{-1}\big(K_1'(\alpha_1)\big)\leq \alpha_1$). By the definition of $m_0'$, we have that \eqref{IMP-eq3} holds naturally. 

Next, we consider the \eqref{IMP-eq2c} constraint. We have that $(m_0',m_1')$ satisfies \eqref{IMP-eq2c} if and only if 
\begin{equation}
\max_{\alpha\in [\underline{c},\alpha_0]}\big[\big(\frac{K_0'(\alpha_0)\cdot\alpha_0 +m_0'}{\alpha_0}\big)\cdot \alpha - K_1(\alpha)\big]\leq K_1'(\alpha_1)\cdot\alpha_1 - K_1(\alpha_1). 
\label{IMP-eq4}
\end{equation}

Define $\alpha_0^1:= k_0^{-1}\big(K_1'(\alpha_1)\big)$. By the definition of $m_0'$, we have that
\begin{equation}
\begin{aligned}
m_0' &= \big[K_1'(\alpha_1)\cdot \alpha_0^1 - K_0(\alpha_0^1)]-\big[K_0'(\alpha_0)\cdot \alpha_0 - K_0(\alpha_0)\big] \\
&= \big[K_0'(\alpha_0^1)\cdot \alpha_0^1 - K_0(\alpha_0^1)]-\big[K_0'(\alpha_0)\cdot \alpha_0 - K_0(\alpha_0)\big] \\
&= \int_{\alpha_0}^{\alpha_0^1}\alpha d K_0'(\alpha) \geq \int_{\alpha_0}^{\alpha_0^1}\alpha_0dK_0'(\alpha)= \big[K_1'(\alpha_1)- K_0'(\alpha_0)\big]\cdot \alpha_0,
\end{aligned}
\label{IMP-eq5}
\end{equation}
where the third equality follows from integration by parts. 

When $\alpha_0 > \underline{c}$, \eqref{IMP-eq5} implies that $\frac{K_0'(\alpha_0)\cdot \alpha_0 + m_0'}{\alpha_0} \geq K_1'(\alpha_1)\geq K_1'(\alpha)$ for all $\alpha\in[\underline{c},\alpha_0]$. It indicates that the objective on the LHS of \eqref{IMP-eq4} is maximized by taking $\alpha = \alpha_0$. When $\alpha_0 = \underline{c}$, the LHS of \eqref{IMP-eq4} is equal to $m_0'$. In conclusion, regardless of the value of $\alpha_0$, the LHS of \eqref{IMP-eq4} is equal to $K_0'(\alpha_0)\cdot\alpha_0-K_1(\alpha_0)+m_0'$. Hence, \eqref{IMP-eq4} is equivalent to
\begin{equation}
    m_0' \leq \big[K_1'(\alpha_1)\cdot\alpha_1 - K_1(\alpha_1) \big] - \big[K_0'(\alpha_0)\cdot\alpha_0-K_1(\alpha_0)\big]
\label{IMP-eq6}
\end{equation}

Following the same reasoning as in \eqref{IMP-eq5}, the LHS of \eqref{IMP-eq6} satisfies that
\begin{equation}
    m_0' = \int_{\alpha_0}^{\alpha_0^1}\alpha dK_0'(\alpha) = \int_{K_0'(\alpha_0)}^{K_1'(\alpha_1)}k_0^{-1}(k)dk. 
\label{IMP-eq7}
\end{equation}
Define $\alpha_1^0:= k_1^{-1}(K_0'(\alpha_0))\in [\alpha_0,\alpha_1]$. The RHS of \eqref{IMP-eq6} satisfies that
\begin{equation}
\begin{aligned}
    &\big[K_1'(\alpha_1)\cdot\alpha_1 - K_1(\alpha_1) \big] - \big[K_0'(\alpha_0)\cdot\alpha_0-K_1(\alpha_0)\big] \\
    \geq & \big[K_1'(\alpha_1)\cdot\alpha_1 - K_1(\alpha_1) \big] - \big[ K_0'(\alpha_0)\cdot \alpha_1^{0} - K_1(\alpha_1^0) \big] \\
    =& \int_{\alpha_1^0}^{\alpha_1}\alpha d K_1'(\alpha) = \int_{K_0'(\alpha_0)}^{K_1'(\alpha_1)}k_1^{-1}(k)dk,
\end{aligned}
\label{IMP-eq8}
\end{equation}
where the inequality follows from that $K_0'(\alpha_0)\cdot \alpha_0 - K_1(\alpha_0)\leq K_0'(\alpha_0)\cdot \alpha_1^0 - K_1(\alpha_1^0)$. From Assumption \ref{OM-as1}, it holds that $k_1^{-1}(k)\geq k_0^{-1}(k)$ for all $k\in [K_0'(\alpha_0),K_1'(\alpha_1)]$. Combining \eqref{IMP-eq7} and \eqref{IMP-eq8}, we can show that \eqref{IMP-eq6} is satisfied.

So far, we have shown that $(m_0',m_1')$ is feasible to problem \eqref{IMP-eq2a}. Take any $(m_0,m_1)$ that is feasible to \eqref{IMP-eq2a}. We have that $m_1\geq 0 = m_1'$, and $m_0 \geq \big[K_1'(\alpha_1)\cdot k_0^{-1}\big(K_1'(\alpha_1)\big) - K_0\big(k_0^{-1}\big(K_1'(\alpha_1)\big)\big)\big]-\big[K_0'(\alpha_0)\cdot \alpha_0 - K_0(\alpha_0)\big] = m_0'$, following from \eqref{IMP-eq2b}. Hence, $(m_0',m_1')$ optimally solves the problem. \\
\\
\textbf{Case 2:} $K_0'(\alpha_0)> K_1'(\alpha_1)$. 

We take $m_0'$ and $m_1'$ as stated and verify that $(m_0',m_1')$ is feasible to problem \eqref{IMP-eq2a}. It is straightforward to see that $m_0'\geq 0$ and $m_1'\geq 0$. We can show that $(m_0',m_1')$ satisfies \eqref{IMP-eq2c} if and only if  
\begin{equation}
\begin{aligned}
    m_1'&\geq \max_{\alpha\in [\underline{c},\alpha_0]}\big[K_0'(\alpha_0)\cdot \alpha - K_1(\alpha)\big] - \big[K_1'(\alpha_1)\cdot\alpha_1 - K_1(\alpha_1)\big]\\
    &= \big[K_0'(\alpha_0)\cdot \alpha_0 - K_1(\alpha_0)\big] - \big[K_1'(\alpha_1)\cdot\alpha_1 - K_1(\alpha_1)\big],
\end{aligned}
\label{IMP-eq9}
\end{equation}
where the equality follows from that $K_0'(\alpha_0)>K_1'(\alpha_1)\geq K_1'(\alpha_0)$. By the definition of $m_1'$, \eqref{IMP-eq9} is naturally satisfied. 

Next, we consider the \eqref{IMP-eq2b} constraint. We have that $(m_0',m_1')$ satisfies \eqref{IMP-eq2b} if and only if
\begin{equation}
\max_{\alpha\in[\underline{c},\alpha_1]}\big[\big(\frac{K_1'(\alpha_1)\cdot \alpha_1+m_1'}{\alpha_1}\big)\cdot \alpha- K_0(\alpha)\big] \leq K_0'(\alpha_0)\cdot \alpha_0 -K_0(\alpha_0).
\label{IMP-eq10}
\end{equation}

We prove that \eqref{IMP-eq10} is satisfied for the following two cases based on the value of $\alpha_0$\\
\\
\textbf{Case 2a:} $\alpha_0 > \underline{c}$.

In this case, we have that $\alpha_1 \geq  \alpha_0> \underline{c}$. We can show that \eqref{IMP-eq10} holds if and only if
\begin{equation}
\frac{K_1'(\alpha_1)\cdot\alpha_1 + m_1'}{\alpha_1} \leq K_0'(\alpha_0) \Longleftrightarrow m_1'\leq \big[K_0'(\alpha_0) - K_1'(\alpha_1)\big]\cdot \alpha_1. 
\label{IMP-eq11}
\end{equation}

The sufficiency is straightforward. The necessity follows from that if $\frac{K_1'(\alpha_1)\cdot\alpha_1 + m_1'}{\alpha_1} > K_0'(\alpha_0)$, then we have that $\max_{\alpha\in[\underline{c},\alpha_1]}\big[\big(\frac{K_1'(\alpha_1)\cdot \alpha_1+m_1'}{\alpha_1}\big)\cdot \alpha- K_0(\alpha)\big] \geq \big(\frac{K_1'(\alpha_1)\cdot \alpha_1+m_1'}{\alpha_1}\big)\cdot \alpha_0- K_0(\alpha_0) > K_0'(\alpha_0)\cdot \alpha_0 - K_0(\alpha_0)$, thus violating \eqref{IMP-eq10}. 

From the assumption that $K_0'(\alpha_0) > K_1'(\alpha_1)$, it holds that $\big[K_0'(\alpha_0)-K_1'(\alpha_1)\big]\cdot \alpha_1\geq 0$. We then compare $\big[K_0'(\alpha_0)-K_1'(\alpha_1)\big]\cdot \alpha_1\geq 0$ with $\big[K_0'(\alpha_0)\cdot \alpha_0 - K_1(\alpha_0)\big] - \big[K_1'(\alpha_1)\cdot\alpha_1 - K_1(\alpha_1)\big]$. We have that
\begin{equation*}
\begin{aligned}
&\big[K_0'(\alpha_0)\cdot \alpha_0 - K_1(\alpha_0)\big] - \big[K_1'(\alpha_1)\cdot\alpha_1 - K_1(\alpha_1)\big]\\
=&\big[K_0'(\alpha_0) - K_1'(\alpha_1)\big]\cdot \alpha_0 + \big[K_1(\alpha_1) - K_1(\alpha_0)\big]-K_1'(\alpha_1)\cdot\big[\alpha_1 - \alpha_0\big]\\
=&\big[K_0'(\alpha_0) - K_1'(\alpha_1)\big]\cdot \alpha_0 + \int_{\alpha_0}^{\alpha_1}\big[K_1'(\alpha)- K_1'(\alpha_1)\big]d\alpha\\
\leq& \big[K_0'(\alpha_0) - K_1'(\alpha_1)\big]\cdot \alpha_0 \leq  \big[K_0'(\alpha_0) - K_1'(\alpha_1)\big]\cdot \alpha_1,
\end{aligned}
\end{equation*}
where the last line follows from $K_1'(\alpha)\leq K_1'(\alpha_1)$ for all $\alpha\in[\alpha_0,\alpha_1]$ and from $\alpha_0\leq \alpha_1$. It implies that $m_1' = \max\Big\{0, \big[K_0'(\alpha_0)\cdot \alpha_0 - K_1(\alpha_0)\big] - \big[K_1'(\alpha_1)\cdot \alpha_1 - K_1(\alpha_1)\big]\Big\} \leq [K_0'(\alpha_0) - K_1'(\alpha_1)\big]\cdot \alpha_1$. Hence, we've shown that \eqref{IMP-eq11}, as well as \eqref{IMP-eq10}, is satisfied. \\
\\
\textbf{Case 2b:} $\alpha_0 = \underline{c}$. 

In this case, we have that $m_1'= 0$. If $\alpha_1 > \underline{c}$, the LHS of \eqref{IMP-eq10} is maximized by taking $\alpha = \underline{c}$, which follows from that $K_1'(\alpha_1)<K_0'(\alpha_0) = K_0'(\underline{c})$, and is thus equal to 0. If $\alpha_1 = \underline{c}$, the LHS of \eqref{IMP-eq10} is also equal to 0. Since $K_0'(\alpha_0)\cdot \alpha_0 - K_0(\alpha_0)=0$, \eqref{IMP-eq10} is satisfied. \\
\\
In conclusion, $(m_0',m_1')$ is feasible to problem \eqref{IMP-eq2a}. Using the same reasoning as in Case 1, we can show that $(m_0',m_1')$ provides a minimized expected payment among all $(m_0,m_1)$ feasible to \eqref{IMP-eq2a}. $\hfill\square$
\subsection{Implementation of Output Distributions}
\label{APP-IMP-OD}
We first provide the exact expression of the function $\check{\alpha}_0(\cdot)$ that appears in Theorem \ref{IMP-pp2} and then present the proof of the theorem. Given a pair of aggregate efforts $(\alpha_0,\alpha_1)$ that satisfies $\alpha_0\leq \alpha_1$ and $K_0'(\alpha_0) > K_1'(\alpha_1)$. Define
\begin{equation}
\alpha_0^1: =\left\{
\begin{array}{cl}
  k_0^{-1}\big(K_1'(\alpha_1)\big)   &\text{ if }K_1'(\alpha_1)\geq K_0'(\underline{c}), \\
  \underline{c}   &\text{ if } K_1'(\alpha_1)< K_0'(\underline{c}),
\end{array}
\right.
\label{IMP-eq12}
\end{equation}
and
\begin{equation}
\alpha_1^0: =\left\{
\begin{array}{cl}
  k_1^{-1}\big(K_0'(\alpha_0)\big)   &\text{ if }K_0'(\alpha_0)\leq K_1'(\overline{c}), \\
  \overline{c}   &\text{ if } K_0'(\alpha_0)> K_1'(\overline{c}).
\end{array}
\right.
\label{IMP-eq13}
\end{equation}

Next, we define a cutoff $\check{\alpha}_1 > \alpha_0^1$ as follows
\begin{equation}
\begin{aligned}
    \check{\alpha}_1 := \inf\{\underline{\alpha}_1\in[\alpha_0^1,\overline{c}]\mid \big[K_0'(\alpha_0)\cdot \alpha_0 - K_0(\alpha_0)\big]-\big[K_1'(\alpha_1)\cdot \underline{\alpha}_1 - K_0(\underline{\alpha}_1)\big] \atop
    \qquad\qquad\qquad\qquad\qquad\quad\quad\geq\big[K_0'(\alpha_0)\cdot \alpha_1^0 - K_1(\alpha_1^0)\big]- \big[K_1'(\alpha_1)\cdot \alpha_1-K_1(\alpha_1)\big] \}.
\end{aligned}
\label{IMP-eq14}
\end{equation}

If the set on the RHS of \eqref{IMP-eq14} is empty, then $\check{\alpha}_1$ is defined as $+\infty$. Otherwise, we can show that $\check{\alpha}_1 > \alpha_0^1$ based on Lemma \ref{APP-IMP-lm1}.
\begin{lemma}
Take a pair of aggregate efforts $(\alpha_0,\alpha_1)$ satisfying  $\alpha_0\leq \alpha_1$ and $K_0'(\alpha_0) > K_1'(\alpha_1)$. Take $\alpha_0^1$, $\alpha_1^0$ as stated in \eqref{IMP-eq12} and \eqref{IMP-eq13}. It holds that
\begin{equation}
\begin{aligned}
&\big[K_0'(\alpha_0)\cdot \alpha_0 - K_0(\alpha_0)\big]-\big[K_1'(\alpha_1)\cdot \alpha_0^1 - K_0(\alpha_0^1)\big]\\
<&\big[K_0'(\alpha_0)\cdot \alpha_1^0 - K_1(\alpha_1^0)\big]- \big[K_1'(\alpha_1)\cdot \alpha_1-K_1(\alpha_1)\big]. 
\end{aligned}
\label{IMP-eq15}
\end{equation}
\label{APP-IMP-lm1}
\end{lemma}
\begin{proof}
According to definition \eqref{IMP-eq12}, we have that $\alpha_0^1 \leq\alpha_0$ and $K_0'(\alpha)\geq K_1'(\alpha_1)$ for all $\alpha \in [\alpha_0^1,\alpha_0]$. Similarly, according to definition \eqref{IMP-eq13}, we have that $\alpha_1^0 \geq \alpha_1$ and $K_1'(\alpha)\leq K_0'(\alpha_0)$ for all $\alpha \in [\alpha_1, \alpha_1^0]$. The LHS of \eqref{IMP-eq15} satisfies that
\begin{equation}
\begin{aligned}
&\big[K_0'(\alpha_0)\cdot \alpha_0 - K_0(\alpha_0)\big]-\big[K_1'(\alpha_1)\cdot \alpha_0^1 - K_0(\alpha_0^1)\big]\\ \leq &\big[K_0'(\alpha_0)\cdot \alpha_0 - K_0(\alpha_0)\big] - \big[K_1'(\alpha_1)\cdot \alpha_0 - K_0(\alpha_0)\big] = \big[K_0'(\alpha_0)-K_1'(\alpha_1)\big]\cdot \alpha_0,
\end{aligned}
\label{IMP-eq16}
\end{equation}
where the inequality follows from that $K_1'(\alpha_1)\cdot \alpha_0^1 - K_0(\alpha_0^1) \geq K_1'(\alpha_1)\cdot \alpha_0 - K_0(\alpha_0)$ and is strict if $\alpha_0^1 < \alpha_0$. 

Consider the RHS of \eqref{IMP-eq15}, we have 
\begin{equation}
\begin{aligned}
&\big[K_0'(\alpha_0)\cdot \alpha_1^0 - K_1(\alpha_1^0)\big]- \big[K_1'(\alpha_1)\cdot \alpha_1-K_1(\alpha_1)\big] \\
\geq & \big[ K_0'(\alpha_0)\cdot \alpha_1 - K_1(\alpha_1)\big] - \big[K_1'(\alpha_1)\cdot \alpha_1-K_1(\alpha_1)\big]= \big[K_0'(\alpha_0) - K_1'(\alpha_1)\big]\cdot \alpha_1,
\end{aligned}
\label{IMP-eq17}
\end{equation}
where the inequality follows from that $K_0'(\alpha_0)\cdot \alpha_1^0 - K_1(\alpha_1^0)\geq  K_0'(\alpha_0)\cdot \alpha_1 - K_1(\alpha_1)$ and is strict if $\alpha_1^0 > \alpha_1$. 

Since $\big[K_0'(\alpha_0) - K_1'(\alpha_1)\big]\cdot \alpha_1 \geq \big[K_0'(\alpha_0)-K_1'(\alpha_1)\big]\cdot \alpha_0$, combining \eqref{IMP-eq16} and \eqref{IMP-eq17}, we have that 
\begin{equation*}
\begin{aligned}
&\big[K_0'(\alpha_0)\cdot \alpha_0 - K_0(\alpha_0)\big]-\big[K_1'(\alpha_1)\cdot \alpha_0^1 - K_0(\alpha_0^1)\big]\leq \big[K_0'(\alpha_0)- K_1'(\alpha_1)\big]\cdot \alpha_0\\
\leq& \big[K_0'(\alpha_0) - K_1'(\alpha_1)\big]\cdot \alpha_1 \leq \big[K_0'(\alpha_0)\cdot \alpha_1^0 - K_1(\alpha_1^0)\big]- \big[K_1'(\alpha_1)\cdot \alpha_1-K_1(\alpha_1)\big].
\end{aligned}
\end{equation*}

Suppose we have $\big[K_0'(\alpha_0)\cdot \alpha_0 - K_0(\alpha_0)\big]-\big[K_1'(\alpha_1)\cdot \alpha_0^1 - K_0(\alpha_0^1)\big] = \big[K_0'(\alpha_0)\cdot \alpha_1^0 - K_1(\alpha_1^0)\big]- \big[K_1'(\alpha_1)\cdot \alpha_1-K_1(\alpha_1)\big]$, it can only happen if all the inequalities in the above expressions hold with equality. It implies that $\alpha_0^1 = \alpha_0$, which only appears when $\alpha_0 = \underline{c}$. Similarly, it also implies that $\alpha_1 = \overline{c}$. However, then it holds that $\big[K_0'(\alpha_0) - K_1'(\alpha_1)]\cdot \alpha_1 > \big[K_0'(\alpha_0) - K_1'(\alpha_1)\big]\cdot \alpha_0$, which leads to a contradiction. Hence, we have shown \eqref{IMP-eq15}. 
\end{proof}
\noindent\textbf{Definition of $\check{\alpha}_0(\cdot)$:} The function $\check{\alpha}_0:[\underline{c},\overline{c}]\to [\underline{c},\overline{c}]$ takes the following form:
\begin{enumerate}
    \item For $\underline{\alpha}_1 < \alpha_0^1$, define $\check{\alpha}_0(\underline{\alpha}_1):= \check{\alpha}_0$ where $\check{\alpha}_0\in (\underline{c}, \alpha_1^0)$ is pinned down by\footnote{The existence of $\check{\alpha}_0\in (\underline{c},\alpha_1^0)$ in the first case and the existence of $\check{\alpha}_0(\underline{\alpha}_1)\in [\check{\alpha}_0,\alpha_1^0)$ in the second case are both guaranteed by Lemma \ref{APP-IMP-lm1}.}
    \begin{equation*}
    \begin{aligned}
        &\big[K_0'(\alpha_0)\cdot \check{\alpha}_0 - K_1(\check{\alpha}_0)\big] - \big[K_1'(\alpha_1)\cdot \alpha_1 - K_1(\alpha_1)\big]\\
        =& \big[K_0'(\alpha_0)\cdot \alpha_0 - K_0(\alpha_0)\big]-\big[K_1'(\alpha_1)\cdot \alpha_0^1 - K_0(\alpha_0^1)\big].
    \end{aligned}
    \end{equation*}
    \item For $ \alpha_0^1 \leq \underline{\alpha}_1 < \check{\alpha}_1$, $\check{\alpha}_0(\underline{\alpha}_1)\in [\check{\alpha}_0,\alpha_1^0)$ is pinned down by
    \begin{equation*}
    \begin{aligned}
        &\big[K_0'(\alpha_0)\cdot \check{\alpha}_0(\underline{\alpha}_1) - K_1\big(\check{\alpha}_0(\underline{\alpha}_1)\big)\big] - \big[K_1'(\alpha_1)\cdot \alpha_1 - K_1(\alpha_1)\big]\\
        =& \big[K_0'(\alpha_0)\cdot \alpha_0 - K_0(\alpha_0)\big]-\big[K_1'(\alpha_1)\cdot \underline{\alpha}_1 - K_0(\underline{\alpha}_1)\big].
    \end{aligned}
    \end{equation*}
    \item For $\underline{\alpha}_1\geq \check{\alpha}_1$, define $\check{\alpha}_0(\underline{\alpha}_1):= \overline{c}$.
\end{enumerate}

It is easy to verify that $\check{\alpha}_0(\cdot)$ is increasing on $[\underline{c},\overline{c}]$. For $\underline{\alpha}_1\in [\alpha_0^1,\check{\alpha}_1)$, it holds that $\frac{d\check{\alpha}_0(\underline{\alpha}_1)}{d\underline{\alpha}_1} = \frac{K_0'(\underline{\alpha}_1)-K_1'(\alpha_1)}{K_0'(\alpha_0)- K_1'\big(\check{\alpha}_0(\underline{\alpha}_1)\big)}$, from which we have that $\check{\alpha}_0(\cdot)$ is convex on $[\alpha_0^1,\check{\alpha}_1)$. \\
\\
\textbf{Proof of Theorem \ref{IMP-pp2}:} Following the same reasoning as in Lemma \ref{IMP-lm1}, a pair of output distributions $\mu = (\mu_0,\mu_1)\in \Delta([\underline{x},\overline{x}])^2$ with aggregate efforts $a = (\alpha_0,\alpha_1)$ is implementable if and only if there exists a feasible menu $(a,m, R',\mu)\in \mathcal{F}$ where $R': = (\text{supp}(\mu_0),\text{supp}(\mu_1))$. It is equivalent to examining whether there exist constant payments $m_0,m_1 \geq 0$ such that the menu $(a,m,R',\mu)$ satisfies the IC constraints \eqref{MBC-IC0} and \eqref{MBC-IC1}. Define $\underline{\alpha}_t := c\big(\min\text{supp}(\mu_t)\big)$ and $\overline{\alpha}_t:= c\big(\max\text{supp}(\mu_t)\big)$. By using \eqref{MBC-conv}, $(a,m,R',\mu)$ satisfies \eqref{MBC-IC0} and \eqref{MBC-IC1} if and only if $(m_0,m_1)$ satisfies that\footnote{The first argument of the max operator on the RHS of \eqref{IMP-eq18a} is equal to $m_1$ if $\underline{\alpha}_1 = \underline{c}$. Similarly, the first argument of the max operator on the RHS of \eqref{IMP-eq18b} is equal to $m_0$ if $\underline{\alpha}_0 = \underline{c}$.} :
\begin{subequations}
\begin{align}
    K_0'(\alpha_0)\cdot \alpha_0 - K_0(\alpha_0) + m_0 \geq \max\Big\{
    \max_{\alpha\in [\underline{c},\underline{\alpha}_1]}\big[\big(\frac{K_1'(\alpha_1)\cdot \underline{\alpha}_1+m_1}{\underline{\alpha}_1}\cdot \alpha - K_0(\alpha)\big)\big],\atop
    \qquad\qquad \qquad \qquad \qquad \qquad\qquad\qquad \max_{\alpha\in [\underline{\alpha}_1,\overline{\alpha}_1]}\big[K_1'(\alpha_1)\cdot \alpha - K_0(\alpha)+m_1\big]
    \Big\},\label{IMP-eq18a}\\
     K_1'(\alpha_1)\cdot \alpha_1 - K_1(\alpha_1) + m_1 \geq \max\Big\{
    \max_{\alpha\in [\underline{c},\underline{\alpha}_0]}\big[\big(\frac{K_0'(\alpha_0)\cdot \underline{\alpha}_0+m_0}{\underline{\alpha}_0}\cdot \alpha - K_1(\alpha)\big)\big],\atop
    \qquad\qquad \qquad \qquad \qquad \qquad\qquad\qquad \max_{\alpha\in [\underline{\alpha}_0,\overline{\alpha}_0]}\big[K_0'(\alpha_0)\cdot \alpha - K_1(\alpha)+m_0\big]
    \Big\}. \label{IMP-eq18b}
\end{align}
\end{subequations}

Under Assumption \ref{OM-as1}, we have that $K_0'(\alpha_0)\geq K_1'(\alpha_0)\geq K_1'(\underline{\alpha}_0)$, which implies that the RHS of \eqref{IMP-eq18b} is equal to $\max_{\alpha\in [\underline{\alpha}_0,\overline{\alpha}_0]}\big[K_0'(\alpha_0)\cdot \alpha - K_1(\alpha)+m_0\big]$. Hence, \eqref{IMP-eq18a} can be written \eqref{IMP-eq19a} and \eqref{IMP-eq19b}, and \eqref{IMP-eq18b} as \eqref{IMP-eq19c}. 
\begin{subequations}
\begin{align}
K_0'(\alpha_0)\cdot \alpha_0 - K_0(\alpha_0) + m_0 &\geq \max_{\alpha\in [\underline{c},\underline{\alpha}_1]}\big[\big(\frac{K_1'(\alpha_1)\cdot \underline{\alpha}_1+m_1}{\underline{\alpha}_1}\big)\cdot \alpha - K_0(\alpha)\big],\label{IMP-eq19a}\\
K_0'(\alpha_0)\cdot \alpha_0 - K_0(\alpha_0) + m_0&\geq \max_{\alpha\in [\underline{\alpha}_1,\overline{\alpha}_1]}\big[K_1'(\alpha_1)\cdot \alpha - K_0(\alpha)+m_1\big], \label{IMP-eq19b}\\
K_1'(\alpha_1)\cdot \alpha_1 - K_1(\alpha_1) + m_1 &\geq \max_{\alpha\in [\underline{\alpha}_0,\overline{\alpha}_0]}\big[K_0'(\alpha_0)\cdot \alpha - K_1(\alpha)+m_0\big]. \label{IMP-eq19c}
\end{align}
\end{subequations}

We prove the statement for the following two cases. \\
\\
\textbf{Case 1: } $\alpha_0\leq \alpha_1$ and  $K_0'(\alpha_0)\leq K_1'(\alpha_1)$. 

In this case, we show that $(\mu_0,\mu_1)$ is always implementable. Specifically, we can show that $(\mu_0,\mu_1)$ can be implemented by the following constant payments $(m_0',m_1')$:
\begin{equation*}
\begin{aligned}
    m_0' :&= \big[K_1'(\alpha_1)\cdot k_0^{-1}\big(K_1'(\alpha_1)\big)-K_0\big(k_0^{-1}\big(K_1'(\alpha_1)\big)\big)\big]- \big[K_0'(\alpha_0)\cdot \alpha_0 - K_0(\alpha_0)\big].\\
    m_1' :&= 0
\end{aligned}
\end{equation*}

It suffices to verify that $(m_0',m_1')$ satisfies \eqref{IMP-eq19a} -- \eqref{IMP-eq19c}. Define $\alpha_0^1$ and $\alpha_1^0$ as in \eqref{IMP-eq12} and \eqref{IMP-eq13}. On one hand, we have that 
\begin{equation*}
\begin{aligned}
K_0'(\alpha_0)\cdot \alpha_0 - K_0(\alpha_0) + m_0' &=   \big[K_1'(\alpha_1)\cdot \alpha_0^1-K_0(\alpha_0^1)\big] \\
&= \max_{\alpha\in [\underline{c},\overline{c}]}\big[K_1'(\alpha_1)\cdot \alpha - K_0(\alpha)\big].
\end{aligned}
\end{equation*}

It directly implies that $(m_0',m_1')$ satisfies \eqref{IMP-eq19a} and \eqref{IMP-eq19b}. 

On the other hand, we have that
\begin{equation*}
\begin{aligned}
m_0' &= \big[K_1'(\alpha_1)\cdot \alpha_0^1-K_0(\alpha_0^1)\big]- \big[K_0'(\alpha_0)\cdot \alpha_0 - K_0(\alpha_0)\big] \\
&= \int_{K_0'(\alpha_0)}^{K_1'(\alpha_1)}k_0^{-1}(k)dk \leq \int_{K_0'(\alpha_0)}^{K_1'(\alpha_1)}k_1^{-1}(k)dk\\
&= \big[K_1'(\alpha_1)\cdot \alpha_1 - K_1(\alpha_1)\big]- \big[K_0'(\alpha_0)\cdot \alpha_1^0 - K_1(\alpha_1^0)\big] \\
&\leq \big[K_1'(\alpha_1)\cdot \alpha_1 - K_1(\alpha_1)\big] -\max_{\alpha\in [\underline{\alpha}_0,\overline{\alpha}_0]}\big[K_0'(\alpha_0)\cdot\alpha - K_1(\alpha)\big],
\end{aligned}
\end{equation*}
where the inequality follows from that $K_0'(\alpha_0)\cdot \alpha_1^0 - K_1(\alpha_1^0) = \max_{\alpha\in[\underline{c},\overline{c}]}\big[K_0'(\alpha_0)\cdot\alpha - K_1(\alpha)\big] \geq \max_{\alpha\in [\underline{\alpha}_0,\overline{\alpha}_0]}\big[K_0'(\alpha_0)\cdot\alpha - K_1(\alpha)\big]$. Hence, $(m_0',m_1')$ also satisfies \eqref{IMP-eq19c}. 

In conclusion, $(\mu_0,\mu_1)$ can be implemented. Take $m':=(m_0',m_1')$. Following the similar idea as in Lemma \ref{IMP-lm1}, $(a,m',R',\mu)$ is the payment-minimizing menu among all the feasible menus that implement $(\mu_0,\mu_1)$. \\
\\
\textbf{Case 2:} $\alpha_0\leq \alpha_1$ and $K_0'(\alpha_0) > K_1'(\alpha_1)$. 

First, we show that, in this case, a necessary condition for $(\mu_0,\mu_1)$ to be implementable is $\overline{\alpha}_0\leq \check{\alpha}_0(\underline{\alpha}_1)$. \\
\\
\textit{Necessity:} Suppose $(\mu_0,\mu_1)$ is implementable, there exists $(m_0,m_1)\in \mathbb{R}_{+}^2$ satisfying \eqref{IMP-eq19b} and \eqref{IMP-eq19c}. Define $\alpha_0^1$ and $\alpha_1^0$ as in \eqref{IMP-eq12} and \eqref{IMP-eq13}.  

When $\underline{\alpha}_1 \geq \alpha_0^1$, by $K_1'(\alpha_1) \leq K_0'(\alpha_0^1)\leq K_0'(\alpha)$ for all $\alpha \in [\underline{\alpha}_1,\overline{\alpha}_1]$, we have that \eqref{IMP-eq19b} is equivalent to 
\begin{equation}
    m_1- m_0 \leq \big[K_0'(\alpha_0)\cdot \alpha_0 - K_0(\alpha_0)\big] - \big[K_1'(\alpha_1)\cdot \underline{\alpha}_1-K_0(\underline{\alpha}_1)\big].
\label{IMP-eq20}
\end{equation}

When $\underline{\alpha}_1 < \alpha_0^1$, it implies that $\alpha_0^1 > \underline{c}$ and $K_0'(\alpha_0^1) = K_1'(\alpha_1)$. From $\underline{\alpha}_1 < \alpha_0^1 \leq \alpha_0\leq \alpha_1$, we have that \eqref{IMP-eq19b} is equivalent to 
\begin{equation}
m_1- m_0 \leq \big[K_0'(\alpha_0)\cdot \alpha_0 - K_0(\alpha_0)\big] - \big[K_1'(\alpha_1)\cdot \alpha_0^1-K_0(\alpha_0^1)\big].
\label{IMP-eq21}
\end{equation}

When $ \overline{\alpha}_0 \leq \alpha_1^0$, by $K_1'(\alpha)\leq K_1'(\alpha_1^0) \leq K_0'(\alpha_0)$ for all $\alpha \in [\underline{\alpha}_0,\overline{\alpha}_0]$, we have that \eqref{IMP-eq19c} is equivalent to 
\begin{equation}
    m_1 - m_0 \geq \big[K_0'(\alpha_0)\cdot \overline{\alpha}_0 - K_1(\overline{\alpha}_0)\big] -\big[K_1'(\alpha_1)\cdot \alpha_1 - K_1(\alpha_1)\big].
\label{IMP-eq22}
\end{equation}

When $\overline{\alpha}_0 > \alpha_1^0$, it implies that $\alpha_1^0 < \overline{c}$ and $K_1'(\alpha_1^0) = K_0'(\alpha_0)$. From $\alpha_0\leq \alpha_1 \leq \alpha_1^0 < \overline{\alpha}_0$, we have that \eqref{IMP-eq19c} is equivalent to 
\begin{equation}
    m_1 - m_0 \geq \big[K_0'(\alpha_0)\cdot \alpha_1^0 - K_1(\alpha_1^0)\big] -\big[K_1'(\alpha_1)\cdot \alpha_1 - K_1(\alpha_1)\big]. 
\label{IMP-eq23}
\end{equation}

We consider the following four cases based on the values of $\overline{\alpha}_0$ and $\underline{\alpha}_1$.\\
\\
\textbf{Case 2a:} $\overline{\alpha}_0 > \alpha_1^0$ and $\underline{\alpha}_1 < \alpha_0^1$.

In this case, $(\mu_0,\mu_1)$ is implementable only if $\big[K_0'(\alpha_0)\cdot \alpha_1^0 - K_1(\alpha_1^0)\big] -\big[K_1'(\alpha_1)\cdot \alpha_1 - K_1(\alpha_1)\big] \leq m_1-m_0 \leq \big[K_0'(\alpha_0)\cdot \alpha_0 - K_0(\alpha_0)\big] - \big[K_1'(\alpha_1)\cdot \alpha_0^1-K_0(\alpha_0^1)\big]$. According to Lemma \ref{APP-IMP-lm1}, this is impossible. \\
\\
\textbf{Case 2b:} $\overline{\alpha}_0 > \alpha_1^0$ and $\underline{\alpha}_1 \geq \alpha_0^1$.

In this case, $(\mu_0,\mu_1)$ is implementable only if $\big[K_0'(\alpha_0)\cdot \alpha_1^0 - K_1(\alpha_1^0)\big] -\big[K_1'(\alpha_1)\cdot \alpha_1 - K_1(\alpha_1)\big]\leq m_1-m_0 \leq \big[K_0'(\alpha_0)\cdot \alpha_0 - K_0(\alpha_0)\big] - \big[K_1'(\alpha_1)\cdot \underline{\alpha}_1-K_0(\underline{\alpha}_1)\big]$. According to the definition \eqref{IMP-eq14}, this is possible only when $\underline{\alpha}_1 \geq \check{\alpha}_1$. \\
\\
\textbf{Case 2c:} $\overline{\alpha}_0 \leq \alpha_1^0$ and $\underline{\alpha}_1 <  \alpha_0^1$.

In this case, $(\mu_0,\mu_1)$ is implementable only if  $\big[K_0'(\alpha_0)\cdot \overline{\alpha}_0 - K_1(\overline{\alpha}_0)\big] -\big[K_1'(\alpha_1)\cdot \alpha_1 - K_1(\alpha_1)\big]\leq m_1-m_0\leq \big[K_0'(\alpha_0)\cdot \alpha_0 - K_0(\alpha_0)\big] - \big[K_1'(\alpha_1)\cdot \alpha_0^1-K_0(\alpha_0^1)\big]$. According to the definition of $\check{\alpha}_0$, it is possible only when $\underline{\alpha}_0 \leq \check{\alpha}_0$. \\
\\
\textbf{Case 2d:} $\overline{\alpha}_0 \leq \alpha_1^0$ and $\underline{\alpha}_1 \geq \alpha_0^1$. 

In this case, $(\mu_0,\mu_1)$ is implementable only if $\big[K_0'(\alpha_0)\cdot \overline{\alpha}_0 - K_1(\overline{\alpha}_0)\big] -\big[K_1'(\alpha_1)\cdot \alpha_1 - K_1(\alpha_1)\big]\leq m_1-m_0\leq \big[K_0'(\alpha_0)\cdot \alpha_0 - K_0(\alpha_0)\big] - \big[K_1'(\alpha_1)\cdot \underline{\alpha}_1-K_0(\underline{\alpha}_1)\big]$. If $\underline{\alpha}_1 \geq \check{\alpha}_1$, this is possible for any $\overline{\alpha}_0 \leq \alpha_1^0$. If $\alpha_0^1\leq \underline{\alpha}_1 < \check{\alpha}_1$, according to the definition of $\check{\alpha}_0$, $(\mu_0,\mu_1)$ is implementable only when $\overline{\alpha}_0\leq \check{\alpha}_0(\underline{\alpha}_1)$.\\
\\
In conclusion, within Case 2, $(\mu_0,\mu_1)$ is implementable only if $(\underline{\alpha}_1,\overline{\alpha}_0)$ satisfies one of the following: \circled{1} $\underline{\alpha}_1 < \alpha_0^1$ and $\overline{\alpha}_0\leq \check{\alpha}_0$; \circled{2} $\underline{\alpha}_1 \in [\alpha_0^1, \check{\alpha}_1)$ and $\overline{\alpha}_0 \leq \check{\alpha}_0(\underline{\alpha}_1)$; \circled{3} $\underline{\alpha}_1\geq \check{\alpha}_1$ and $\overline{\alpha}_0\leq \overline{c}$. Hence, we have shown the necessity of $\overline{\alpha}_0 \leq \check{\alpha}_0(\underline{\alpha}_1)$. \\
\\
\textit{Sufficiency:} We aim to show that, within Case 2, $(\mu_0,\mu_1)$ is implementable as long as $(\underline{\alpha}_1,\overline{\alpha}_0)$ satisfies one of the following: \circled{1} $\underline{\alpha}_1 < \alpha_0^1$ and $\overline{\alpha}_0\leq \check{\alpha}_0$; \circled{2} $\underline{\alpha}_1 \in [\alpha_0^1, \check{\alpha}_1)$ and $\overline{\alpha}_0 \leq \check{\alpha}_0(\underline{\alpha}_1)$; \circled{3} $\underline{\alpha}_1\geq \check{\alpha}_1$ and $\overline{\alpha}_0\leq \overline{c}$. Given any $(\underline{\alpha}_1,\overline{\alpha}_0)$ in cases \circled{1}, \circled{2} or \circled{3}, we construct constant payments $(m_0,m_1)$ meeting \eqref{IMP-eq19a} -- \eqref{IMP-eq19c}. 

In case of \circled{1}, take
\begin{equation*}
    \tilde{m}_1 := \big[K_0'(\alpha_0)\cdot \alpha_0 - K_0(\alpha_0)\big] - \big[K_1'(\alpha_1)\cdot \alpha_0^1-K_0(\alpha_0^1)\big].
\end{equation*}

In case of \circled{2} and \circled{3}, take
\begin{equation*}
    \tilde{m}_1:=\big[K_0'(\alpha_0)\cdot \alpha_0 - K_0(\alpha_0)\big] - \big[K_1'(\alpha_1)\cdot \underline{\alpha}_1-K_0(\underline{\alpha}_1)\big].
\end{equation*}

We have that $\tilde{m}_1 \geq \big[K_0'(\alpha_0)\cdot \alpha_0 - K_0(\alpha_0)\big] - \big[K_1'(\alpha_1)\cdot \alpha_0^1-K_0(\alpha_0^1)\big] \geq 0$ and 
\begin{equation}
\tilde{m}_1  \leq \big[K_0'(\alpha_0)\cdot \alpha_0 - K_0(\alpha_0)\big] - \big[K_1'(\alpha_1)\cdot \underline{\alpha}_1-K_0(\underline{\alpha}_1)\big],
\label{IMP-eq24}
\end{equation}
which follows from that $K_1'(\alpha_1)\cdot \alpha_0^1 - K_0(\alpha_0^1) \geq  K_1'(\alpha_1)\cdot \alpha - K_0(\alpha)$ for all $\alpha \in [\underline{c},\overline{c}]$. 

If $\tilde{m}_1$ satisfies that $K_0'(\alpha_0)\cdot \alpha_0- K_0(\alpha_0)\geq\max_{\alpha \in [\underline{c},\underline{\alpha}_1]}\big[\big(\frac{K_1'(\alpha_1)\cdot \underline{\alpha}_1 + \tilde{m}_1}{\underline{\alpha}_1}\big)\cdot \alpha - K_0(\alpha)\big]$, let $m_0 := 0$ and $m_1:= \tilde{m}_1$. Then, $(m_0,m_1)$ satisfies \eqref{IMP-eq19a} -- \eqref{IMP-eq19c}, and implements $(\mu_0,\mu_1)$. 

Suppose $K_0'(\alpha_0)\cdot \alpha_0- K_0(\alpha_0)<\max_{\alpha \in [\underline{c},\underline{\alpha}_1]}\big[\big(\frac{K_1'(\alpha_1)\cdot \underline{\alpha}_1 + \tilde{m}_1}{\underline{\alpha}_1}\big)\cdot \alpha - K_0(\alpha)\big]$. It must hold that $\underline{\alpha}_1 > \underline{c}$. Furthermore, we can show that $\frac{K_1'(\alpha_1)\cdot \underline{\alpha}_1 + \tilde{m}_1}{\underline{\alpha}_1}< K_0'(\underline{\alpha}_1)$. This is because if $\frac{K_1'(\alpha_1)\cdot \underline{\alpha}_1 + \tilde{m}_1}{\underline{\alpha}_1} \geq K_0'(\underline{\alpha}_1)$, according to our assumption and \eqref{IMP-eq24}, we have that
\begin{equation*}
\begin{aligned}
K_0'(\alpha_0)\cdot \alpha_0- K_0(\alpha_0)&<\max_{\alpha \in [\underline{c},\underline{\alpha}_1]}\big[\big(\frac{K_1'(\alpha_1)\cdot \underline{\alpha}_1 + \tilde{m}_1}{\underline{\alpha}_1}\big)\cdot \alpha - K_0(\alpha)\big]\\
&= K_1'(\alpha_1)\cdot \underline{\alpha}_1 - K_0(\underline{\alpha}_1)+ \tilde{m}_1 \\
&\leq K_0'(\alpha_0)\cdot \alpha_0 - K_0(\alpha_0),
\end{aligned}
\end{equation*}
which is impossible. It implies that $K_1'(\alpha_1)<K_0'(\underline{\alpha}_1)$, and we are in either case \circled{2} or case \circled{3}. We have that $\tilde{m}_1 =\big[K_0'(\alpha_0)\cdot \alpha_0 - K_0(\alpha_0)\big] - \big[K_1'(\alpha_1)\cdot \underline{\alpha}_1-K_0(\underline{\alpha}_1)\big]$ as defined. 

Define $\Delta m_1$ as follows:
\begin{equation*}
\Delta m_1 := \big[K_0'(\underline{\alpha}_1) - K_1'(\alpha_1)\big]\cdot \underline{\alpha}_1 - \tilde{m}_1 \Longrightarrow \frac{K_1'(\alpha_1)\cdot \underline{\alpha}_1+\tilde{m}_1+\Delta m_1}{\underline{\alpha}_1} = K_0'(\underline{\alpha}_1). 
\end{equation*}

We have that $\Delta m_1 >0$ by $\frac{K_1'(\alpha_1)\cdot \underline{\alpha}_1 + \tilde{m}_1}{\underline{\alpha}_1}< K_0'(\underline{\alpha}_1)$. By the definition of $\tilde{m}_1$, we have that
\begin{equation}
\begin{aligned}
\Delta m_1  &= \big[K_0'(\underline{\alpha}_1) - K_1'(\alpha_1)\big]\cdot \underline{\alpha}_1 - \big[K_0'(\alpha_0)\cdot \alpha_0 - K_0(\alpha_0)\big] + \big[K_1'(\alpha_1)\cdot \underline{\alpha}_1-K_0(\underline{\alpha}_1)\big]  \\
&= \big[ K_0'(\underline{\alpha}_1)\cdot \underline{\alpha}_1- K_0(\underline{\alpha}_1)\big] - \big[K_0'(\alpha_0)\cdot \alpha_0 - K_0(\alpha_0)\big]\\
&= \max_{\alpha\in [\underline{c}, \underline{\alpha}_1]}\big[\big(\frac{K_1'(\alpha_1)\cdot \underline{\alpha}_1 + \tilde{m}_1+\Delta m_1}{\underline{\alpha}_1}\big)\cdot \alpha - K_0(\alpha)\big] - \big[K_0'(\alpha_0)\cdot \alpha_0 - K_0(\alpha_0)\big].
\end{aligned}
\label{IMP-eq25}
\end{equation}

Take $m_0:= \Delta m_1$ and $m_1:= \tilde{m}_1 + \Delta m_1$. Since $m_1 - m_0 = \tilde{m}_1$, $(m_0,m_1)$ satisfies \eqref{IMP-eq19b} and \eqref{IMP-eq19c}. Furthermore, by \eqref{IMP-eq25}, $(m_0,m_1)$ satisfies \eqref{IMP-eq19a}. In conclusion, $(\mu_0,\mu_1)$ can be implemented by $(m_0,m_1)$.$\hfill\square$ \\
\\
\textbf{Proof of Corollary \ref{IMP-co1}:} The statement directly follows from Theorem \ref{IMP-pp2} and the fact that $\check{\alpha}_0 < \alpha_1^0 \leq \overline{c}$.$\hfill\square$

\subsection{Implementation by Regular Constant Payments}
\label{APP-IMP-RCP}
We first provide the exact expression of the function $\tilde{\alpha}_0(\cdot)$. Given a pair of aggregate efforts $(\alpha_0,\alpha_1)$ that satisfies $\alpha_0\leq \alpha_1$ and $K_0'(\alpha_0) > K_1'(\alpha_1)$, define the following threshold:
\begin{equation}
\tilde{\alpha}_1 := \frac{1}{K_0'(\alpha_0) - K_1'(\alpha_1)}\cdot \Big\{\big[K_0'(\alpha_0)\cdot \alpha_1^0 - K_1(\alpha_1^0)\big] - \big[K_1'(\alpha_1)\cdot \alpha_1 - K_1(\alpha_1)\big]\Big\}. 
\label{IMP-eq26}
\end{equation}

We can verify that $\tilde{\alpha}_1 \geq \check{\alpha}_1$\footnote{It follows from that $\big[K_0'(\alpha_0)\cdot \alpha_0 - K_0(\alpha_0)\big] - \big[K_1'(\alpha_1)\cdot \tilde{\alpha}_1 - K_0(\tilde{\alpha}_1)\big]= \big[K_0'(\alpha_0)-K_1'(\alpha_1)\big]\cdot \tilde{\alpha}_1 + \int_{\alpha_0}^{\tilde{\alpha}_1}\big[K_0'(\alpha)-K_0'(\alpha_0)\big]d\alpha\geq \big[K_0'(\alpha_0)\cdot \alpha_1^0 - K_1(\alpha_1^0)\big] - \big[K_1'(\alpha_1)\cdot \alpha_1 - K_1(\alpha_1)\big]$, where the inequality follows from that $\tilde{\alpha}_1 \geq \alpha_0$}.\\
\\
\textbf{Definition of $\tilde{\alpha}_0(\cdot)$:} The function $\tilde{\alpha}_0:[\underline{c},\overline{c}]\to [\underline{c},\overline{c}]$ takes the following form:
\begin{enumerate}
\item For $\underline{\alpha}_1 < \alpha_0^1$, define $\tilde{\alpha}_0(\underline{\alpha}_1) := \check{\alpha}_0 \in (\underline{c},\alpha_1^0)$.
\item For $\alpha_0^1\leq \underline{\alpha}_1 \leq \alpha_0$, define $\tilde{\alpha}_0(\underline{\alpha}_1):= \check{\alpha}_0(\underline{\alpha}_1) \in [\check{\alpha}_0,\alpha_1^0)$.
\item For $\alpha_0 < \underline{\alpha}_1 <\tilde{\alpha}_1$, $\tilde{\alpha}_0(\underline{\alpha}_1)\in [\check{\alpha}_0,\alpha_1^0)$ is pinned down by
\begin{equation*}
\big[K_0'(\alpha_0)\cdot \tilde{\alpha}_0(\underline{\alpha}_1)-K_1\big(\tilde{\alpha}_0(\underline{\alpha}_1)\big)\big]-\big[K_1'(\alpha_1)\cdot \alpha_1-K_1(\alpha_1)\big] = \big[K_0'(\alpha_0)-K_1'(\alpha_1)\big]\cdot \underline{\alpha}_1. 
\end{equation*}
\item For $\underline{\alpha}_1\geq \tilde{\alpha}_1$, define $\tilde{\alpha}_0(\underline{\alpha}_1):= \overline{c}$. 
\end{enumerate}
We say $(\mu_0,\mu_1)$ is implementable by regular constant payments if it can be implemented through constant payments analogous to those in Lemma \ref{IMP-lm1}. When $(\mu_0,\mu_1)$ induces $(\alpha_0,\alpha_1)$ that satisfies $\alpha_0\leq \alpha_1$ and $K_0'(\alpha_0)>K_1'(\alpha_1)$, regular constant payments mean that the type-0 agent receives zero constant payment. The following proposition establishes that, in the current case, $(\mu_0,\mu_1)$ can be implemented by regular constant payments if and only if $(\underline{\alpha}_1,\overline{\alpha}_0)$ lies in the region below $\tilde{\alpha}_0(\cdot)$. 
\begin{proposition}
Take a pair of output distributions $(\mu_0,\mu_1)\in \Delta([\underline{x},\overline{x}])^2$ with corresponding aggregate efforts $(\alpha_0,\alpha_1)$ satisfying $\alpha_0 \leq \alpha_1$ and $K_0'(\alpha_0) > K_1'(\alpha_1)$. Under Assumption \ref{OM-as1}, $(\mu_0,\mu_1)$ can be implemented by regular constant payments if and only if
\begin{equation*}
    \overline{\alpha}_0 \leq \tilde{\alpha}_0(\underline{\alpha}_1).
\end{equation*}
\label{APP-IMP-pp1}
\end{proposition}
\begin{proof}
Take a pair of output distributions $(\mu_0,\mu_1)\in \Delta([\underline{x},\overline{x}])^2$ with corresponding aggregate efforts $(\alpha_0,\alpha_1)$ satisfying $\alpha_0 \leq \alpha_1$ and $K_0'(\alpha_0) > K_1'(\alpha_1)$. Following the same reasoning as in Lemma \ref{IMP-lm1}, $\mu = (\mu_0,\mu_1)$ is implementable by regulart constant payments if and only if there exists a feasible menu $(a,m, R',\mu)\in \mathcal{F}$ where $R': = (\text{supp}(\mu_0),\text{supp}(\mu_1))$ and $m_0 = 0$. Similar to the proof of Theorem \ref{IMP-pp2}, it is equivalent to examining whether there exists $m_1\geq 0$ that satisfies the following constraints:
\begin{subequations}
\begin{align}
K_0'(\alpha_0)\cdot \alpha_0 - K_0(\alpha_0) &\geq \max_{\alpha\in [\underline{c},\underline{\alpha}_1]}\big[\big(\frac{K_1'(\alpha_1)\cdot \underline{\alpha}_1+m_1}{\underline{\alpha}_1}\big)\cdot \alpha - K_0(\alpha)\big],\label{IMP-eq27a}\\
K_0'(\alpha_0)\cdot \alpha_0 - K_0(\alpha_0) &\geq \max_{\alpha\in [\underline{\alpha}_1,\overline{\alpha}_1]}\big[K_1'(\alpha_1)\cdot \alpha - K_0(\alpha)+m_1\big], \label{IMP-eq27b}\\
K_1'(\alpha_1)\cdot \alpha_1 - K_1(\alpha_1) + m_1 &\geq \max_{\alpha\in [\underline{\alpha}_0,\overline{\alpha}_0]}\big[K_0'(\alpha_0)\cdot \alpha - K_1(\alpha)\big]. \label{IMP-eq27c}
\end{align}
\end{subequations}
We first consider constraint \eqref{IMP-eq27a}. Following the similar idea as in Lemma \ref{IMP-lm1}, when $\underline{\alpha}_1>\alpha_0$, \eqref{IMP-eq27a} is equivalent to $m_1 \leq \big[K_0'(\alpha_0) - K_1'(\alpha_1)\big]\cdot \underline{\alpha}_1$. When $\underline{\alpha}_1 \leq \alpha_0$, we claim that \eqref{IMP-eq27a} is equivalent to $m_1 \leq \big[K_0'(\alpha_0)\cdot \alpha_0 - K_0(\alpha_0)\big]-\big[K_1'(\alpha_1)\cdot\underline{\alpha}_1 -K_0(\underline{\alpha}_1)\big]$.
\begin{claim}
    When $\underline{\alpha}_1 \leq \alpha_0$, \eqref{IMP-eq27a} holds if and only if
    \begin{equation}
         m_1 \leq \big[K_0'(\alpha_0)\cdot \alpha_0 - K_0(\alpha_0)\big]-\big[K_1'(\alpha_1)\cdot\underline{\alpha}_1 -K_0(\underline{\alpha}_1)\big]. 
        \label{IMP-eq28}
    \end{equation}
\label{APP-IMP-cl1}
\end{claim}

\begin{proof}[Proof of Claim \ref{APP-IMP-cl1}.]
We first show that \eqref{IMP-eq28} implies \eqref{IMP-eq27a}. Suppose $m_1$ satisfies $m_1 \leq \big[K_0'(\alpha_0)\cdot \alpha_0 - K_0(\alpha_0)\big]-\big[K_1'(\alpha_1)\cdot\underline{\alpha}_1 -K_0(\underline{\alpha}_1)\big]$. By $\big[K_0'(\alpha_0)-K_1'(\alpha_1)\big]\cdot \underline{\alpha}_1 \leq \big[K_0'(\alpha_0)\cdot \alpha_0 -K_0(\alpha_0)\big]-\big[K_1'(\alpha_1)\cdot\underline{\alpha}_1 - K_0(\underline{\alpha}_1)\big]$, we divide this situation into two cases.

Suppose $m_1 \leq \big[K_0'(\alpha_0) - K_1'(\alpha_1)\big]\cdot \underline{\alpha}_1$. If $\underline{\alpha}_1 = \underline{c}$, it implies that $m_1 = 0$, and \eqref{IMP-eq27a} naturally holds. If $\underline{\alpha}_1 > \underline{c}$, it implies that $\frac{K_1'(\alpha_1)\cdot \underline{\alpha}_1+m_1}{\underline{\alpha}_1}\leq K_0'(\alpha_0)$. In this case, we have that $\max_{\alpha\in [\underline{c},\underline{\alpha}_1]}\big[\big(\frac{K_1'(\alpha_1)\cdot \underline{\alpha}_1+m_1}{\underline{\alpha}_1}\big)\cdot \alpha - K_0(\alpha)\big] \leq \max_{\alpha\in [\underline{c},\underline{\alpha}_1]}\big[K_0'(\alpha_0)\cdot \alpha - K_0(\alpha)\big] \leq K_0'(\alpha_0)\cdot \alpha_0 - K_0(\alpha_0)$. Hence, \eqref{IMP-eq27a} is satisfied. 

Suppose $\big[K_0'(\alpha_0) - K_1'(\alpha_1)\big]\cdot \underline{\alpha}_1 < m_1 \leq \big[K_0'(\alpha_0)\cdot \alpha_0 -K_0(\alpha_0)\big]-\big[K_1'(\alpha)\cdot\underline{\alpha}_1 - K_0(\underline{\alpha}_1)\big]$. If $\underline{\alpha}_1 = \underline{c}$, we have that $m_1 \leq K_0'(\alpha_0)\cdot \alpha_0 - K_0(\alpha_0)$, which further implies \eqref{IMP-eq27a}. If $\underline{\alpha}_1 > \underline{c}$, we have that $\frac{K_1'(\alpha_1)\cdot \underline{\alpha}_1+m_1}{\underline{\alpha}_1} > K_0'(\alpha_0)\geq K_0'(\underline{\alpha}_1)$. The RHS of \eqref{IMP-eq27a} satisfies that $\max_{\alpha\in [\underline{c},\underline{\alpha}_1]}\big[\big(\frac{K_1'(\alpha_1)\cdot \underline{\alpha}_1+m_1}{\underline{\alpha}_1}\big)\cdot \alpha - K_0(\alpha)\big] = K_1'(\alpha_1)\cdot \underline{\alpha}_1 - K_0(\underline{\alpha}_1) + m_1 \leq K_0'(\alpha_0)\cdot \alpha_0 - K_0(\alpha_0)$. Hence, \eqref{IMP-eq27a} is satisfied.

Next, we show that if $m_1 > \big[K_0'(\alpha_0)\cdot \alpha_0 - K_0(\alpha_0)\big]-\big[K_1'(\alpha_1)\cdot\underline{\alpha}_1 -K_0(\underline{\alpha}_1)\big]$, \eqref{IMP-eq28} does not hold. It is because, in this case, the RHS of \eqref{IMP-eq27a} satisfies that $\max_{\alpha\in [\underline{c},\underline{\alpha}_1]}\big[\big(\frac{K_1'(\alpha_1)\cdot \underline{\alpha}_1+m_1}{\underline{\alpha}_1}\big)\cdot \alpha - K_0(\alpha)\big]  \geq K_1'(\alpha_1)\cdot \underline{\alpha}_1 - K_0(\underline{\alpha}_1) + m_1 > K_0'(\alpha_0)\cdot \alpha_0 - K_0(\alpha_0)$.
\end{proof}
$\quad$\\
\textit{Proof of Proposition \ref{APP-IMP-pp1} continued.} For the next step, we consider constraint \eqref{IMP-eq27b}. When $\underline{\alpha}_1 \geq \alpha_0^1$, we have that $K_0'(\underline{\alpha}_1)\geq K_1'(\alpha_1)$, and the RHS of \eqref{IMP-eq27b} is equal to $K_1'(\alpha_1)\cdot \underline{\alpha}_1 - K_0(\underline{\alpha}_1)+ m_1$. Hence, \eqref{IMP-eq27b} is equivalent to $m_1 \leq \big[K_0'(\alpha_0)\cdot \alpha_0 - K_0(\alpha_0)\big] - \big[K_1'(\alpha_1)\cdot \underline{\alpha}_1 - K_0(\underline{\alpha}_1)\big]$. When $\underline{\alpha}_1 < \alpha_0^1$, we have that $K_0'(\underline{\alpha}_1) < K_1'(\alpha_1)$, and the RHS of \eqref{IMP-eq27b} is equal to $K_1'(\alpha_1)\cdot \alpha_0^1 - K_0(\alpha_0^1) +m_1$. Consequently, \eqref{IMP-eq27b} is equivalent to $m_1 \leq \big[K_0'(\alpha_0)\cdot \alpha_0 - K_0(\alpha_0)\big] - \big[K_1'(\alpha_1)\cdot \alpha_0^1 - K_0(\alpha_0^1)\big]$. 

Finally, we consider constraint \eqref{IMP-eq27c}. When $\overline{\alpha}_0 \leq \alpha_1^0$, we have that $K_1'(\overline{\alpha}_0) \leq K_0'(\alpha_0)$, and the RHS of \eqref{IMP-eq27c} is equal to $K_0'(\alpha_0)\cdot \overline{\alpha}_0 - K_1(\overline{\alpha}_0)$. Hence, \eqref{IMP-eq27c} is equivalent to $m_1 \geq \big[K_0'(\alpha_0)\cdot \overline{\alpha}_0 - K_1(\overline{\alpha}_0)\big]-\big[K_1'(\alpha_1)\cdot \alpha_1 - K_1(\alpha_1)\big]$. When $\overline{\alpha}_0 > \alpha_1^0$, we have that $K_1'(\overline{\alpha}_0) > K_0'(\alpha_0)$, and the RHS of \eqref{IMP-eq27c} is equal to $K_0'(\alpha_0)\cdot \alpha_1^0 - K_1(\alpha_1^0)$. Consequently, \eqref{IMP-eq27c} is equivalent to $m_1 \geq \big[K_0'(\alpha_0)\cdot \alpha_1^0 - K_1(\alpha_1^0)\big] - \big[K_1'(\alpha_1)\cdot \alpha_1 - K_1(\alpha_1)\big]$. 

We examine the implementation of $(\mu_0,\mu_1)$ in the following cases based on the values of $\underline{\alpha}_1$ and $\overline{\alpha}_0$. \\
\\
\textbf{Case a:} $\overline{\alpha}_0 > \alpha_1^0 $ and $\underline{\alpha}_1 < \alpha_0^1$.

In this case, $(\mu_0,\mu_1)$ is implementable by regular constant payments if and only if there exists $m_1\geq 0$ such that $\big[K_0'(\alpha_0)\cdot \alpha_1^0 - K_1(\alpha_1^0)\big] - \big[K_1'(\alpha_1)\cdot \alpha_1 - K_1(\alpha_1)\big]\leq m_1 \leq \big[K_0'(\alpha_0)\cdot \alpha_0 - K_0(\alpha_0)\big] - \big[K_1'(\alpha_1)\cdot \alpha_0^1 - K_0(\alpha_0^1)\big]$. According to Lemma \ref{APP-IMP-lm1}, this is impossible. \\
\\
\textbf{Case b:} $\overline{\alpha}_0 > \alpha_1^0$ and $\alpha_0^1\leq\underline{\alpha}_1 \leq \alpha_0$.

In this case, $(\mu_0,\mu_1)$ is implementable by regular constant payments if and only if there exists $m_1\geq 0$ such that $\big[K_0'(\alpha_0)\cdot \alpha_1^0 - K_1(\alpha_1^0)\big] - \big[K_1'(\alpha_1)\cdot \alpha_1 - K_1(\alpha_1)\big] \leq \big[K_0'(\alpha_0)\cdot \alpha_0 - K_0(\alpha_0)\big] - \big[K_1'(\alpha_1)\cdot \underline{\alpha}_1 - K_0(\underline{\alpha}_1)\big]$. This is impossible because $\big[K_0'(\alpha_0)\cdot \alpha_0 - K_0(\alpha_0)\big] - \big[K_1'(\alpha_1)\cdot \underline{\alpha}_1 - K_0(\underline{\alpha}_1)\big] \leq \big[K_0'(\alpha_0) - K_1'(\alpha_1)\big]\cdot \alpha_0 \leq \big[K_0'(\alpha_0) - K_1'(\alpha_1)\big]\cdot \alpha_1 < \big[K_0'(\alpha_0)\cdot \alpha_1^0 - K_1(\alpha_1^0)\big] - \big[K_1'(\alpha_1)\cdot \alpha_1 - K_1(\alpha_1)\big]$. \\
\\
\textbf{Case c:} $\overline{\alpha}_0 > \alpha_1^0$ and $\underline{\alpha}_1 > \alpha_0$. 

In this case, $(\mu_0,\mu_1)$ is implementable by regular constant payments if and only if there exists $m_1\geq 0$ such that $\big[K_0'(\alpha_0)\cdot \alpha_1^0 - K_1(\alpha_1^0)\big] - \big[K_1'(\alpha_1)\cdot \alpha_1 - K_1(\alpha_1)\big] \leq m_1 \leq  \big[K_0'(\alpha_0) - K_1'(\alpha_1)\big]\cdot \underline{\alpha}_1$. When $\underline{\alpha}_1 \geq \tilde{\alpha}_1$, this interval in non-empty, and $(\mu_0,\mu_1)$ can be implemented. When $\underline{\alpha}_1 < \tilde{\alpha}_1$, this interval is empty, implying that it is impossible to implement $(\mu_0,\mu_1)$ by regular constant payments.\\
\\
\textbf{Case d:} $\overline{\alpha}_0\leq \alpha_1^0$ and $\underline{\alpha}_1 <\alpha_0^1$. 

In this case, $(\mu_0,\mu_1)$ is implementable by regular constant payments if and only if there exists $m_1\geq 0$ such that $\big[K_0'(\alpha_0)\cdot \overline{\alpha}_0 - K_1(\overline{\alpha}_0)\big]-\big[K_1'(\alpha_1)\cdot \alpha_1 - K_1(\alpha_1)\big]\leq m_1 \leq \big[K_0'(\alpha_0)\cdot \alpha_0 - K_0(\alpha_0)\big] - \big[K_1'(\alpha_1)\cdot \alpha_0^1 - K_0(\alpha_0^1)\big]$. Such $m_1$ exists only when $\overline{\alpha}_0 \leq \check{\alpha}_0$. When $\overline{\alpha}_0 > \check{\alpha}_0$, $(\mu_0,\mu_1)$ can not be implemented by regular constant payments.\\
\\
\textbf{Case e:} $\overline{\alpha}_0 \leq \alpha_1^0$ and $\alpha_0^1 \leq \underline{\alpha}_1 \leq \alpha_0$. 

In this case, $(\mu_0,\mu_1)$ is implementable by regular constant payments if and only if there exists $m_1\geq 0$ such that $\big[K_0'(\alpha_0)\cdot \overline{\alpha}_0 - K_1(\overline{\alpha}_0)\big]-\big[K_1'(\alpha_1)\cdot \alpha_1 - K_1(\alpha_1)\big]\leq m_1 \leq \big[K_0'(\alpha_0)\cdot \alpha_0 - K_0(\alpha_0)\big] - \big[K_1'(\alpha_1)\cdot \underline{\alpha}_1 - K_0(\underline{\alpha}_1)\big]$. Such $m_1$ exists only when $\overline{\alpha}_0 \leq \check{\alpha}_0(\underline{\alpha}_1)$.\\
\\
\textbf{Case f:} $\overline{\alpha}_0 \leq \alpha_1^0$ and $\underline{\alpha}_1 > \alpha_0$.

In this case, $(\mu_0,\mu_1)$ is implementable by regular constant payments if and only if there exists $m_1\geq 0$ such that $\big[K_0'(\alpha_0)\cdot \overline{\alpha}_0 - K_1(\overline{\alpha}_0)\big]-\big[K_1'(\alpha_1)\cdot \alpha_1 - K_1(\alpha_1)\big]\leq m_1 \leq \big[K_0'(\alpha_0) - K_1'(\alpha_1)\big] \cdot \underline{\alpha}_1$. If $\underline{\alpha}_1 < \tilde{\alpha}_1$, such $m_1$ exists only when $\overline{\alpha}_0 \leq \tilde{\alpha}_0(\underline{\alpha}_1)$. If $\underline{\alpha}_1 \geq \tilde{\alpha}_1$, it is always possible to implement $(\mu_0,\mu_1)$ by regular constant payments. \\
\\
To sum, $(\mu_0,\mu_1)$ is implementable by regular constant payments if and only if $(\underline{\alpha}_1,\overline{\alpha}_0)$ falls in one of the following four cases: \circled{1} $\underline{\alpha}_1 < \alpha_0^1$ and $\overline{\alpha}_0\leq \check{\alpha}_0$; \circled{2} $\alpha_0^1 \leq \underline{\alpha}_1 \leq \alpha_0$ and $\overline{\alpha}_0 \leq \check{\alpha}_0(\underline{\alpha}_1)$; \circled{3} $\alpha_0 < \underline{\alpha}_1 < \tilde{\alpha}_1$ and $\overline{\alpha}_0 \leq  \tilde{\alpha}_0(\underline{\alpha}_1)$; \circled{4} $\underline{\alpha}_1 \geq \tilde{\alpha}_1$ and $\overline{\alpha}_0\leq \overline{c}$.  
\end{proof}
In Corollary \ref{APP-IMP-pp2}, we present the payment-minimizing menus that implement $(\mu_0,\mu_1)$ when $(\mu_0,\mu_1)$ is implementable by regular constant payments. The proof is similar to that of Lemma \ref{IMP-lm1} and is thus omitted. 
\begin{corollary}
Take a pair of output distributions $\mu = (\mu_0,\mu_1)\in \Delta([\underline{x},\overline{x}])^2$ with aggregate efforts $a = (\alpha_0,\alpha_1)$ satisfying $\alpha_0\leq \alpha_1$. Take $R':= (\text{supp}(\mu_0),\text{supp}(\mu_1))$. If $K_0'(\alpha_0)\leq K_1'(\alpha_1)$, define
\begin{equation*}
\begin{aligned}
    m_0' &:= \max_{\alpha\in [\underline{c},\overline{\alpha}_1]}\big[K_1'(\alpha_1)\cdot \alpha - K_0(\alpha)\big] - \big[K_0'(\alpha_0)\cdot \alpha_0 - K_0(\alpha_0)\big],\\
    m_1'&:= 0. 
\end{aligned}
\end{equation*}
If $K_0'(\alpha_0) > K_1'(\alpha_1)$ and $\overline{\alpha}_0\leq \tilde{\alpha}_0(\underline{\alpha}_1)$, define
\begin{equation*}
\begin{aligned}
    m_0'&:= 0,\\
    m_1'&:= \max\Big\{0, \max_{\alpha\in[\underline{\alpha}_0,\overline{\alpha}_0]}\big[K_0'(\alpha_0)\cdot \alpha - K_1(\alpha)\big] - \big[K_1'(\alpha_1)\cdot \alpha_1 - K_1(\alpha_1)\big]\Big\}.
\end{aligned}
\end{equation*}
Take $m':= (m_0',m_1')$. Under Assumption \ref{OM-as1}, the menu $(a,m',R',\mu)$ minimizes the principal's expected payment among all feasible menus that implement $(\mu_0,\mu_1)$. 
\label{APP-IMP-pp2}
\end{corollary}
\begin{corollary}
Take a pair of output distributions $\mu = (\mu_0,\mu_1)\in \Delta([\underline{x},\overline{x}])^2$ with aggregate efforts $a=(\alpha_0,\alpha_1)$ satisfying $\alpha_0\leq \alpha_1$ and $K_0'(\alpha_0) > K_1'(\alpha_1)$. Under Assumption \ref{OM-as1}, if $\tilde{\alpha}_0(\underline{\alpha}_1) < \overline{\alpha}_0 \leq \check{\alpha}_0(\underline{\alpha}_1)$, then for any feasible menu $(a,m,R,\mu)\in \mathcal{F}$ that implements $\mu$, it holds that $m_0>0$, $m_1>0$. 
\label{APP-IMP-pp3}
\end{corollary}
\begin{proof}
Take a pair of output distributions $\mu = (\mu_0,\mu_1)\in \Delta([\underline{x},\overline{x}])^2$ with aggregate efforts $a=(\alpha_0,\alpha_1)$ satisfying $\alpha_0\leq \alpha_1$ and $K_0'(\alpha_0) > K_1'(\alpha_1)$. According to Theorem \ref{IMP-pp2} and Proposition \ref{APP-IMP-pp1}, we have that for any feasible menu $(a,m,R,\mu)\in \mathcal{F}$ that implements $\mu$, it holds that $m_0>0$. It suffices to show that we also have $m_1 >0$. 

Suppose there exists a feasible menu $(a,m,R,\mu)\in \mathcal{F}$ such that $m_0 > 0$ and $m_1=0$. Following the same reasoning as in Lemma \ref{IMP-lm1}, the menu $(a,m,R',\mu)$ is also feasible, with $R':= (\text{supp}(\mu_0), \text{supp}(\mu_1))$. From \eqref{MBC-IC1}, we have that
\begin{equation}
-m_0 \geq \max_{\alpha\in [\underline{\alpha}_0,\overline{\alpha}_0]}\big[K_0'(\alpha_0)\cdot \alpha - K_1(\alpha)\big] - \big[K_1'(\alpha_1)\cdot \alpha_1 - K_1(\alpha_1)\big].
\label{IMP-eq29}
\end{equation}

The assumption that $\tilde{\alpha}_0(\underline{\alpha}_1) < \overline{\alpha}_0 \leq \check{\alpha}_0(\underline{\alpha}_1)$ implies that $\overline{\alpha}_0 > \check{\alpha}_0$. If $\overline{\alpha}_0 > \alpha_1^0$, the RHS of \eqref{IMP-eq29} is equal to $\big[K_0'(\alpha_0)\cdot \alpha_1^0 - K_1(\alpha_1^0)\big] - \big[K_1'(\alpha_1)\cdot \alpha_1 - K_1(\alpha_1)\big] \geq 0$. If $\check{\alpha}_0<\overline{\alpha}_0 \leq \alpha_1^0$, the RHS of \eqref{IMP-eq29} is equal to $\big[K_0'(\alpha_0)\cdot \overline{\alpha}_0 - K_1(\overline{\alpha}_0)\big] - \big[K_1'(\alpha_1)\cdot \alpha_1 - K_1(\alpha_1)\big] \geq \big[K_0'(\alpha_0)\cdot \check{\alpha}_0 - K_1(\check{\alpha}_0)\big] - \big[K_1'(\alpha_1)\cdot \alpha_1 - K_1(\alpha_1)\big] = \big[K_0'(\alpha_0)\cdot \alpha_0 - K_0(\alpha_0)\big] - \big[K_1'(\alpha_1)\cdot \alpha_0^1 - K_0(\alpha_0^1)\big]\geq 0$. To sum, the RHS of \eqref{IMP-eq29} is always non-negative. It contradicts $m_0>0$. Hence, it is necessary that $m_1 > 0$.  
\end{proof}

\section{Proofs in Section \ref{MBCF-om}}
\label{MBCF-app-om}
\textbf{Proof of Proposition \ref{OM-pp1}:} Take any optimal menu $(a^{*},m^{*},R^{*},\mu^{*})$ that solves problem \eqref{MBC-P}. Suppose that $K_0'(\alpha_0^{*}) < K_1'(\alpha_1^{*})$. Define $\overline{\alpha}_1^{*}:= \max \text{supp}(\mu_1^{*})$. According to Corollary \ref{APP-IMP-pp2}, the menu $(a^{*},m',R',\mu^{*})$, with $R': = (\text{supp}(\mu_0^{*}),\text{supp}(\mu_1^{*}))$, $m_0': = \max_{\alpha\in [\underline{c},\overline{\alpha}_1^{*}]}\big[K_1'(\alpha_1^{*})\cdot \alpha - K_0(\alpha)\big] - \big[K_0'(\alpha_0^{*})\cdot \alpha_0^{*} - K_0(\alpha_0^{*})\big]$, and $m_1': = 0$, is the payment-minimizing menu among all feasible menus that implement $\mu^{*}$. Hence, the menu $(a^{*},m',R',\mu^{*})$ is also optimal to \eqref{MBC-P}. 

Fix $\mu_1^{*}$ and consider any output distribution $\mu_0$ with aggregate effort $\alpha_0$ such that $K_0'(\alpha_0) < K_1'(\alpha_1^{*})$ (which implies $\alpha_0\leq \alpha_1^{*}$). Then $(a^{*},m',R',\mu^{*})$ is optimal among all the payment-minimizing menus, as described in Corollary \ref{APP-IMP-pp2}, that implement such $(\mu_0,\mu_1^{*})$. Formally, we have that $\alpha_0^{*}$ and $\mu_0^{*}$ optimally solve the following problem:
\begin{equation}
\begin{aligned}
\max_{\alpha_0\in [\underline{c},\overline{c}]:\atop K_0'(\alpha_0)< K_1'(\alpha_1^{*})} \max_{\mu_0\in \Delta([\underline{x},\overline{x}]), \atop m_0\in \mathbb{R}}  \quad & \int x\mu_0(dx) - K_0'(\alpha_0)\cdot \alpha_0 - m_0 \\
s.t. \quad & \int c(x)\mu_0(dx) = \alpha_0 \\
& m_0 = \max_{\alpha\in [\underline{c},\overline{\alpha}_1^{*}]}\big[K_1'(\alpha_1^{*})\cdot \alpha - K_0(\alpha)\big] - \big[K_0'(\alpha_0)\cdot \alpha_0 - K_0(\alpha_0)\big].
\end{aligned}
\label{OM-eq5}
\end{equation}

We can simply plug $m_0 = \max_{\alpha\in [\underline{c},\overline{r}_1]}\big[K_1'(\alpha_1^{*})\cdot \alpha - K_0(\alpha)\big] - \big[K_0'(\alpha_0)\cdot \alpha_0 - K_0(\alpha_0)\big]$ into the objective function, and then we have that $\alpha_0^{*}$ is the solution to the following problem
\begin{equation*}
    \max_{\alpha_0\in [\underline{c},\overline{c}]:\atop K_0'(\alpha_0)<K_1'(\alpha_1^{*})} \Theta(\alpha_0) -K_0(\alpha_0), 
\end{equation*}
whose first-order condition gives that (the differentiability of $\Theta(\cdot)$ follows from the differentiability of $\theta(\cdot)$)
\begin{equation}
\Theta'(\alpha_0^{*})\leq K_0'(\alpha_0^{*}). 
\label{OM-eq8}
\end{equation}

At the menu $(a^{*},m',R',\mu^{*})$, it is possible to extend the contract range of the type-1 agent to $[\underline{x},\overline{x}]$, and replace $\mu_1^{*}$ with any output distribution $\mu_1$ such that $\int c(x)\mu_1(dx)= \alpha_1^{*}$ without harming the feasibility of the menu. It is because those modifications would leave the RHS of \eqref{MBC-IC0} unchanged: $\max_{\alpha\in [\underline{c},\overline{\alpha}_1^{*}]}\big[K_1'(\alpha_1^{*})\cdot \alpha - K_0(\alpha)\big] = \max_{\alpha\in [\underline{c},\overline{c}]}\big[K_1'(\alpha_1^{*})\cdot \alpha - K_0(\alpha)\big]$. Among all such $\mu_1$, the principal would optimally recommend an output distribution to the type-1 agent such that the concavification value $\Theta(\alpha_1^{*})$ is attained, which indicates that $\int x\mu_1^{*}(dx) = \Theta(\alpha_1^{*})$ from the optimality of $\mu_1^{*}$. Consequently, the profit obtained by the principal from the type-1 agent is $\Theta(\alpha_1^{*}) - K_1'(\alpha_1^{*})\cdot \alpha_1^{*}$, whose first derivative at $\alpha_1^{*}$ is given by $\Theta'(\alpha_1^{*}) - K_1'(\alpha_1^{*}) - K_1''(\alpha_1^{*})\cdot \alpha_1^{*}$ and satisfies that
\begin{equation*}
    \Theta'(\alpha_1^{*}) - K_1'(\alpha_1^{*}) - K_1''(\alpha_1^{*})\cdot \alpha_1^{*} \leq \Theta'(\alpha_0^{*}) - K_1'(\alpha_1^{*}) <\Theta'(\alpha_0^{*}) - K_0'(\alpha_0^{*})\leq 0,
\end{equation*}
where the first inequality follows from $\Theta'(\alpha_0^{*})\geq \Theta'(\alpha_1^{*})$, the second inequality follows from $K_0'(\alpha_0^{*})<K_1'(\alpha_1^{*})$, and the third inequality follows from \eqref{OM-eq8}. This implies that $\alpha_1^{*}> \alpha_1^{MH}$. From Assumption \ref{OM-as3}, we have that the profit $\Theta(\alpha_1)- K_1'(\alpha_1)\cdot \alpha_1$ strictly increases if $\alpha_1$ decreases at $\alpha_1 = \alpha_1^{*}$. 

We can fix $\alpha_0^{*}$ and slightly decrease $\alpha_1^{*}$ to $\hat{\alpha}_1$ such that $K_1'(\alpha_1^{*}) > K_1'(\hat{\alpha}_1) > K_0'(\alpha_0^{*})$. Take any $\hat{\mu}_1$ satisfying that $\int c(x)\hat{\mu}_1(dx) =\hat{\alpha}_1$ and $\int x\hat{\mu}_1(dx) = \Theta(\hat{\alpha}_1)$. Take the payment-minimizing menu that implements $(\mu_0^{*},\hat{\mu}_1)$. Under the new menu, we have that the profit from the type-1 agent strictly increases while the constant payment to the type-0 agent strictly decreases. Overall, the principal's expected profit is strictly improved, which is contradictory to the optimality of $(a^{*},m',R',\mu^{*})$. 

In conclusion, we have that $K_0'(\alpha_0^{*}) \geq K_1'(\alpha_1^{*})$ at the optimal menu. $\hfill\square$\\
\\
\textbf{Proof of Proposition \ref{OM-pp2}:} Take $s_0(\cdot)$, $s_1(\cdot)$ as the contracts implemented by the menu $(\alpha,m,R,\mu)$ and let $\underline{r}_t:= \inf c(R_t)$, $\overline{r}_t := \sup c(R_t)$ for $t\in\{0,1\}$.

We first show that the IC constraint for type-0 agent still holds under the menu $(\alpha,m, R',\mu)$. Consider contract $s'_1(\cdot)$ pinned down by $\alpha_1$, $m_1$ and $R'_1$. Then $\widehat{s_1\circ\theta}$ coincides with $\widehat{s'_1\circ\theta}$ on interval $[\underline{c},\overline{r}_1]$. On interval $\alpha\in [\overline{r}_1,\overline{c}]$,
\begin{equation*}
\begin{aligned}
\widehat{s'_1\circ\theta}(\alpha) - K_0(\alpha)  &= K_1'(\alpha_1)\cdot \alpha - K_0(\alpha) + m_1
\leq K_1'(\alpha_1)\cdot \overline{r}_1 - K_0(\overline{r}_1) + m_1\\
&= \widehat{s_1\circ \theta}(\overline{r}_1) - K_0(\overline{r}_1),
\end{aligned}
\end{equation*}
where the inequality follows from that, under Assumption \ref{OM-as1}, $K_1'(\alpha_1) \leq K_0'(\alpha_1) \leq K_0'(\alpha)$ for $\alpha \in [\overline{r}_1,\overline{c}]$. Hence, we have that
\begin{equation*}
\max_{\alpha \in [\underline{c},\overline{c}]}\big[\widehat{s'_1\circ\theta}(\alpha) - K_0(\alpha)\big] \leq \max_{\alpha\in[\underline{c},\overline{c}]}\big[ \widehat{s_1\circ\theta}(\alpha) - K_0(\alpha)\big] \leq K_0'(\alpha_0)\cdot \alpha_0 - K_0(\alpha_0). 
\end{equation*}

Next, we show that the IC constraint for type-1 agent still holds under the menu $(\alpha,m,R',\Gamma)$. Consider contract $s_0'(\cdot)$ pinned down by $\alpha_0$, $m_0$ and $R_0'$. Then $\widehat{s_0\circ\theta}$ coincides with $\widehat{s_0'\circ\theta}$ on interval $[\underline{r}_0,\overline{c}]$. On interval $\alpha \in [\underline{c},\underline{r}_0]$,
\begin{equation*}
\begin{aligned}
\widehat{s_0'\circ\theta}(\alpha) - K_1(\alpha) &= K_0'(\alpha_0)\cdot \alpha - K_1(\alpha) + m_0 \leq K_0'(\alpha_0)\cdot \underline{r}_0 - K_1(\underline{r}_0) + m_0\\
&= \widehat{s_0\circ \theta}(\underline{r}_0) - K_1(\underline{r}_0),
\end{aligned}
\end{equation*}
where the inequality follows from that, under Assumption \ref{OM-as1}, $K_0'(\alpha_0)\geq K_1'(\alpha_0)\geq K_1'(\alpha)$ for all $\alpha \in [\underline{c},\underline{r}_0]$.

Thereby, from \eqref{MBC-IC1}, we have that
\begin{equation*}
\max_{\alpha \in [\underline{c},\overline{c}]}\big[ \widehat{s'_0\circ \theta}(\alpha) - K_1(\alpha)\big] \leq\max_{\alpha \in [\underline{c},\overline{c}]}\big[ \widehat{s_0^{*}\circ \theta}(\alpha) - K_1(\alpha)\big] \leq K'_1(\alpha_1^{*})\cdot \alpha_1^{*} - K_1(\alpha_1^{*}). 
\end{equation*}

Thus, the menu $(\alpha,m,R',\Gamma)$ satisfies the IC constraints and is thereby feasible. $\hfill\square$\\
\\
\textbf{Proof of Theorem \ref{OM-pp3}:} Take any optimal menu $(a^{*},m^{*},R^{*},\mu^{*})$ that solves \eqref{MBC-P}. We show the sufficiency and necessity of $K_0'(\alpha_0^{*}) = K_1'(\alpha_1^{*})$.\\
\\
\textit{Sufficiency:} We first show that given $K_0'(\alpha_0^{*}) = K_1'(\alpha_1^{*})$, the menu $(\alpha^{*},m^{*},R_F,\mu^{*})$ satisfies the IC constraints, and is thereby feasible.

Consider the menu $(\alpha^{*},m,R_F,\mu^{*})$ with $m = (m_0,m_1) = (0,0)$. Take the contract $s_t(\cdot)$ ($t\in \{0,1\}$) pinned down by $\alpha_t^{*}$, $m_t = 0$ and the full range. Apparently the IC constraints are satisfied for $(\alpha^{*},m,R_F,\mu^{*})$, due to that $\widehat{s_0\circ\theta}(\cdot) = \widehat{s_1\circ\theta}(\cdot)$. Hence, $(\alpha^{*},m,R_F,\mu^{*})$ is feasible. Because of the optimality of $(a^{*},m^{*},R^{*},\mu^{*})$, we have that $p_0m_0^{*} + p_1m_1^{*}\leq p_0m_0 + p_1m_1 = 0$. By $m_0^{*},m_1^{*}\geq 0$, it holds that $m^{*} = (0,0) = m$. Hence, $(\alpha^{*},m^{*},R_F,\mu^{*})$ is feasible. \\
\\
\textit{Necessity: }We aim to show that if there is $K_0'(\alpha_0^{*}) \neq K_1'(\alpha_1^{*})$, then $(\alpha^{*},m^{*},R_F,\mu^{*})$ cannot be feasible. Suppose, contrary to this, that $K_0'(\alpha_0^{*}) \neq K_1'(\alpha_1^{*})$ and yet $(\alpha^{*},m^{*},R_F,\mu^{*})$ remains feasible.

By Proposition \ref{OM-pp1}, we have that $K_0'(\alpha_0^{*}) >  K_1'(\alpha_1^{*})$. The menu $(\alpha^{*},m^{*},R_F,\mu^{*})$ satisfies the IC constraints 
\begin{equation}
\begin{aligned}
    K_0'(\alpha_0^{*})\cdot \alpha_0^{*} - K_0(\alpha_0^{*})+m_0^{*}&\geq \max_{\alpha\in [\underline{c},\overline{c}]}\big[K_1'(\alpha_1^{*})\cdot \alpha - K_0(\alpha)\big] + m_1^{*},\\
    K_1'(\alpha_1^{*})\cdot \alpha_1^{*} - K_1(\alpha_1^{*}) +m_1^{*}&\geq \max_{\alpha\in [\underline{c},\overline{c}]}\big[K_0'(\alpha_0^{*})\cdot \alpha - K_1(\alpha)\big] + m_0^{*}.
\end{aligned} 
\label{OM-eq10}
\end{equation}

Define $\alpha_0^{1*}$ and $\alpha_1^{0*}$ as in \eqref{IMP-eq12} and \eqref{IMP-eq13}:
\begin{equation*}
\alpha_0^{1*}: =\left\{
\begin{array}{cl}
  k_0^{-1}\big(K_1'(\alpha_1^{*})\big)   &\text{ if }K_1'(\alpha_1^{*})\geq K_0'(\underline{c}), \\
  \underline{c}   &\text{ if } K_1'(\alpha_1^{*})< K_0'(\underline{c}),
\end{array}
\right.
\end{equation*}
\begin{equation*}
\alpha_1^{0*}: =\left\{
\begin{array}{cl}
  k_1^{-1}\big(K_0'(\alpha_0^{*})\big)   &\text{ if }K_0'(\alpha_0^{*})\leq K_1'(\overline{c}), \\
  \overline{c}   &\text{ if } K_0'(\alpha_0^{*})> K_1'(\overline{c}).
\end{array}
\right.
\end{equation*}

Then the IC constraints \eqref{OM-eq10} implies that
\begin{equation*}
\begin{aligned}
&\big[K_0'(\alpha_0^{*})\cdot \alpha_0^{*} - K_0(\alpha_0^{*})\big] - \big[K_1'(\alpha_1^{*})\cdot \alpha_0^{1*} - K_0(\alpha_0^{1*}) \big] \\
\geq&\big[ K_0'(\alpha_0^{*})\cdot \alpha_1^{0*} - K_1(\alpha_1^{0*})\big] - \big[K_1'(\alpha_1^{*})\cdot \alpha_1^{*} - K_1(\alpha_1^{*})\big].
\end{aligned}
\label{OM-eq11}
\end{equation*}

However, following Lemma \ref{APP-IMP-lm1} and the fact that $K_0'(\alpha_0^{*}) > K_1'(\alpha_1^{*})$, we have $\big[K_0'(\alpha_0^{*})\cdot \alpha_0^{*} - K_0(\alpha_0^{*})\big] - \big[K_1'(\alpha_1^{*})\cdot \alpha_0^{1*} - K_0(\alpha_0^{1*}) \big] < \big[ K_0'(\alpha_0^{*})\cdot \alpha_1^{0*} - K_1(\alpha_1^{0*})\big] - \big[K_1'(\alpha_1^{*})\cdot \alpha_1^{*} - K_1(\alpha_1^{*})\big]$, which is contradictory.  

In conclusion, if there is $K_0'(\alpha_0^{*})\neq K_1'(\alpha_1^{*})$, then $(a^{*},m^{*},R_F,\mu^{*})$ must be infeasible. $\hfill\square$

\section{Proofs in Section \ref{MBCF-cef}}
\label{MBCF-app-cef}
\textbf{Proof of Lemma \ref{EF-lm1}:} It suffices to show that the menu $(a^{*},m^{*},R',\mu')$ is feasible. It straightforwardly holds that $\text{supp}(\delta_{\theta(\alpha_0^{*})}) \subseteq \{\theta(\alpha_0^{*})\}$, $\text{supp}(\delta_{\theta(\alpha_1^{*})})\subseteq \{\theta(\alpha_1^{*})\}$, $\int c(x)\delta_{\theta(\alpha_0^{*})}(dx) = \alpha_0^{*}$, and $\int c(x)\delta_{\theta(\alpha_1^{*})}(dx) = \alpha_1^{*}$. Take contracts $s_0'(\cdot)$, $s_1'(\cdot)$ such that $s_0'\circ \theta(\alpha) = \mathbb{1}_{\{\alpha = \alpha_0^{*}\}}\big[K_0'(\alpha_0^{*})\cdot \alpha + m_0^{*}\big]$ and $s_1'\circ \theta(\alpha) = \mathbb{1}_{\{\alpha = \alpha_1^{*}\}}\big[K_1'(\alpha_1^{*})\cdot \alpha + m_1^{*}\big]$. Then \eqref{MBC-MH} and \eqref{MBC-LL} hold for the menu $(a^{*},m^{*},R',\mu')$ given $s_0'(\cdot)$ and $s_1'(\cdot)$. We only need to check whether it satisfies incentive compatibility. 

Take $s_0^{*}(\cdot)$ and $s_1^{*}(\cdot)$ as the contracts implemented by the menu $(a^{*},m^{*},R^{*},\mu^{*})$, i.e., $s_0^{*}\circ \theta(\alpha)=\mathbb{1}_{\{\alpha \in c(R_0^{*})\}}\big[K_0'(\alpha_0^{*})\cdot \alpha + m_0^{*}\big]$, $s_1^{*}\circ \theta(\alpha)=\mathbb{1}_{\{\alpha \in c(R_1^{*})\}}\big[K_1'(\alpha_1^{*})\cdot \alpha + m_1^{*}\big]$. It holds that $s_0'\circ \theta(\alpha)\leq s_0^{*}\circ \theta(\alpha)$ and $s_1'\circ \theta(\alpha)\leq s_1^{*}\circ \theta(\alpha)$ for all $\alpha \in [\underline{c},\overline{c}]$. Since the menu $(a^{*},m^{*},R^{*},\mu^{*})$ satisfies incentive compatibility, we have that 
\begin{equation*}
\begin{aligned}
\max_{\alpha \in [\underline{c},\overline{c}]} \big[ \widehat{s_1'\circ \theta}(\alpha) - K_0(\alpha))\big] &\leq \max_{\alpha \in [\underline{c},\overline{c}]} \big[ \widehat{s_1^{*}\circ \theta}(\alpha) - K_0(\alpha))\big]\leq K_0'(\alpha_0^{*})\cdot \alpha_0^{*} - K_0(\alpha_0^{*})+m_0^{*},\\
\max_{\alpha \in [\underline{c},\overline{c}]} \big[ \widehat{s_0'\circ \theta}(\alpha) - K_1(\alpha))\big] &\leq \max_{\alpha \in [\underline{c},\overline{c}]} \big[ \widehat{s_0^{*}\circ \theta}(\alpha) - K_1(\alpha))\big]\leq K_1'(\alpha_1^{*})\cdot \alpha_1^{*} - K_1(\alpha_1^{*}) + m_1^{*}.
\end{aligned}
\end{equation*}

Hence, the menu $(a^{*},m^{*},R',\mu')$ also satisfies incentive compatibility and is thus feasible to problem \eqref{MBC-P}. $\hfill\square$\\
\\
\textbf{Proof of Lemma \ref{EF-lm2}:} Take an optimal menu $(a^{*},m^{*},R^{*},\mu^{*})$ to \eqref{MBC-P}. Take $R':= (\big\{\theta(\alpha_0^{*})\big\},\big\{\theta(\alpha_1^{*})\big\})$, and $\mu':= (\delta_{\theta(\alpha_0^{*})},\delta_{\theta(\alpha_1^{*})})$. According to Lemma \ref{EF-lm1}, the menu $(a^{*},m^{*},R',\mu')$ is also optimal to \eqref{MBC-P}. From Proposition \ref{OM-pp1}, under Assumption \ref{OM-as1}, we have that $\alpha_0^{*} \leq \alpha_1^{*}$ and $K_0'(\alpha_0^{*})\geq K_1'(\alpha_1^{*})$. Take $m_0':= 0$ and $m_1':= \max\Big\{0, \big[K_0'(\alpha_0^{*})\cdot \alpha_0^{*} - K_1(\alpha_0^{*})\big] - \big[K_1'(\alpha_1^{*})\cdot \alpha_1^{*} - K_1(\alpha_1^{*})\big]\Big\}$, and $m':= (m_0',m_1')$. By Lemma \ref{IMP-lm1}, the menu $(a^{*},m',R',\mu')$ is the payment-minimizing menu that implements $a^{*}$. Consequently, $(a^{*},m',R',\mu')$ also optimally solves \eqref{MBC-P}. It implies that $m^{*} = m'$ by the optimality of $(a^{*},m^{*},R',\mu')$.

Hence, the optimal value of \eqref{MBC-P} is unchanged if we focus on designing aggregate efforts $a=(\alpha_0,\alpha_1)$ satisfying $\alpha_0\leq \alpha_1$ and $K_0'(\alpha_0)\geq K_1'(\alpha_1)$ and adopting the payment-minimizing menu, as described in Lemma \ref{IMP-lm1}, to implement $a$. It directly leads to \eqref{EF-eq1}.$\hfill\square$\\
\\
\textbf{Proof of Theorem \ref{EF-th1}:} \textit{Sufficiency:} Suppose we have that $K_0'(\alpha_0^{MH})\leq K_1'(\alpha_1^{MH})$. Given the optimal menu $(a^{*},m^{*},R^{*},\mu^{*})$, according to Lemma \ref{EF-lm2}, $(\alpha_0^{*},\alpha_1^{*},m_1^{*})$ optimally solves problem \eqref{EF-eq1}. Suppose the statement is untrue and there is $K_0'(\alpha_0^{*})>K_1'(\alpha_1^{*})$. Consider the following two cases.\\
\\
\textbf{Case 1: }$K_0'(\alpha_0^{*})\leq K_0'(\alpha_0^{MH})$.

In this case, there is $K_1'(\alpha_1^{*}) < K_0'(\alpha_0^{*})\leq K_0'(\alpha_0^{MH})\leq K_1'(\alpha_1^{MH})$. Since $(\alpha_0^{*},\alpha_1^{*},m_1^{*})$ optimally solves \eqref{EF-eq1}, we have $m_1^{*} = \max\Big\{0,\big[K_0'(\alpha_0^{*})\cdot \alpha_0^{*} - K_1(\alpha_0^{*})\big] - \big[K_1'(\alpha_1^{*})\cdot \alpha_1^{*} - K_1(\alpha_1^{*})\big]\Big\}$. Increase $\alpha_1^{*}$ to $\hat{\alpha}_1 > \alpha_1^{*}$ such that $\alpha_1^{*} <\hat{\alpha}_1< \alpha_1^{MH}$,  $\alpha_0^{*} \leq \hat{\alpha}_1$ and $K_0'(\alpha_0^{*})\geq K_1'(\hat{\alpha}_1)$ are satisfied. Take $\hat{m}_1: =  \max\Big\{0,\big[K_0'(\alpha_0^{*})\cdot \alpha_0^{*} - K_1(\alpha_0^{*})\big] - \big[K_1'(\hat{\alpha}_1)\cdot \hat{\alpha}_1 - K_1(\hat{\alpha}_1)\big]\Big\}$. We have that $\hat{m}_1\leq m_1^{*}$ since $K_1'(\hat{\alpha}_1)\cdot \hat{\alpha}_1 - K_1(\hat{\alpha}_1)> K_1'(\alpha_1^{*})\cdot \alpha_1^{*} - K_1(\alpha_1^{*})$. $(\alpha_0^{*},\hat{\alpha}_1, \hat{m}_1)$ is also feasible to problem \eqref{EF-eq1}. The profit from $(\alpha_0^{*},\hat{\alpha}_1, \hat{m}_1)$ is strictly higher than that from $(\alpha_0^{*},\alpha_1^{*},m_1^{*})$, due to that $\Theta(\hat{\alpha}_1) - K_1'(\hat{\alpha}_1)\cdot \hat{\alpha}_1 - \hat{m}_1 > \Theta(\alpha_1^{*}) - K_1'(\alpha_1^{*})\cdot \alpha_1^{*} - m_1^{*}$. It is contradictory to the optimality of $(\alpha_0^{*},\alpha_1^{*},m_1^{*})$.\\
\\
\textbf{Case 2: }$K_0'(\alpha_0^{*})>K_0'(\alpha_0^{MH})$.

In this case, decrease $\alpha_0^{*}$ to $\hat{\alpha}_0<\alpha_0^{*}$ such that $\alpha_0^{*} > \hat{\alpha}_0>\alpha_0^{MH}$, $\hat{\alpha}_0 \leq \alpha_1^{*}$ and $K_0'(\hat{\alpha}_0)\geq K_1'(\alpha_1^{*})$ are satisfied. Take $\hat{m}_1 := \max\Big\{0,\big[K_0'(\hat{\alpha}_0)\cdot \hat{\alpha}_0 - K_1(\hat{\alpha}_0)\big] - \big[K_1'(\alpha_1^{*})\cdot \alpha_1^{*} - K_1(\alpha_1^{*})\big]\Big\}$. We have that $\hat{m_1}\leq m_1^{*}$ since $K_0'(\hat{\alpha}_0)\cdot \hat{\alpha}_0-K_1(\hat{\alpha}_0) < K_0'(\alpha_0^{*})\cdot \alpha_0^{*}-K_1(\alpha_0^{*})$. $(\hat{\alpha}_0,\alpha_1^{*},\hat{m}_1)$ is also feasible to problem \eqref{EF-eq1} and provides a strictly higher profit than $(\alpha_0^{*},\alpha_1^{*},m_1^{*})$, due to that $\Theta(\hat{\alpha}_0) - K_0'(\hat{\alpha}_0)\cdot \hat{\alpha}_0 > \Theta(\alpha_0^{*})-K_0'(\alpha_0^{*})\cdot \alpha_0^{*}$ and that $\hat{m_1}\leq m_1^{*}$. Again, it is contradictory to the optimality of $(\alpha_0^{*},\alpha_1^{*},m_1^{*})$. \\

In conclusion, both cases are not possible under our assumption. Hence, it must hold that $K_0'(\alpha_0^{*}) = K_1'(\alpha_1^{*})$. \\
\\
\textit{Necessity: }Suppose we have that $K_0'(\alpha_0^{MH}) > K_1'(\alpha_1^{MH})$ and yet there is $K_0'(\alpha_0^{*}) = K_1'(\alpha_1^{*})$, which indicates that either $\alpha_0^{*} < \alpha_1^{*}$ or $\alpha_0^{*} = \alpha_1^{*} = \underline{c}$. If $\alpha_0^{*} < \alpha_1^{*}$, then we have
\begin{equation*}
\begin{aligned}
&\big[K_0'(\alpha_0^{*})\cdot \alpha_0^{*} - K_1(\alpha_0^{*})\big] - \big[K_1'(\alpha_1^{*})\cdot \alpha_1^{*} -K_1(\alpha_1^{*})\big] \\
= & K_1'(\alpha_1^{*})\cdot (\alpha_0^{*}-\alpha_1^{*}) + \int_{\alpha_0^{*}}^{\alpha_1^{*}}K_1'(t)dt \\
= & \int_{\alpha_0^{*}}^{\alpha_1^{*}}\big[K_1'(t)-K_1'(\alpha_1^{*})\big]dt < 0.  
\end{aligned}
\end{equation*}

According to Lemma \ref{EF-lm2}, $(\alpha_0^{*},\alpha_1^{*},m_1^{*})$ optimally solves \eqref{EF-eq1}. Hence, $m_1^{*} = \max\Big\{0,\big[K_0'(\alpha_0^{*})\cdot \alpha_0^{*} - K_1(\alpha_0^{*})\big] - \big[K_1'(\alpha_1^{*})\cdot \alpha_1^{*} - K_1(\alpha_1^{*})\big]\Big\} = 0$.  While if $\alpha_0^{*} = \alpha_1^{*} = \underline{c}$, the corresponding $m_1^{*}$ is also equal to 0. Consider the following three cases.\\
\\
\textbf{Case 1: }$\alpha_0^{*}< \alpha_1^{*}$ and $K_1'(\alpha_1^{*}) \leq K_1'(\alpha_1^{MH})$.

In this case, we have that $K_0'(\alpha_0^{*}) = K_1'(\alpha_1^{*})\leq K_1'(\alpha_1^{MH}) < K_0'(\alpha_0^{MH})$. We can increase $\alpha_0^{*}$ to $\hat{\alpha}_0>\alpha_0^{*}$ such that $\alpha_0^{*}< \hat{\alpha}_0 < \alpha_0^{MH}$, $\hat{\alpha}_0 \leq \alpha_1^{*}$ and $K_0'(\hat{\alpha}_0)\geq K_1'(\alpha_1^{*})$ are satisfied, and we require that $\hat{\alpha}_0$ is sufficiently close to $\alpha_0^{*}$ such that $\big[K_0'(\hat{\alpha}_0)\cdot \hat{\alpha}_0 - K_1(\hat{\alpha}_0)\big] - \big[K_1'(\alpha_1^{*})\cdot \alpha_1^{*} - K_1(\alpha_1^{*})\big] < 0$. Take $\hat{m}_1: = \max\Big\{0,\big[K_0'(\hat{\alpha}_0)\cdot \hat{\alpha}_0 - K_1(\hat{\alpha}_0)\big] - \big[K_1'(\alpha_1^{*})\cdot \alpha_1^{*} - K_1(\alpha_1^{*})\big]\Big\} = 0$. Then $(\hat{\alpha}_0,\alpha_1^{*},\hat{m}_1)$ is also feasible to problem \eqref{EF-eq1}. It yields a strictly higher profit than $(\alpha_0^{*},\alpha_1^{*},m_1^{*})$ since $\Theta(\hat{\alpha}_0) - K_0'(\hat{\alpha}_0)\cdot \hat{\alpha}_0 > \Theta(\alpha_0^{*})-K_0'(\alpha_0^{*})\cdot \alpha_0^{*}$ and that $\hat{m_1}=m_1^{*}$. It is contradictory to the optimality of $(\alpha_0^{*},\alpha_1^{*},m_1^{*})$. \\
\\
\textbf{Case 2: }$\alpha_0^{*}<\alpha_1^{*}$ and $K_1'(\alpha_1^{*}) > K_1'(\alpha_1^{MH})$.

In this case, we can decrease $\alpha_1^{*}$ to $\hat{\alpha}_1 < \alpha_1^{*}$ such that $\alpha_1^{*}>\hat{\alpha}_1>\alpha_1^{MH}$, $\alpha_0^{*} \leq \hat{\alpha}_1$, and $K_0'(\alpha_0^{*})\geq K_1'(\hat{\alpha}_1)$ are satisfied. We further require that $\hat{\alpha}_1$ is sufficiently close to $\alpha_1^{*}$ such that $\big[K_0'(\alpha_0^{*})\cdot \alpha_0^{*} - K_1(\alpha_0^{*})\big] - \big[K_1'(\hat{\alpha}_1)\cdot \hat{\alpha}_1 - K_1(\hat{\alpha}_1)\big]<0$. Take $\hat{m}_1: =  \max\Big\{0,\big[K_0'(\alpha_0^{*})\cdot \alpha_0^{*} - K_1(\alpha_0^{*})\big] - \big[K_1'(\hat{\alpha}_1)\cdot \hat{\alpha}_1 - K_1(\hat{\alpha}_1)\big]\Big\} = 0$. Then $(\alpha_0^{*},\hat{\alpha}_1,\hat{m}_1)$ is feasible to \eqref{EF-eq1} and yields a strictly higher profit than $(\alpha_0^{*},\alpha_1^{*},m_1^{*})$, due to that $\Theta(\hat{\alpha}_1) - K_1'(\hat{\alpha}_1)\cdot \hat{\alpha}_1 > \Theta(\alpha_1^{*}) - K_1'(\alpha_1^{*})\cdot \alpha_1^{*}$ and $\hat{m}_1 = m_1^{*}$. It is contradictory to the optimality of $(\alpha_0^{*},\alpha_1^{*},m_1^{*})$. \\
\\
\textbf{Case 3: }$\alpha_0^{*} = \alpha_1^{*} = \underline{c}$.

We first show that it must hold that $\alpha_1^{MH} > \underline{c}$. Suppose this is not the case and there is $\alpha_1^{MH} = \underline{c}$. From $\alpha_0^{*} = \alpha_1^{*} = \underline{c}$ we have that $K_0'(\underline{c}) = K_1'(\underline{c})$. Then there is $\alpha_0^{MH} > \underline{c} = \alpha_1^{MH}$ given $K_0'(\alpha_0^{MH})> K_1'(\alpha_1^{MH})$. From the definition of $\alpha_1^{MH}$, it holds that $\Theta'(\alpha_1^{MH})\leq  K_1'(\alpha_1^{MH})+K_1''(\alpha_1^{MH})\cdot \alpha_1^{MH} = K_1'(\alpha_1^{MH})$. As a result, we have that
\begin{equation*}
 \Theta'(\alpha_0^{MH}) \leq \Theta'(\alpha_1^{MH}) \leq K_1'(\alpha_1^{MH})< K_0'(\alpha_0^{MH})+ K_0''(\alpha_0^{MH})\cdot \alpha_0^{MH}.
\end{equation*}
It indicates that $\frac{d}{d\alpha}\big[\Theta(\alpha) - K_0'(\alpha)\cdot \alpha\big]\Big|_{\alpha=\alpha_0^{MH}} < 0$, which contradicts the definition of $\alpha_0^{MH}$. Therefore, we have that $\alpha_1^{MH}>\underline{c}$.

Since $K_0'(\alpha_0^{MH}) > K_1'(\alpha_1^{MH}) > K_1'(\underline{c}) = K_0'(\underline{c})$, we also have that $\alpha_0^{MH} > \underline{c}$. In this case we can increase $\alpha_0^{*}$ to $\hat{\alpha}_0 > \alpha_0^{*} = \underline{c}$ and $\alpha_1^{*}$ to $\hat{\alpha}_1 > \alpha_1^{*} = \underline{c}$ such that $\alpha_0^{*} < \hat{\alpha}_0 < \alpha_0^{MH}$, $\alpha_1^{*} < \hat{\alpha}_1 < \alpha_1^{MH}$ and $K_0'(\hat{\alpha}_0) = K_1'(\hat{\alpha}_1)$. $(\hat{\alpha}_0,\hat{\alpha}_1,0)$ is feasible to \eqref{EF-eq1}, and yields strictly higher profit than $(\alpha_0^{*},\alpha_1^{*},m_1^{*})$. It is contradictory to the optimality of $(\alpha_0^{*},\alpha_1^{*},m_1^{*})$.  \\

In conclusion, given that $K_0'(\alpha_0^{MH})> K_1'(\alpha_1^{MH})$, it must hold that $K_0'(\alpha_0^{*}) > K_1'(\alpha_1^{*})$. Hence, the necessity is established. $\hfill\square$ \\
\\
\textbf{Proof of Proposition \ref{EF-pp1}:} We first show that $\alpha_0^{MH}\leq \alpha_0^{FB}$ and $\alpha_1^{MH}\leq \alpha_1^{FB}$. For any $t\in \{0,1\}$, $\alpha_t^{MH}\leq \alpha_t^{FB}$ trivially holds if $\alpha_t^{FB} = \overline{c}$. When $\alpha_t^{FB} \in [\underline{c},\overline{c})$, suppose that $\alpha_t^{MH} > \alpha_t^{FB}$. Since $\alpha_t^{FB} = \argmax_{\alpha\in[\underline{c},\overline{c}]}S_t(\alpha) = \argmax_{\alpha\in [\underline{c},\overline{c}]}\big[\Theta(\alpha) - K_t(\alpha)\big]$, we have that $\Theta'(\alpha_t^{FB})\leq K_t'(\alpha_t^{FB})$. It follows that
\begin{equation}
    \Theta'(\alpha_t^{MH})\leq \Theta'(\alpha_t^{FB}) \leq K_t'(\alpha_t^{FB})< K_t'(\alpha_t^{MH}) + K_t''(\alpha_t^{MH})\cdot \alpha_t^{MH},
\label{EF-eq4}
\end{equation}
where the first and the last inequality both follow from the assumption that $\alpha_t^{MH}>\alpha_t^{FB}$. \eqref{EF-eq4} implies that $\frac{d}{d\alpha}\big[\Theta(\alpha) - K_t'(\alpha)\cdot \alpha\big]\Big|_{\alpha = \alpha_t^{MH}} = \Theta'(\alpha_t^{MH}) - K_t'(\alpha_t^{MH}) - K_t''(\alpha_t^{MH})\cdot \alpha_t^{MH} < 0$, which contradicts to $\alpha_t^{MH} = \argmax_{\alpha\in[\underline{c},\overline{c}]}\big[\Theta(\alpha) - K_t'(\alpha)\cdot \alpha\big]$. Hence, we have that $\alpha_t^{MH}\leq \alpha_t^{FB}$ for all $t\in \{0,1\}$. 

Furthermore, we can show that $\alpha_0^{FB}\leq \alpha_1^{FB}$. Suppose that this is not the case and we have $\alpha_0^{FB} > \alpha_1^{FB}$, then we have
\begin{equation*}
\Theta'(\alpha_0^{FB}) \leq \Theta'(\alpha_1^{FB})\leq K_1'(\alpha_1^{FB}) < K_1'(\alpha_0^{FB}) < K_0'(\alpha_0^{FB}),
\end{equation*}
where the first and the second-to-last inequality follow from our assumption $\alpha_0^{FB} > \alpha_1^{FB}$, and the last inequality follows from Assumption \ref{OM-as1}. Consequently, $\frac{d}{d\alpha}\big[\Theta(\alpha) - K_0(\alpha)\big]\Big|_{\alpha = \alpha_0^{FB}}<0$, which contradicts to $\alpha_0^{FB} = \argmax_{\alpha\in [\underline{c},\overline{c}]}\big[\Theta(\alpha)- K_0(\alpha)\big]$.\\
\\
\textit{Statement} \circled{1}: If $K_0'(\alpha_0^{MH})\leq K_1'(\alpha_1^{MH})$, according to Theorem \ref{EF-th1}, we have $K_0'(\alpha_0^{*}) = K_1'(\alpha_1^{*})$ at the optimal menu. 

First, we can establish that $K_0'(\alpha_0^{MH})\leq K_0'(\alpha_0^{*}) = K_1'(\alpha_1^{*})\leq K_1'(\alpha_1^{MH})$. Otherwise, if $K_0'(\alpha_0^{*}) = K_1'(\alpha_1^{*}) < K_0'(\alpha_0^{MH})$, then the principal can strictly improve her profit by slightly increasing $\alpha_0^{*}$, following similar arguments as in the proof of Theorem \ref{EF-th1}. While if $K_1'(\alpha_1^{MH}) < K_0'(\alpha_0^{*}) = K_1'(\alpha_1^{*})$, the principal can similarly get strictly better-off by slightly decreasing $\alpha_1$. Consequently, we have that $\alpha_0^{*}\geq \alpha_0^{MH}$ and $\alpha_1^{*}\leq \alpha_1^{MH}$. 

It suffices to show that $\alpha_0^{*}\leq \alpha_0^{FB}$. It trivially holds if $\alpha_0^{FB} = \overline{c}$. Suppose $\alpha_0^{FB} \in [\underline{c},\overline{c})$ and there is $\alpha_0^{*} > \alpha_0^{FB}$. From $\Theta'(\alpha_0^{FB})\leq K_0'(\alpha_0^{FB})$ and $\alpha_0^{*} > \alpha_0^{FB}$, we have that $\Theta'(\alpha_0^{*})< K_0'(\alpha_0^{*})$. Since $\alpha_0^{*} < \alpha_1^{*} \leq \alpha_1^{MH}$, we have
\begin{equation*}
   \Theta'(\alpha_1^{MH})\leq \Theta'(\alpha_0^{*})<K_0'(\alpha_0^{*})=K_1'(\alpha_1^{*})<K_1'(\alpha_1^{MH}) + K_1''(\alpha_1^{MH})\cdot \alpha_1^{MH}, 
\end{equation*}
which implies that $\frac{d}{d\alpha}\big[\Theta(\alpha) - K_1'(\alpha)\cdot \alpha\big]\Big|_{\alpha = \alpha_1^{MH}} = \Theta'(\alpha_1^{MH}) - K_1'(\alpha_1^{MH}) - K_1''(\alpha_1^{MH})\cdot \alpha_1^{MH} < 0$, which contradicts to $\alpha_1^{MH} = \argmax_{\alpha\in [\underline{c},\overline{c}]}\big[\Theta(\alpha) - K_1'(\alpha)\cdot \alpha\big]$. \\
\\
\textit{Statement} \circled{2}: If $K_0'(\alpha_0^{MH}) > K_1'(\alpha_1^{MH})$, according to Theorem \ref{EF-th1}, we have $K_0'(\alpha_0^{*})>K_1'(\alpha_1^{*})$ at the optimal menu. 

We can first establish that $K_1'(\alpha_1^{MH})\leq K_1'(\alpha_1^{*})<K_0'(\alpha_0^{*}) \leq K_0'(\alpha_0^{MH})$. Otherwise, if $K_0'(\alpha_0^{*}) > K_0'(\alpha_0^{MH})$, the principal can get strictly better-off by slightly decreasing $\alpha_0^{*}$. While if $K_1'(\alpha_1^{*}) < K_1'(\alpha_1^{MH})$, the principal can strictly benefit from slightly increasing $\alpha_1^{*}$. As a result, we have $\alpha_0^{*}\leq \alpha_0^{MH}$ and $\alpha_1^{*}\geq \alpha_1^{MH}$. 

It suffices to show that $\alpha_1^{*} \leq \alpha_1^{FB}$. It trivially holds if $\alpha_1^{FB} = \overline{c}$. Suppose $\alpha_1^{FB}\in [\underline{c},\overline{c})$ and $\alpha_1^{*} > \alpha_1^{FB}$. Then there is $\alpha_1^{*} >\alpha_1^{FB}\geq \alpha_1^{MH}$. We consider the following three cases based on the relationship between $\alpha_0^{*}$ and $\alpha_1^{*}$.  \\
\\
\textbf{Case 1: }$\alpha_0^{*}<\alpha_1^{*}$ and $\big[K_0'(\alpha_0^{*})\cdot \alpha_0^{*} - K_1(\alpha_0^{*})\big] -\big[K_1'(\alpha_1^{*})\cdot \alpha_1^{*} - K_1(\alpha_1^{*})\big] < 0$.

In this case, $m_1^{*} = \max\Big\{0, \big[K_0'(\alpha_0^{*})\cdot \alpha_0^{*} - K_1(\alpha_0^{*})\big] -\big[K_1'(\alpha_1^{*})\cdot \alpha_1^{*} - K_1(\alpha_1^{*})\big]\Big\} = 0$. Suppose the principal decreases $\alpha_1^{*}$ to $\hat{\alpha}_1$ such that $\alpha_1^{*} > \hat{\alpha}_1 > \alpha_1^{MH}$ and $\hat{\alpha}_1 > \alpha_0^{*}$. We also require that $\hat{\alpha}_1$ is sufficiently close to $\alpha_1^{*}$ such that $\hat{m}_1 : = \max\Big\{0, \big[K_0'(\alpha_0^{*})\cdot \alpha_0^{*} - K_1(\alpha_0^{*})\big] -\big[K_1'(\hat{\alpha}_1)\cdot \hat{\alpha}_1 - K_1(\hat{\alpha}_1)\big]\Big\}=0$. Then $(\alpha_0^{*},\hat{\alpha}_1,\hat{m}_1)$ is feasible to problem \eqref{EF-eq1} and $\Theta(\hat{\alpha}_1) - K_1'(\hat{\alpha}_1)\cdot \hat{\alpha}_1- \hat{m}_1> \Theta(\alpha_1^{*}) - K_1'(\alpha_1^{*})\cdot \alpha_1^{*}- m_1^{*}$. It is a contradiction to the optimality of $(\alpha_0^{*},\alpha_1^{*},m_1^{*})$ to \eqref{EF-eq1}. \\
\\
\textbf{Case 2: }$\alpha_0^{*}<\alpha_1^{*}$ and $\big[K_0'(\alpha_0^{*})\cdot \alpha_0^{*} - K_1(\alpha_0^{*})\big] -\big[K_1'(\alpha_1^{*})\cdot \alpha_1^{*} - K_1(\alpha_1^{*})\big] \geq 0$.

In this case, $m_1^{*} = \big[K_0'(\alpha_0^{*})\cdot \alpha_0^{*} - K_1(\alpha_0^{*})\big] -\big[K_1'(\alpha_1^{*})\cdot \alpha_1^{*} - K_1(\alpha_1^{*})\big] \geq 0$. Suppose the principal decreases $\alpha_1^{*}$ to $\hat{\alpha}_1$ such that $\alpha_1^{*} > \hat{\alpha}_1 > \alpha_1^{FB}$ and $\hat{\alpha}_1 > \alpha_0^{*}$. Define $\hat{m}_1:=\big[K_0'(\alpha_0^{*})\cdot \alpha_0^{*} - K_1(\alpha_0^{*})\big] -\big[K_1'(\hat{\alpha}_1)\cdot \hat{\alpha}_1 - K_1(\hat{\alpha}_1)\big]>0$. Then $(\alpha_0^{*},\hat{\alpha}_1,\hat{m}_1)$ is feasible to problem \eqref{EF-eq1}. The profit obtained from the type-1 agent under $(\alpha_0^{*},\alpha_1^{*},m_1^{*})$ is given by
\begin{equation}
\Theta(\alpha_1^{*}) - K_1'(\alpha_1^{*})\cdot \alpha_1^{*} - m_1^{*} = \big[\Theta(\alpha_1^{*}) - K_1(\alpha_1^{*})\big]- \big[K_1(\alpha_1^{*})\cdot \alpha_1^{*} - K_1(\alpha_1^{*})+m_1^{*}\big],   
\label{EF-eq5}
\end{equation}
and the profit obtained from the type-0 agent under $(\alpha_0^{*},\hat{\alpha}_1,\hat{m}_1)$ is given by
\begin{equation}
\Theta(\hat{\alpha}_1) - K_1'(\hat{\alpha}_1)\cdot \hat{\alpha}_1 - \hat{m}_1 = \big[\Theta(\hat{\alpha}_1) - K_1(\hat{\alpha}_1)\big]- \big[K_1(\hat{\alpha}_1)\cdot \hat{\alpha}_1 - K_1(\hat{\alpha}_1)+\hat{m}_1\big]. 
\label{EF-eq6}
\end{equation}

From $\alpha_1^{*} > \hat{\alpha}_1 \geq \alpha_1^{FB}$, the surplus from $\hat{\alpha}_1$ would be strictly higher than that from $\alpha_1^{*}$, i.e., $\Theta(\hat{\alpha}_1) - K_1(\hat{\alpha}_1) > \Theta(\alpha_1^{*}) - K_1(\alpha_1^{*})$. However, the payoff obtained by the type-1 agent would be the same, since $K_1(\alpha_1^{*})\cdot \alpha_1^{*} - K_1(\alpha_1^{*})+m_1^{*} = K_1(\hat{\alpha}_1)\cdot \hat{\alpha}_1 - K_1(\hat{\alpha}_1)+\hat{m}_1$. From \eqref{EF-eq5} and \eqref{EF-eq6}, we have that $\Theta(\hat{\alpha}_1) - K_1'(\hat{\alpha}_1)\cdot \hat{\alpha}_1 - \hat{m}_1 > \Theta(\alpha_1^{*}) - K_1'(\alpha_1^{*})\cdot \alpha_1^{*} - m_1^{*}$, which contradicts to the optimality of $(\alpha_0^{*},\alpha_1^{*},m^{*}_1)$ to \eqref{EF-eq1}.\\
\\
\textbf{Case 3: }$\alpha_0^{*} = \alpha_1^{*}$.

In this case, we have $\alpha_0^{*} = \alpha_1^{*}> \alpha_1^{FB}\geq \alpha_0^{FB}\geq \alpha_0^{MH}$. The principal can decrease $\alpha_0^{*}$ to $\hat{\alpha}_0$ such that $\alpha_0^{*} > \hat{\alpha}_0 > \alpha_0^{MH}$ and $K_0'(\hat{\alpha}_0) > K_1'(\alpha_1^{*})$. Define $\hat{m}_1:=\max\Big\{ 0,\big[K_0'(\hat{\alpha}_0)\cdot \hat{\alpha}_0 - K_1(\hat{\alpha}_0)\big] -\big[K_1'(\alpha_1^{*})\cdot \alpha_1^{*} - K_1(\alpha_1^{*})\big]\Big\}$. Then $(\hat{\alpha}_0,\alpha_1^{*},\hat{m}_1)$ is feasible to problem \eqref{EF-eq1}. Furthermore, it holds that $\hat{m}_1 \leq m_1^{*} = \max\Big\{ 0,\big[K_0'(\alpha_0^{*})\cdot \alpha_0^{*} - K_1(\alpha_0^{*})\big] -\big[K_1'(\alpha_1^{*})\cdot \alpha_1^{*} - K_1(\alpha_1^{*})\big]\Big\}$. We have that $(\hat{\alpha_0},\alpha_1^{*},\hat{m}_1)$ is feasible to \eqref{EF-eq1}, with $\Theta(\hat{\alpha}_0) - K_0'(\hat{\alpha}_0)\cdot\hat{\alpha}_0 > \Theta(\alpha_0^{*}) - K_0'(\alpha_0^{*})\cdot \alpha_0^{*}$ and $\Theta(a^{*}) - K_1'(\alpha_1^{*})\cdot \alpha_1^{*} - \hat{m}_1\geq \Theta(a^{*}) - K_1'(\alpha_1^{*})\cdot \alpha_1^{*} -m_1^{*}$. It contradicts the optimality of $(\alpha_0^{*},\alpha_1^{*},m_1^{*})$ to \eqref{EF-eq1}.\\

In conclusion, all three cases are impossible. We have shown that $\alpha_1^{*}\leq \alpha_1^{FB}$.$\hfill\square$
\section{Proofs in Section \ref{MBCF-os}}
\label{MBCF-app-os}
\textbf{Proof of Proposition \ref{WA-pp1}: }Take the optimal menu $(a^{*},m^{*},R^{*},\mu^{*})$ that consists of a single full-range contract. Then, $(\alpha_0^{*},\alpha_1^{*})$ solves the following problem:
\begin{equation}
\begin{aligned}
\max_{(\alpha_0,\alpha_1)} \quad& p_0\big[\Theta(\alpha_0)-K_0'(\alpha_0)\cdot \alpha_0\big] + p_1\big[\Theta(\alpha_1) - K_1'(\alpha_1)\cdot \alpha_1\big]\\
s.t. \quad& K_0'(\alpha_0) = K_1'(\alpha_1). 
\end{aligned}
\label{WA-eq1}
\end{equation}
\textit{Statement \circled{1}:} We can show that if $K_0'(\alpha_0^{MH}) \leq K_1'(\alpha_1^{MH})$, then it must hold that $K_0'(\alpha_0^{MH})\leq K_0'(\alpha_0^{*}) = K_1'(\alpha_1^{*}) \leq K_1'(\alpha_1^{MH})$. Otherwise, if $K_0'(\alpha_0^{*}) = K_1'(\alpha_1^{*}) < K_0'(\alpha_0^{MH})\leq K_1'(\alpha_1^{MH})$, we can increase $\alpha_0^{*}$ to $\hat{\alpha}_0$ and $\alpha_1^{*}$ to $\hat{\alpha}_1$ such that $\alpha_0^{*} < \hat{\alpha}_0 < \alpha_0^{MH}$, $\alpha_1^{*} < \hat{\alpha}_1 < \alpha_1^{MH}$, and $K_0'(\hat{\alpha}_0) = K_1'(\hat{\alpha}_1)$. Then $(\hat{\alpha}_0,\hat{\alpha}_1)$ is feasible to \eqref{WA-eq1} and yields a strictly higher profit than $(\alpha_0^{*},\alpha_1^{*})$, contradictionary to the optimality of $(\alpha_0^{*},\alpha_1^{*})$. Similarly, it is also impossible to have $K_0'(\alpha_0^{MH}) \leq K_1'(\alpha_1^{MH})< K_0'(\alpha_0^{*}) = K_1'(\alpha_1^{*})$. Hence, we've established that $K_0'(\alpha_0^{MH})\leq K_0'(\alpha_0^{*}) = K_1'(\alpha_1^{*}) \leq K_1'(\alpha_1^{MH})$, which implies that $\alpha_0^{*}\geq \alpha_0^{MH}$ and $\alpha_1^{*}\leq \alpha_1^{MH}$. 

It suffices to show that $\alpha_0^{*}\leq \alpha_0^{FB}$. Suppose that this is not the case and we have that $\alpha_0^{*} > \alpha_0^{FB}$. It implies that $\alpha_0^{FB} < \overline{c}$, $\alpha_1^{*} \geq \alpha_0^{*} > \underline{c}$, and 
\begin{equation*}
\begin{aligned}
    \Theta'(\alpha_0^{*}) &\leq \Theta'(\alpha_0^{FB}) \leq K_0'(\alpha_0^{FB}) < K_0'(\alpha_0^{*}),\\
\Longrightarrow \Theta'(\alpha_1^{*}) &\leq \Theta'(\alpha_0^{*}) < K_0'(\alpha_0^{*}) = K_1'(\alpha_1^{*}) < K_1'(\alpha_1^{*}) + K_1''(\alpha_1^{*})\cdot \alpha_1^{*},  
\end{aligned}
\end{equation*}
which implies that $\alpha_1^{MH} < \alpha_1^{*}$, contradicting $\alpha_1^{*}\leq \alpha_1^{MH}$. \\
\\
\textit{Statement \circled{2}:} If $K_0'(\alpha_0^{MH}) > K_1'(\alpha_1^{MH})$, we can similarly show that $K_0'(\alpha_0^{MH})\geq K_0'(\alpha_0^{*}) = K_1'(\alpha_1^{*})\geq K_1'(\alpha_1^{MH})$, which implies that $\alpha_0^{*}\leq \alpha_0^{MH}$, and $\alpha_1^{*}\geq \alpha_1^{MH}$. $\hfill\square$\\
\\
Before presenting the proof of Proposition \ref{WA-pp2}, we first prove the following lemma.
\begin{lemma}
Take $R_0, R_1 \subseteq [\underline{x},\overline{x}]$ such that $R_0$ is a right-truncated interval, and $R_1$ is a left-truncated interval, i.e., $\exists x_0,x_1 \in [\underline{x},\overline{x}]$ such that $R_0 = [\underline{x},x_0]$, $R_1 = [x_1,\overline{x}]$. Take $\Theta_{R_0}:c(R_0)\to \mathbb{R}$, $\Theta_{R_1}:c(R_1)\to \mathbb{R}$ as the concavification of $\theta(\cdot)$ over $c(R_0)$ and $c(R_1)$. It holds that
\begin{equation*}
\begin{aligned}
    \Theta_{R_0}^{'}(\alpha) &\leq \Theta'(\alpha),\quad \forall \alpha\in c(R_0)\\
    \Theta_{R_1}^{'}(\alpha) &\geq \Theta'(\alpha),\quad \forall \alpha \in c(R_0).
\end{aligned}
\end{equation*}
\label{WA-lm1}
\end{lemma}
\noindent\textbf{Proof of Lemma \ref{WA-lm1}: }Take any $\alpha \in c(R_0)$, we have that $\Theta_{R_0}(\alpha) \leq \Theta(\alpha)$. If $\Theta(\alpha) = \theta(\alpha)$, it implies that $\Theta(\alpha) = \Theta_{R_0}(\alpha) = \theta(\alpha)$, which leads to that $\Theta'(\alpha) = \Theta_{R_0}'(\alpha) = \theta'(\alpha)$. 

If $\Theta(\alpha) > \theta(\alpha)$, then there exist $u,u'\in [\underline{c},\overline{c}]$ such that $u < \alpha < u'$, $\Theta(u) = \theta(u)$, $\Theta(u') = \theta(u')$, and that $\Theta(\cdot)$ is affine on $[u,u']$. We have that $\Theta'(\alpha) = \theta'(u)$. Since $u<\alpha$ and that $R_0$ is a right-truncated interval, we have that $u\in c(R_0)$. At $u$ there is $\Theta(u) = \theta(u)$, it implies that $\Theta(u) = \Theta_{R_0}(u) = \theta(u)$ and $\Theta_{R_0}'(u) = \theta'(u) = \Theta'(\alpha)$. From the concavity of $\Theta_{R_0}(\cdot)$, we have that $\Theta_{R_0}'(\alpha) \leq \Theta_{R_0}'(u) = \Theta'(\alpha)$. 

Similarly we can establish that $\Theta'_{R_1}(\alpha)\geq \Theta'(\alpha)$ for all $\alpha \in c(R_1)$. $\hfill\square$\\
\\
\textbf{Proof of Proposition \ref{WA-pp2}: }Take an arbitrary optimal menu $(a^{*},m^{*},R^{*},\mu^{*})$ that solves \eqref{MBC-P}. It holds that $K_0'(\alpha_0^{*}) > K_1'(\alpha_1^{*})$. From Proposition \ref{OM-pp2}, it is without loss of generality to assume that $R_0^{*}$ is a right-truncated interval and that $R_1^{*}$ is a left-truncated interval. Take $\Theta_{R_0^{*}}:c(R_0^{*})\to \mathbb{R}$, $\Theta_{R_1^{*}}:c(R_1^{*})\to \mathbb{R}$ as the concavification of $\theta(\cdot)$ over $c(R_0^{*})$ and $c(R_1^{*})$. Take $\overline{r}_0^{*} := \sup c(R_0^{*})$ and $\underline{r}_1^{*} := \inf c(R_1^{*})$. The principal's profit from $(a^{*},m^{*},R^{*},\mu^{*})$ is given by
\begin{equation*}
    p_0\big[\Theta_{R_0^{*}}(\alpha_0^{*})-K_0'(\alpha_0^{*})\cdot \alpha_0^{*} - m_0^{*}\big] + p_1 \big[\Theta_{R_1^{*}}(\alpha_1^{*}) - K_1'(\alpha_1^{*})\cdot \alpha_1^{*} - m_1^{*}\big].
\end{equation*}
\textit{Statement \circled{1}:} We first show that $\alpha_0^{*}\leq \alpha_0^{FB}$. Suppose that this is not the case and we have $\alpha_0^{*} > \alpha_0^{FB}$. Take $s_0^{*}(\cdot)$, $s_1^{*}(\cdot)$ as the contracts pinned down by $a^{*}$, $m^{*}$ and $R^{*}$. $(a^{*},m^{*},R^{*},\mu^{*})$ satisfies the IC constraints:
\begin{subequations}
\begin{align}
K_0'(\alpha_0^{*})\cdot \alpha_0^{*} - K_0(\alpha_0^{*}) + m_0^{*} &\geq \max_{\alpha \in [\underline{c},\overline{c}]}\big[\widehat{s_1^{*}\circ \theta}(\alpha) - K_0(\alpha)\big], \label{WA-eq2a} \\
K_1'(\alpha_1^{*})\cdot \alpha_1^{*} - K_1(\alpha_1^{*}) + m_1^{*} &\geq \max_{\alpha\in [\underline{c},\overline{c}]}\big[\widehat{s_0^{*}\circ \theta}(\alpha) - K_1(\alpha)\big]. \label{WA-eq2b}
\end{align}
\end{subequations}

Fix $R_0^{*}$, $R_1^{*}$, $\alpha_1^{*}$ and $m_1^{*}$. Decrease $\alpha_0^{*}$ to $\hat{\alpha}_0$ (and adjust the output distribution accordingly such that $\Theta_{R_0^{*}}(\hat{\alpha}_0)$ is attained) such that $\alpha_0^{*} > \hat{\alpha}_0 > \alpha_0^{FB}$ and $K_0'(\alpha_0^{*})>K_0'(\hat{\alpha}_0) > K_1'(\alpha_1^{*})\geq K_1'(\alpha_0^{*})$. Take $\hat{m}_0$ such that 
\begin{equation}
    K_0'(\hat{\alpha}_0)\cdot \hat{\alpha}_0 - K_0(\hat{\alpha}_0) + \hat{m}_0 = K_0'(\alpha_0^{*})\cdot \alpha_0^{*} - K_0(\alpha_0^{*}) + m_0^{*}.
\label{WA-eq3}
\end{equation}

Increase $m_0^{*}$ to $\hat{m}_0$. Let $\hat{s}_0$ denote the contract that is pinned down by $\hat{\alpha}_0$, $\hat{m}_0$ and $R_0^{*}$. Under the new menu, \eqref{WA-eq2a} still holds since the LHS of \eqref{WA-eq2a} is unchanged. We aim to establish that the new menu also satisfies \eqref{WA-eq2b}. The RHS of \eqref{WA-eq2b} is decreased by $\max_{\alpha\in [\underline{c},\overline{c}]}\big[\widehat{s_0^{*}\circ \theta}(\alpha) - K_1(\alpha)\big] - \max_{\alpha\in [\underline{c},\overline{c}]}\big[\widehat{\hat{s}_0\circ \theta}(\alpha) - K_1(\alpha)\big] = \max_{[\underline{c},\overline{r}_0^{*}]}\big[K_0'(\alpha_0)^{*}\cdot \alpha - K_1(\alpha)+ m_0^{*}\big] -\max_{[\underline{c},\overline{r}_0^{*}]}\big[K_0'(\hat{\alpha}_0)\cdot \alpha - K_1(\alpha) + \hat{m}_0\big] = \Delta - (\hat{m}_0 - m_0^{*})$, where
\begin{equation}
    \Delta := \max_{[\underline{c},\overline{r}_0^{*}]}\big[K_0'(\alpha_0^{*})\cdot \alpha - K_1(\alpha) \big] -\max_{[\underline{c},\overline{r}_0^{*}]}\big[K_0'(\hat{\alpha}_0)\cdot \alpha - K_1(\alpha)\big].
\label{WA-eq4}
\end{equation}

Consider the following two cases based on the value of $\alpha_0^{*}$.\\
\\
\textbf{Case 1:} $K_0'(\alpha_0^{*}) \leq K_1'(\overline{r}_0^{*})$.

In this case, we have that $K_1'(\underline{c}) < K_0'(\alpha_0^{*}) \leq K_1'(\overline{r}_0^{*})$, and $K_1'(\underline{c}) < K_0'(\hat{\alpha}_0) <K_1'(\overline{r}_0^{*})$. There exists $\beta_0^{*}$ and $\hat{\beta}_0 \in [\underline{c},\overline{r}_0^{*}]$ such that $K_1'(\beta_0^{*}) = K_0'(\alpha_0^{*})$, $K_1'(\hat{\beta}_0) = K_0'(\hat{\alpha}_0)$ ($\beta_0^{*} > \hat{\beta}_0$). From \eqref{WA-eq4} we have that
\begin{equation}
\begin{aligned}
    \Delta &= \big[K_1'(\beta_0^{*})\cdot \beta_0^{*} - K_1(\beta_0^{*})\big] - \big[K_1'(\hat{\beta}_0)\cdot \hat{\beta}_0 - K_1(\hat{\beta}_0)\big]  \\
    &= \int_{\hat{\beta}_0}^{\beta_0^{*}}t dk_1(t) = \int_{k_0(\hat{\alpha}_0)}^{k_0(\alpha_0^{*})}k_1^{-1}(k)dk. 
\end{aligned}
\label{WA-eq5}
\end{equation}

On the other hand, from \eqref{WA-eq3} we have
\begin{equation}
\begin{aligned}
    \hat{m}_0 - m_0^{*} &= \big[K_0'(\alpha_0^{*})\cdot \alpha_0^{*} - K_0(\alpha_0^{*})\big] - \big[K_0'(\hat{\alpha}_0)\cdot \hat{\alpha}_0 - K_0(\hat{\alpha}_0)\big] \\
    &= \int_{k_0(\hat{\alpha}_0)}^{k_0(\alpha_0^{*})} k_0^{-1}(k)dk < \int_{k_0(\hat{\alpha}_0)}^{k_0(\alpha_0^{*})}k_1^{-1}(k)dk = \Delta. 
\end{aligned}
\end{equation}

Consequently, the RHS of \eqref{WA-eq2b} is reduced under the new menu. Hence, the new menu is feasible. \\
\\
\textbf{Case 2:} $K_0'(\alpha_0^{*}) > K_1'(\overline{r}_0^{*})$.

In this case, we further require that $K_0'(\hat{\alpha}_0) > K_1'(\overline{r}_0^{*})$ when choosing $\hat{\alpha}_0$. Then from \eqref{WA-eq4} we have that
\begin{equation*}
\begin{aligned}
    \Delta &= \big[K_0'(\alpha_0^{*})\cdot \overline{r}_0^{*} - K_1(\overline{r}_0^{*})\big] - \big[K_0'(\hat{\alpha}_0)\cdot \overline{r}_0^{*} - K_1(\overline{r}_0^{*})\big] \\
    &= \big[K_0'(\alpha_0^{*})-K_0'(\hat{\alpha}_0)\big]\cdot \overline{r}_0^{*} = \int_{\hat{\alpha}_0}^{\alpha_0^{*}}\overline{r}_0^{*}dk_0(t) \\
    &\geq \int_{\hat{\alpha}_0}^{\alpha_0^{*}}tdk_0(t) = \hat{m}_0 - m_0^{*},
\end{aligned}
\end{equation*}
where the inequality follows from that $\alpha_0^{*} \leq \overline{r}_0^{*}$. Hence, the RHS of \eqref{WA-eq2b} is reduced under the new menu, and the new menu is also feasible. \\
\\
We've shown that the new menu is always feasible. Under the new menu, the profit from the type-1 agent is unchanged while the profit from the type-0 agent changes from
\begin{equation}
    \Theta_{R_0^{*}}(\alpha_0^{*}) - K_0'(\alpha_0^{*})\cdot \alpha_0^{*} - m_0^{*} = \big[ \Theta_{R_0^{*}}(\alpha_0^{*}) - K_0(\alpha_0^{*})\big] - \big[K_0'(\alpha_0^{*})\cdot \alpha_0^{*} - K_0(\alpha_0^{*})+m_0^{*}\big],
\label{WA-eq7}
\end{equation}
to
\begin{equation}
 \Theta_{R_0^{*}}(\hat{\alpha}_0) - K_0'(\hat{\alpha}_0)\cdot \hat{\alpha}_0 - \hat{m}_0 = \big[ \Theta_{R_0^{*}}(\hat{\alpha}_0) - K_0(\hat{\alpha}_0)\big] - \big[K_0'(\hat{\alpha}_0)\cdot \hat{\alpha}_0 - K_0(\hat{\alpha}_0)+\hat{m}_0\big],
\label{WA-eq8}
\end{equation}

Consider the terms in the first pair of brackets on the RHS of \eqref{WA-eq7} and \eqref{WA-eq8}. We can establish that $\Theta_{R_0^{*}}(\alpha_0^{*}) - K_0(\alpha_0^{*}) < \Theta_{R_0^{*}}(\hat{\alpha}_0) - K_0(\hat{\alpha}_0)$ as follows. First, from $\alpha_0^{*}>\hat{\alpha}_0 > \alpha_0^{FB}$, we have that $\Theta(\hat{\alpha}_0) - K_0(\hat{\alpha}_0) > \Theta(\alpha_0^{*}) - K_0(\alpha_0^{*})$. Furthermore, from Lemma \ref{WA-lm1}, there is $\Theta_{R_0^{*}}'(\alpha) \leq \Theta'(\alpha)$ for all $\alpha\in c(R_0^{*})$. It implies that $\Theta_{R_0^{*}}(\hat{\alpha}_0) - \Theta(\hat{\alpha}_0) \geq \Theta_{R_0^{*}}(\alpha_0^{*}) - \Theta(\alpha_0^{*})$. As a result, we have that $\Theta_{R_0^{*}}(\alpha_0^{*}) - K_0(\alpha_0^{*}) < \Theta_{R_0^{*}}(\hat{\alpha}_0) - K_0(\hat{\alpha}_0)$. 

By \eqref{WA-eq3}, the terms in the second pair of brackets on the RHS of \eqref{WA-eq7} and \eqref{WA-eq8} are equal. In conclusion, under the new menu, the profit obtained from the type-0 agent is strictly improved, leading to a contradiction. Hence, it must hold that $\alpha_0^{*}\leq \alpha_0^{FB}$. \\
\\
\textit{Statement \circled{2}:} Next we show that $\alpha_1^{*}\geq \alpha_1^{MH}$. Suppose this is not the case and we have $\alpha_1^{*}<\alpha_1^{MH}$. Consider the following two cases based on the $m_1^{*}$ value. \\
\\
\textbf{Case 1:} $m_1^{*} >0$.

In this case, we can increase $\alpha_1^{*}$ to $\hat{\alpha}_1$ (and adjust the output distribution accordingly such that $\Theta_{R_1^{*}}(\hat{\alpha}_1)$ is attained) such that $\alpha_1^{*} <\hat{\alpha}_1 < \alpha_1^{MH}$. Decrease $m_1^{*}$ to $\hat{m}_1$ such that 
\begin{equation}
K_1'(\alpha_1^{*})\cdot \alpha_1^{*} - K_1(\alpha_1^{*}) + m_1^{*} = K_1'(\hat{\alpha}_1)\cdot \hat{\alpha}_1 - K_1(\hat{\alpha}_1) + \hat{m}_1.
\label{WA-eq9}
\end{equation}

We require that $\hat{\alpha}_1$ is sufficiently close to $\alpha_1^{*}$ such that $\hat{m}_1 \geq 0$. 

The new menu satisfies \eqref{WA-eq2b} since the LHS of \eqref{WA-eq2b} is unchanged. We divide Case 1 into two subcases and would like to show that \eqref{WA-eq2a} is also satisfied in both subcases.\\
\\
\textbf{Case 1a: }$K_1'(\alpha_1^{*}) \geq K_0'(\underline{r}_1^{*})$. 

In this case, we have that $K_1'(\hat{\alpha}_1) > K_1'(\alpha_1^{*}) \geq K_0'(\underline{r}_1^{*})$. Furthermore, it holds that $K_1'(\alpha_1^{*}) < K_1'(\hat{\alpha}_1) < K_0'(\hat{\alpha}_1) \leq K_0'(\overline{c})$. 

Let $\hat{s}_1$ denote the contract pinned down by $\hat{m}_1$, $\hat{m}_1$ and $R_1^{*}$. Under the new menu, the RHS of \eqref{WA-eq2a} is reduced by $\max_{\alpha \in [\underline{c},\overline{c}]}\big[\widehat{s_1^{*}\circ \theta}(\alpha) - K_0(\alpha)\big]-\max_{\alpha \in [\underline{c},\overline{c}]}\big[\widehat{\hat{s}_1\circ \theta}(\alpha) - K_0(\alpha)\big] = \max_{\alpha\in[\underline{r}_1^{*},\overline{c}]}\big[K_1'(\alpha_1^{*})\cdot \alpha - K_0(\alpha)+m_1^{*}\big]-\max_{\alpha\in[\underline{r}_1^{*},\overline{c}]}\big[K_1'(\hat{\alpha}_1)\cdot \alpha - K_0(\alpha)+\hat{m}_1\big]= m_1^{*} - \hat{m}_1 - \Delta'$, where
\begin{equation}
    \Delta' := \max_{\alpha\in[\underline{r}_1^{*},\overline{c}]}\big[K_1'(\hat{\alpha}_1)\cdot \alpha - K_0(\alpha)\big]-\max_{\alpha\in[\underline{r}_1^{*},\overline{c}]}\big[K_1'(\alpha_1^{*})\cdot \alpha - K_0(\alpha)\big]. 
\label{WA-eq10}
\end{equation}

By $K_0'(\underline{r}_1^{*})\leq K_1'(\alpha_1^{*}) < K_1'(\hat{\alpha}_1)<K_0'(\overline{c})$, there exist $\beta_1^{*},\hat{\beta}_1\in [\underline{r}_1^{*},\overline{c}]$ such that $K_0'(\beta_1^{*}) = K_1'(\alpha_1^{*})$ and $K_0'(\hat{\beta}_1) = K_1'(\hat{\alpha}_1)$ ($\beta_1^{*} < \hat{\beta}_1$). We have that
\begin{equation*}
\begin{aligned}
    \Delta' &= \big[K_0'(\hat{\beta}_1)\cdot \hat{\beta}_1 - K_0(\hat{\beta}_1)\big] -\big[K_0'(\beta_1^{*})\cdot \beta_1^{*} - K_0(\beta_1^{*})\big]\\
    &=\int_{\beta_1^{*}}^{\hat{\beta}_1}tdk_0(t) = \int_{k_1(\alpha_1^{*})}^{k_1(\hat{\alpha}_1)}k_0^{-1}(k)dk,
\end{aligned}
\end{equation*}
and 
\begin{equation*}
\begin{aligned}
    m_1^{*} - \hat{m}_1 &= \big[K_1'(\hat{\alpha}_1)\cdot \hat{\alpha}_1 - K_1(\hat{\alpha}_1)\big] - \big[K_1'(\alpha_1^{*})\cdot \alpha_1^{*} - K_1(\alpha_1^{*})\big] \\
    &=\int_{\alpha_1^{*}}^{\hat{\alpha}_1}tdk_1(t) = \int_{k_1(\alpha_1^{*})}^{k_1(\hat{\alpha}_1)}k_1^{-1}(k)dk \\
    &> \int_{k_1(\alpha_1^{*})}^{k_1(\hat{\alpha}_1)}k_0^{-1}(k)dk = \Delta'. 
\end{aligned}
\end{equation*}

Hence, we have that $m_1^{*} - \hat{m}_1 - \Delta' > 0$. It implies that the RHS of \eqref{WA-eq2b} is reduced under the new menu. Consequently, \eqref{WA-eq2b} is still satisfied. \\
\\
\textbf{Case 1b:} $K_1'(\alpha_1^{*}) <K_0'(\underline{r}_1^{*})$. 

If $\underline{r}_1^{*} = \underline{c}$, we can establish that \eqref{WA-eq2b} is satisfied at the new menu by using arguments similar to those in Case 1a. Hence, we restrict ourselves to the situation where $\underline{r}_1^{*} > \underline{c}$. 

In this case, we further require $\hat{\alpha}_1$ to be sufficiently close to $\alpha_1^{*}$, when choosing $\hat{\alpha}_1$, such that $K_1'(\alpha_1^{*})<K_1'(\hat{\alpha}_1)<K_0'(\underline{r}_1^{*})$. Under the new menu, the RHS of \eqref{WA-eq2b} is reduced by $\max_{\alpha \in [\underline{c},\overline{c}]}\big[\widehat{s_1^{*}\circ \theta}(\alpha) - K_0(\alpha)\big]-\max_{\alpha \in [\underline{c},\overline{c}]}\big[\widehat{\hat{s}_1\circ \theta}(\alpha) - K_0(\alpha)\big] = \max_{\alpha\in[\underline{c},\underline{r}_1^{*}]}\big[\big(K_1'(\alpha_1^{*})+\frac{m_1^{*}}{\underline{r}_1^{*}}\big)\cdot \alpha - K_0(\alpha)\big]-\max_{\alpha\in[\underline{c},\underline{r}_1^{*}]}\big[\big(K_1'(\hat{\alpha}_1)+\frac{\hat{m}_1}{\underline{r}_1^{*}}\big)\cdot \alpha - K_0(\alpha)\big]$. 

We can show that this term is non-negative. From $m_1^{*} - \hat{m}_1 = \int_{\alpha_1^{*}}^{\hat{\alpha}_1}tdk_1(t)\geq \int_{\alpha_1^{*}}^{\hat{\alpha}_1}\underline{r}_1^{*}dk_1(t)= \underline{r}_1^{*}\cdot \big[K_1'(\hat{\alpha}_1) - K_1'(\alpha_1^{*})\big]$, we have that $K_1'(\alpha_1^{*})+\frac{m_1^{*}}{\underline{r}_1^{*}}\geq K_1'(\hat{\alpha}_1)+\frac{\hat{m}_1}{\underline{r}_1^{*}}$, leading to that $\max_{\alpha \in [\underline{c},\overline{c}]}\big[\widehat{s_1^{*}\circ \theta}(\alpha) - K_0(\alpha)\big] \geq \max_{\alpha \in [\underline{c},\overline{c}]}\big[\widehat{\hat{s}_1\circ \theta}(\alpha) - K_0(\alpha)\big]$. Hence, the RHS of \eqref{WA-eq2b} is still satisfied under the new menu.\\
\\
In conclusion, the new menu is also feasible. Under the new menu, the profit from the type-0 agent is unchanged, while the profit from the type-1 agent changes from
\begin{equation}
\Theta_{R_1^{*}}(\alpha_1^{*}) - K_1'(\alpha_1^{*})\cdot \alpha_1^{*} - m_1^{*} = \big[\Theta_{R_1^{*}}(\alpha_1^{*}) - \Theta(\alpha_1^{*})\big] + \big[\Theta(\alpha_1^{*})-K_1'(\alpha_1^{*})\cdot \alpha_1^{*} - m_1^{*}\big],
\label{WA-eq11}
\end{equation}
and
\begin{equation}
\Theta_{R_1^{*}}(\hat{\alpha}_1) - K_1'(\hat{\alpha}_1)\cdot \hat{\alpha}_1 - \hat{m}_1 = \big[\Theta_{R_1^{*}}(\hat{\alpha}_1) - \Theta(\hat{\alpha}_1)\big] + \big[\Theta(\hat{\alpha}_1)-K_1'(\hat{\alpha}_1)\cdot \hat{\alpha}_1 - \hat{m}_1\big].
\label{WA-eq12}
\end{equation}

Consider the terms in the first pair of brackets on the RHS of \eqref{WA-eq11} and \eqref{WA-eq12}. From Lemma \ref{WA-lm1}, we have that $\Theta_{R_1^{*}}'(\alpha)\geq \Theta'(\alpha)$ for all $\alpha \in c(R_1^{*})$. It implies that $\Theta_{R_1^{*}}(\hat{\alpha}_1) - \Theta(\hat{\alpha}_1) \geq \Theta_{R_1^{*}}(\alpha_1^{*})-\Theta(\alpha_1^{*}).$

Next, we consider the terms in the second pair of brackets on the RHS of \eqref{WA-eq11} and \eqref{WA-eq12}. From $\alpha_1^{*}<\hat{\alpha}_1 < \alpha_1^{MH}$, we have that $\Theta(\hat{\alpha}_1) - K_1'(\hat{\alpha}_1)\cdot \hat{\alpha}_1 >\Theta(\alpha_1^{*}) - K_1'(\alpha_1^{*})\cdot \alpha_1^{*}$. From \eqref{WA-eq9}, we have $\hat{m}_1 < m_1^{*}$. In conclusion, the profit from the type-1 agent is strictly increased under the new menu. It is contradictory to the optimality of $(a^{*},m^{*},R^{*},\mu^{*})$. Hence, we must have $\alpha_1^{*} \geq \alpha_1^{MH}$ in Case 1. \\
\\
\textbf{Case 2:} $m_1^{*} = 0$.

According to Corollary \ref{APP-IMP-pp2} and \ref{APP-IMP-pp3}, $m_1^{*} = 0$ implies that $m_0^{*} = 0$. $(a^{*},m^{*},R^{*},\mu^{*})$ satisfies the IC constraints 
\begin{subequations}
\begin{align}
    K_0'(\alpha_0^{*})\cdot \alpha_0^{*} - K_0(\alpha_0^{*}) &\geq \max_{\alpha\in [\underline{c},\overline{c}]}\big[K_1'(\alpha_1^{*})\cdot \alpha - K_0(\alpha)\big], \label{WA-eq13a} \\
    K_1'(\alpha_1^{*})\cdot \alpha_1^{*} - K_1(\alpha_1^{*}) &\geq \max_{\alpha\in [\underline{c}, \overline{r}_0^{*}]}\big[K_0'(\alpha_0^{*})\cdot \alpha - K_1(\alpha)\big]. \label{WA-eq13b}
\end{align}
\end{subequations}

Increase $\alpha_1^{*}$ to $\hat{\alpha}_1$ (and adjust the output distribution accordingly such that $\Theta_{R_1^{*}}(\hat{\alpha}_1)$ is attained) such that $\alpha_1^{*} < \hat{\alpha}_1 <\alpha_1^{MH}$. We require $\hat{\alpha}_1$ to be sufficiently close to $\alpha_1^{*}$ such that \eqref{WA-eq13a} is still satisfied under the new menu. This is obviously possible when $K_0'(\alpha_0^{*})\cdot \alpha_0^{*} - K_0(\alpha_0^{*}) > \max_{\alpha\in [\underline{c},\overline{c}]}\big[K_1'(\alpha_1^{*})\cdot \alpha - K_0(\alpha)\big]$. Suppose we have $K_0'(\alpha_0^{*})\cdot \alpha_0^{*} - K_0(\alpha_0^{*}) = \max_{\alpha\in [\underline{c},\overline{c}]}\big[K_1'(\alpha_1^{*})\cdot \alpha - K_0(\alpha)\big]$, this is only possible when $\alpha_0^{*} = \underline{c}$. In this case, we have $K_1'(\alpha_1^{*}) < K_0'(\alpha_0^{*}) = K_0'(\underline{c})$. We can let $\hat{\alpha}_1$ be sufficiently close to $\alpha_1^{*}$ such that $K_1'(\alpha_1^{*}) < K_1'(\hat{\alpha}_1) < K_0'(\underline{c})$, which guarantees that \eqref{WA-eq13a} is still satisfied under the new menu. \eqref{WA-eq13b} is also maintained since the LHS of \eqref{WA-eq13b} is increased. 

Under the new menu, the profit from the type-0 agent is unchanged, while the profit from the type-1 agent changes from 
\begin{equation}
\Theta_{R_1^{*}}(\alpha_1^{*}) - K_1'(\alpha_1^{*})\cdot \alpha_1^{*} = \big[\Theta_{R_1^{*}}(\alpha_1^{*}) - \Theta(\alpha_1^{*})\big] + \big[\Theta(\alpha_1^{*})-K_1'(\alpha_1^{*})\cdot \alpha_1^{*}\big],
\label{WA-eq14}
\end{equation}
and
\begin{equation}
\Theta_{R_1^{*}}(\hat{\alpha}_1) - K_1'(\hat{\alpha}_1)\cdot \hat{\alpha}_1 = \big[\Theta_{R_1^{*}}(\hat{\alpha}_1) - \Theta(\hat{\alpha}_1)\big] + \big[\Theta(\hat{\alpha}_1)-K_1'(\hat{\alpha}_1)\cdot \hat{\alpha}_1 \big]
\label{WA-eq15}
\end{equation}

From Lemma \ref{WA-lm1}, we have that $\Theta_{R_1^{*}}(\hat{\alpha}_1) - \Theta(\hat{\alpha}_1) \geq \Theta_{R_1^{*}}(\alpha_1^{*})-\Theta(\alpha_1^{*})$. From $\alpha_1^{*}<\hat{\alpha}_1 < \alpha_1^{MH}$, we have that $\Theta(\hat{\alpha}_1)-K_1'(\hat{\alpha}_1)\cdot \hat{\alpha}_1 > \Theta(\alpha_1^{*})-K_1'(\alpha_1^{*})\cdot \alpha_1^{*}$. Hence, the profit from the type-1 agent is strictly increased under the new menu. It is contradictory to the optimality of $(a^{*},m^{*},R^{*},\mu^{*})$. $\hfill\square$
\section{Proofs in Section \ref{MBCF-general}}
\label{MBCF-app-general}
\textbf{Proof of Proposition \ref{NT-pp1}: }\textit{Sufficiency: }Omitted.\\
\\
\textit{Necessity: }Suppose $(a^{*},m^{*},R_F,\mu^{*})$ is feasible, it also optimally solves \eqref{NT-P-N}. Take $\kappa_i^{*} := K_i'(\alpha_i^{*})$. We aim to show that $\kappa_i^{*} = \kappa_j^{*}$ for all $i,j\in \{1,...,N\}$.

Note that in the menu $(a^{*},m^{*},R_F,\mu^{*})$, the contract $s_i^{*}(\cdot)$ assigned to the type-$i$ agent takes the form of $s_i^{*}\circ\theta(\alpha) = \kappa_i^{*} \cdot \alpha + m_i^{*}$ for all $\alpha\in [\underline{c},\overline{c}]$. Hence, we can denote this contract by the pair $(\kappa_i^{*},m_i^{*})$. Let $M := \big\{(\kappa_i^{*},m_i^{*})\big\}_{i\in\{1,...,N\}}$ be the alternative notion of the menu $(a^{*},m^{*},R_F,\mu^{*})$. 

We aim to show that it must hold that $\kappa_1^{*}\leq ...\leq \kappa_n^{*}$. Suppose this is not true, and there exist $i,j\in \{1,...,N\}$ such that $i<j$ and $\kappa_i^{*} > \kappa_j^{*}$. Following the similar argument to that in Lemma \ref{APP-IMP-lm1}, this is impossible. 

Fix $\kappa_1^{*},...,\kappa_N^{*}$ (the corresponding $\alpha_1^{*},...,\alpha_N^{*}$ are also fixed). Consider the constant payments $m_1,...,m_N$ that can implement $\kappa_1^{*},...,\kappa_N^{*}$. For any type $i\in \{1,...,N-1\}$, the upward IC constraints for type-$i$ requires that $m_i - m_{i+1} \geq \max_{\alpha\in [\underline{c},\overline{c}]}\big[\kappa_{i+1}^{*}\cdot \alpha - K_i(\alpha)\big] - \big[\kappa_i^{*}\cdot \alpha_i^{*} - K_i(\alpha_i^{*})\big] = \int_{\kappa_i^{*}}^{\kappa_{i+1}^{*}}k_i^{-1}(k)dk$. Hence, at the menu $M$, it holds that $m_i^{*} -m_{i+1}^{*} \geq \int_{\kappa_i^{*}}^{\kappa_{i+1}^{*}}k_i^{-1}(k)dk$.

Next, we aim to show that $m_1^{*},...,m_N^{*}$ are pinned down by the binding upward IC constraints and the binding LL constraint for type $N$, i.e., 
\begin{equation}
m_i^{*}-m_{i+1}^{*} = \int_{\kappa_i^{*}}^{\kappa_{i+1}^{*}}k_i^{-1}(k)dk \text{ for all }i\in \{1,...,N\}\text{, and } m_N^{*} = 0.
\label{NT-eq1}
\end{equation}

It suffices to show that if we take payments $m_1,...,m_N$ such that $m_i-m_{i+1} = \int_{\kappa_i}^{\kappa_{i+1}^{*}}k_i^{-1}(k)dk$ for all $i\in \{1,...,N-1\}$ and $m_N = 0$, they support $\kappa_1^{*},...,\kappa_N^{*}$ to become a feasible menu. Take any $p,q\in \{1,...,N\}$ such that $p<q$. We have that
\begin{equation*}
m_p - m_q = \int_{\kappa_p}^{\kappa_{p+1}}k_p^{-1}(k)dk+...+\int_{\kappa_{q-1}}^{\kappa_q}k_{q-1}^{-1}(k)dk,
\end{equation*}
which implies that $m_p-m_q\geq \int_{\kappa_p}^{\kappa_q}k_p^{-1}(k)dk$ and $m_p-m_q\leq \int_{\kappa_p}^{\kappa_q}k_q^{-1}(k)dk$. Hence, the IC constraints between arbitrary types $p$ and $q$ ($p<q$) are satisfied. All the LL constraints are also satisfied. In conclusion, we've established that \eqref{NT-eq1} holds at the menu $M$.

Finally, we want to show that $\kappa_1^{*} = \kappa_2^{*}=...=\kappa_N^{*}$. Suppose that it is not the case, and there exists $i\in \{1,...,N-1\}$ such that 
\begin{equation*}
\kappa_1^{*}\leq...\leq \kappa_i^{*} < \kappa_{i+1}^{*}\leq ...\leq \kappa_N^{*}. 
\end{equation*}

Consider any $\hat{\kappa}_i\in [\kappa_i^{*},\kappa_{i+1}^{*}]$. We would like to replace $\kappa_i^{*}$ with $\hat{\kappa}_i$, and we need to find the constant payments that implement the contract powers after the replacement. To do so, we define $\hat{m}_i:= m_i^{*} - \int_{\kappa_i^{*}}^{\hat{\kappa}_i}k_i^{-1}(k)dk$. 

We first verify that $m_1^{*},...,m_{i-1}^{*},\hat{m}_i,m_{i+1}^{*},...,m_N^{*}$ implement $\kappa_1^{*},...,\kappa_{i-1}^{*},\hat{\kappa}_i,\kappa_{i+1}^{*},...,\kappa_N^{*}$. It suffices to verify whether the IC constraints between type $i$ and any other type are satisfied. Take any type $q>i$, 
\begin{equation*}
\hat{m}_i - m_q^{*} = \int_{\hat{\kappa}_i}^{\kappa_{i+1}^{*}}k_i^{-1}(k)dk + ...+\int_{\kappa_{q-1}^{*}}^{\kappa_q^{*}}k_{q-1}^{-1}(k)dk.
\end{equation*}

Hence, we have that $\hat{m}_i - m_1^{*} \geq \int_{\hat{\kappa}_i}^{\kappa_q^{*}}k_i^{-1}(k)dk$ and $\hat{m}_i - m_1^{*} \leq \int_{\hat{\kappa}_i}^{\kappa_q^{*}}k_q^{-1}(k)dk$. The IC constraints between types $i$ and $q$ are satisfied. Then we take any type $p<i$, it holds that
\begin{equation*}
    m_p^{*} - \hat{m}_i = \int_{\kappa_p^{*}}^{\kappa_{p+1}^{*}}k_p^{-1}(k)dk+...+\int_{\kappa_{i-1}^{*}}^{\kappa_i^{*}}k_{i-1}^{-1}(k)dk + \int_{\kappa_i^{*}}^{\hat{\kappa}_i}k_i^{-1}(k)dk.
\end{equation*}

Hence, we have that $m_p^{*} - \hat{m}_i \geq \int_{\kappa_p^{*}}^{\hat{\kappa}_i}k_p^{-1}(k)dk$ and $m_p^{*} - \hat{m}_i \leq \int_{\kappa_p^{*}}^{\hat{\kappa}_i}k_i^{-1}(k)dk$. The IC constraints between types $p$ and $i$ are satisfied. 

Varying $\hat{\kappa}_i$ between $\kappa_i^{*}$ and $\kappa_{i+1}^{*}$, among the menu $\big\{(\kappa_1^{*},m_1^{*}),...,(\hat{\kappa}_1,\hat{m}_1),...,(\kappa_N^{*},m_N^{*})\big\}$, it is optimal to choose the one with $\hat{\kappa_i} = \kappa_i^{*}$. Hence, $(\kappa_i^{*},m_i^{*})$ maximizes
\begin{equation*}
\begin{aligned}
    \max_{\hat{\kappa}_i\in [\kappa_i^{*},\kappa_{i+1}^{*}],\hat{m}_i} & \Theta(k_i^{-1}(\hat{\kappa}_i)) - \hat{\kappa}_i \cdot k_i^{-1}(\hat{\kappa}_i) - \hat{m}_i\\
    s.t.\quad& \hat{m}_i = m_i^{*} - \int_{\kappa_i^{*}}^{\hat{\kappa}_i}k_i^{-1}(k)dk.
\end{aligned}   
\end{equation*}

Equivalently, $\alpha_i^{*}$ maximizes
\begin{equation*}
\begin{aligned}
    \max_{\hat{\alpha}_i \in [\alpha_i^{*},k_i^{-1}(\kappa_{i+1}^{*})],\hat{m}_i} & \Theta(\hat{\alpha}_i) - k_i(\hat{\alpha}_i)\cdot \hat{\alpha}_i - \hat{m}_i\\
    s.t.\quad& \hat{m}_i = m_i^{*} + \big[\kappa_i^{*} \cdot \alpha_i^{*} - K_i(\alpha_i^{*})\big] - \big[k_i(\hat{\alpha}_i)\cdot\hat{\alpha}_i- K_i(\hat{\alpha}_i)\big].
\end{aligned}
\end{equation*}

It implies that
\begin{equation*}
    \alpha_i^{*}\in \argmax_{\hat{\alpha}_i \in [\alpha_i^{*},k_i^{-1}(\kappa_{i+1}^{*})]}\big[\Theta(\hat{\alpha}_i) - K_i(\hat{\alpha}_i)\big],
\end{equation*}
and thus,
\begin{equation*}
    \Theta'(\alpha_i^{*})\leq K_i'(\alpha_i^{*})= \kappa_i^{*}. 
\end{equation*}

For any type $j>i$, we have
\begin{equation}
    \Theta'(\alpha_j^{*}) \leq \Theta'(\alpha_i^{*})\leq \kappa_i^{*} < \kappa_j^{*} \leq K_j'(\alpha_j^{*}) + K_j''(\alpha_j^{*})\cdot \alpha_j^{*}. 
\label{NT-eq2}
\end{equation}

We have that $\alpha_j^{*} > \alpha_j^{MH}$. Take $\Delta \kappa$ sufficiently small such that $\kappa_{i+1}^{*} - \Delta \kappa > \kappa_i^{*}$, and that $k_j^{-1}(\kappa_j^{*}-\Delta\kappa) > \alpha_j^{MH}$ for all $j>i$. Decrease $\kappa_j^{*}$ by $\Delta \kappa$ for all $j>i$. Adjust the constant payments according to \eqref{NT-eq1}. It is straightforward to see that all the constant payments are weakly reduced. The principal's profit from the type-$j$ agent is weakly increased for all $j\in \{1,...,N\}$. And the increase is strict for type $j > i$. Hence, we have found a strict improvement on the menu $M$, contradicting its optimality. $\hfill\square$ \\
\\
\textbf{Proof of Proposition \ref{OSG-pp1}: }Take an optimal menu $M = \{s_t^{*}(\cdot),\mu_t^{*}\}_{t\in T}$ to \eqref{OSG-GP}. Suppose that the statement \circled{2} doesn't hold. Let $\Pi_0$ denote the set of type $t\in T$ such that $s_t^{*}(\cdot)$ does not have full range with regard to $t$. It holds that $\phi(\Pi_0) = 0$ and $s_t^{*}(\cdot)$ has full range with regard to $t$ for all $t\in T\backslash \Pi_0$. We aim to show that $s_t^{*}(\cdot)$ is the same across all types $t\in T\backslash \Pi_0$ except those in a zero-measure subset of $T\backslash \Pi_0$. 

Suppose that this is not the case. Then for any possible contract $s(\cdot)$, there exists a positive-measure set of types such that for any $t$ in that set, it holds that $s_t^{*}(\cdot)\neq s(\cdot)$. We aim to show that, in this case, we can strictly improve the principal's profit by replacing the contracts to agents of types $t\in T\backslash \Pi_0$ with a single contract.

For all $t\in T\backslash \Pi_0$, since $s_t^{*}(\cdot)$ has full range with regard to type $t$, then there exists constant $m_t^{*}\in \mathbb{R}$ such that
\begin{equation*}
\begin{aligned}
    s_t^{*}(x) &= K_t'\big(\int c(y)\mu_t^{*}(dy)\big)\cdot c(x) + m_t^{*}  \\
    &= K_t'\big(\int c(y)\mu_t^{*}(dy)\big)\cdot c\cdot x +  K_t'\big(\int c(y)\mu_t^{*}(dy)\big) \cdot d + m_t^{*}.
\end{aligned}
\end{equation*}
Define $\kappa_t^{*} := K_t'\big(\int c(y)\mu_t^{*}(dy)\big)\cdot c$ and $b_t^{*} := K_t'\big(\int c(y)\mu_t^{*}(dy)\big) \cdot d + m_t^{*}$. Then $s_t^{*}(\cdot)$ can be written as $s_t^{*}(x) = \kappa_t^{*}\cdot x + b_t^{*}$ for all $t\in T\backslash \Pi_0$, where $\kappa_t^{*}\in [K_t'(\underline{c})\cdot c,K_t'(\overline{c})\cdot c]$. 

Define $\overline{T}: = \{t\in T\backslash \Pi_0: \kappa_t^{*}>1\}$. Consider the following two cases based on $\phi(\overline{T})$. \\
\\
\textbf{Case 1: }Suppose that $\phi(\overline{T}) > 0$.

Take some $t_0 \in \overline{T}$. We have that $s_{t_0}^{*}(x) = \kappa_{t_0}^{*}\cdot x + b_{t_0}^{*}$ where $\kappa_{t_0}^{*} >1$. Take $\delta = \frac{1}{\kappa_{t_0}^{*}}\in (0,1)$. Define $\tilde{\kappa}:= \delta \cdot \kappa_{t_0}^{*} = 1$, $\tilde{b}:= \delta \cdot b_{t_0}^{*}$ and $\tilde{s}(x):= \tilde{\kappa}\cdot x + \tilde{b} = \delta \cdot s_{t_0}^{*}(x)$. Then $\tilde{s}(\cdot)$ satisfies limited liability. We consider offering $\tilde{s}(\cdot)$ to each type-$t$ ($t\in T$) agent instead of $s_t^{*}(\cdot)$. Let $\tilde{\mu}_t$ denote the distribution recommended to type-$t$ agent under contract $\tilde{s}(\cdot)$. We aim to show that the principal's profit is strictly increased after replacing the menu $M$ with $\{\tilde{s}(\cdot),\tilde{\mu}_t\}_{t\in T}$. 

For the type-$t$ agent with $t\in T\backslash \Pi_0$, fix $\mu_t$ as his action choice, then the utility obtained by the agent is lower under $\tilde{s}(\cdot)$ than under $s_t^{*}(\cdot)$: 
\begin{equation}
\int s_t^{*}(x)\mu_t^{*}(dx) - C_t(\mu_t^{*})\geq \int s_{t_0}^{*}(x)\mu_t^{*}(dx) - C_t(\mu_t^{*})\geq \int \tilde{s}(x)\mu_t^{*}(dx) - C_t(\mu_t^{*}),
\label{OSG-eq8}
\end{equation}
where the first inequality follows from incentive compatibility and the second one follows from that $\tilde{s}(x) = \delta\cdot s_{t_0}^{*}(x)$ with $\delta \in (0,1)$ and $s_{t_0}^{*}(x)\geq 0$ for all $x\in [\underline{x},\overline{x}]$. It implies that with $\mu_t^{*}$ fixed, the principal's profit from the type-$t$ agent under $\tilde{s}(\cdot)$ is no smaller than that under $s_t^{*}(\cdot)$:
\begin{equation}
    \int \big[x - \tilde{s}(x)\big]\mu_t^{*}(dx) \geq \int \big[x-s_t^{*}(x)\big]\mu_t^{*}(dx). 
\label{OSG-eq7}
\end{equation}

Next, we show that \eqref{OSG-eq7} is strict for type $t \in \overline{T}$. For such a type $t$, we have $\kappa_t^{*}>1$. According to Lemma \ref{PRE-lm2}, $\int c(y)\mu_t^{*}(dy) \in \argmax_{\alpha\in [\underline{c},\overline{c}]}\big[\frac{\kappa_t^{*}}{c}(\alpha -d)+b_t^{*} - K_t'(\alpha)\big]$. Under Assumption \ref{OSG-as1}, $\frac{\kappa_t^{*}}{c} > \frac{1}{c} > K_t'(\underline{c})$, then we have that $\int c(x)\mu_t^{*}(dx) > \underline{c}$, which implies that $\text{supp}(\mu_t^{*}) \neq \{\underline{x}\}$. Due to that $s_{t_0}^{*}(\underline{x})\geq 0$ and that $s_{t_0}^{*}(\cdot)$ is strictly increasing over $[\underline{x},\overline{x}]$, it holds that $\int s_{t_0}^{*}(x)\mu_t^{*}(dx) > 0$, which implies that the second inequality in \eqref{OSG-eq8} is strict in this case, leading to the strictness of \eqref{OSG-eq7}. 

Furthermore, due to that $\tilde{\kappa} = 1$, we have that the type-$t$ agent's shift from $\mu_t^{*}$ to $\tilde{\mu}_t$ facing with $\tilde{s}(\cdot)$ has no impact on the principal's profit: 
\begin{equation}
    \int \big[x - \tilde{s}(x)\big]\tilde{\mu}_t(dx) = \int \big[x - \tilde{s}(x)\big]\mu_t^{*}(dx). 
\label{OSG-eq9}
\end{equation}

Combining \eqref{OSG-eq7} and \eqref{OSG-eq9}, we have that 
\begin{equation*}
\begin{aligned}
    \int \big[x - \tilde{s}(x)\big]\tilde{\mu}_t(dx) &\geq \int \big[x - s_t^{*}(x)\big]\mu_t^{*}(dx),\quad\forall t\in T\backslash \Pi_0, \\
    \text{and }\int \big[x - \tilde{s}(x)\big]\tilde{\mu}_t(dx) &> \int \big[x - s_t^{*}(x)\big]\mu_t^{*}(dx),\quad\forall t\in \overline{T},
\end{aligned}
\end{equation*}
which indicates that the principal's profit from agents with types in $T\backslash \Pi_0$ is increased after the replacement while the increase is strict for types in $\overline{T}$. In conclusion, the replacement strictly increases the principal's expected profit. \\
\\
\textbf{Case 2: }Suppose that $\phi(\overline{T}) = 0$.

It indicates that for all $t\in T\backslash(\Pi_0\cup \overline{T})$, we have that $\kappa_t^{*} \leq 1$. Take $\tilde{\kappa}:= \sup_{t\in T\backslash(\Pi_0\cup \overline{T})}\kappa_t^{*}$. We consider Case 2a and Case 2b based on whether this supremum is attainable, and construct a contract $\tilde{s}(\cdot)$, such that \eqref{OSG-eq7} is established for all $t\in T\backslash(\Pi_0\cup \overline{T})$. \\
\\
\textbf{Case 2a: }If $\argmax_{t\in T\backslash(\Pi_0\cup \overline{T})}\kappa_t^{*}$ is non-empty, take any $t_0 \in \argmax_{t\in T\backslash(\Pi_0\cup \overline{T})}\kappa_t^{*}$. Then $\tilde{\kappa} = \sup_{t\in T\backslash(\Pi_0\cup \overline{T})}\kappa_t^{*} = \kappa_{t_0}^{*}$. Define $\tilde{b}: = b_{t_0}^{*}$ and $\tilde{s}(x): = \tilde{\kappa}\cdot x + \tilde{b} = s_{t_0}^{*}(x)$ for all $x\in [\underline{x},\overline{x}]$. Then the contract $\tilde{s}(\cdot)$ belongs to the original menu $M$, and \eqref{OSG-eq7} can be established for all $t\in T\backslash(\Pi_0\cup \overline{T})$ due to incentive compatibility. \\
\\
\textbf{Case 2b: }If $\argmax_{t\in T\backslash(\Pi_0\cup \overline{T})}\kappa_t^{*}$ is empty, define $\tilde{b} := \min\{b: \tilde{\kappa}\cdot x + b\geq 0, \forall x\in [\underline{x},\overline{x}]\}$ and $\tilde{s}(x) = \tilde{\kappa}\cdot x + \tilde{b}$ for all $x\in [\underline{x},\overline{x}]$.

There exists a sequence of contracts $\{s^n(\cdot)\}_{n\in \mathbb{N}}$ such that for each $n\in \mathbb{N}$, $s^n(\cdot)$ belongs to $\{s_t^{*}(\cdot)\}_{t\in T\backslash(\Pi_0\cup\overline{T})}$ and that $\lim_{n\to \infty}\kappa^n = \tilde{\kappa}$, where $\kappa^n$ is the slope of $s^n(\cdot)$. Let $b^n$ denote the intercept of $s^n(\cdot)$. We aim to show that $\tilde{b}\leq b^n$ for all $n\in \mathbb{N}$. Suppose this is not the case and there exists $n\in \mathbb{N}$ such that $\tilde{b} > b^n$. Since $\tilde{\kappa} > \kappa^n$, we have that $\tilde{\kappa}\cdot \underline{x} + \tilde{b} > \kappa^n\cdot \underline{x} + b^n \geq  0$, which indicates that $\tilde{b}$ can be further reduced without violating limited liability of contract $\tilde{s}(\cdot)$, which is contradictory to the definition of $\tilde{b}$. 

For any $t\in T\backslash (\Pi_0\cup \overline{T})$, for any $n\in \mathbb{N}$,
\begin{equation}
\int s_t^{*}(x)\mu_t^{*}(dx) \geq \int s^n(x)\mu_t^{*}(dx) = \int(\kappa^n\cdot x + b^n)\mu_t^{*}(dx) \geq \int (\kappa^n \cdot x + \tilde{b})\mu_t^{*}(dx). 
\label{OSG-eq10}
\end{equation}
Take $n\to \infty$ on the right-hand side of \eqref{OSG-eq10}. We have that $\lim_{n\to \infty} \int (\kappa^n \cdot x + \tilde{b})\mu_t^{*}(dx) = \int (\tilde{\kappa}\cdot x + \tilde{b})\mu_t^{*}(dx) = \int \tilde{s}(x)\mu_t^{*}(dx)$. Hence, \eqref{OSG-eq10} indicates that $\int s_t^{*}(x)\mu_t^{*}(dx)\geq \int \tilde{s}(x)\mu_t^{*}(dx)$, which leads to \eqref{OSG-eq7}. \\

In conclusion, we have constructed the single contract $\tilde{s}(\cdot)$ and established \eqref{OSG-eq7} for all $t\in T\backslash (\Pi_0\cup \overline{T})$ in both Case 2a and Case 2b.

Now we consider how the principal's profit is changed by replacing the menu $M$ with $\{\tilde{s}(\cdot),\tilde{\mu}_t\}_{t\in T}$, where $\tilde{\mu}_t$ is the distribution recommended to the type-$t$ agent under contract $\tilde{s}(\cdot)$. We require that $\tilde{s}(\cdot)$ implements $\tilde{\mu}_t$ for the type-$t$ agent. 

Since contract $s_t^{*}(\cdot)$ ($t\in T$) implements $\mu_t^{*}$, while $\tilde{s}(\cdot)$ implements $\tilde{\mu}_t$. We have that for all $t\in T\backslash(\Pi_0\cup \overline{T})$,
\begin{equation*}
\begin{aligned}
    \int \tilde{s}(x)\tilde{\mu}_t(dx) - C_t(\tilde{\mu}_t) &\geq \int \tilde{s}(x)\mu_t^{*}(dx) - C_t(\mu_t^{*}), \\
    \int s_t^{*}(x)\mu_t^{*}(dx) - C_t(\mu_t^{*}) &\geq \int s_t^{*}(x)\tilde{\mu}_t(dx) - C_t(\tilde{\mu}_t), 
\end{aligned}
\end{equation*}
\begin{equation}
    \Longrightarrow \int \tilde{s}(x)(\tilde{\mu}_t - \mu_t^{*})(dx) \geq \int s_t^{*}(x)(\tilde{\mu}_t-\mu_t^{*})(dx). 
\label{OSG-eq11}
\end{equation}

We can write \eqref{OSG-eq11} as $\tilde{\kappa}\big(\mathbb{E}[\tilde{\mu}_t] - \mathbb{E}[\mu_t^{*}]\big)\geq \kappa_t^{*}\big(\mathbb{E}[\tilde{\mu}_t] -\mathbb{E}[\mu_t^{*}]\big)$. If $\tilde{\kappa} > \kappa_t^{*}$, it implies that $\mathbb{E}[\tilde{\mu}_t]\geq \mathbb{E}[\mu_t^{*}]$. If $\tilde{\kappa} = \kappa_t^{*}$, then $\mu_t^{*}$ can also be implemented by $\tilde{s}(\cdot)$. Let $\tilde{\mu}_t = \mu_t^{*}$, then we have that $\mathbb{E}[\tilde{\mu}_t] = \mathbb{E}[\mu_t^{*}]$. In conclusion, \eqref{OSG-eq11} implies that $\mathbb{E}[\tilde{\mu}_t]\geq \mathbb{E}[\mu_t^{*}]$ for all $t\in T\backslash (\Pi_0\cup \overline{T})$. 

Consequently, the principal's profit from the type-$t$ agent under contract $\tilde{s}(\cdot)$ is increased when the agent shifts from $\mu_t^{*}$ to $\tilde{\mu}_t$: 
\begin{equation*}
\int \big[ x - \tilde{s}(x)\big]\tilde{\mu}_t(dx) - \int\big[x-\tilde{s}(x)\big]\mu_t^{*}(dx) = (1-\tilde{\kappa}) \cdot \big(\mathbb{E}[\tilde{\mu}_t] - \mathbb{E}[\mu_t^{*}]\big) \geq 0.
\end{equation*}
\begin{equation}
\Longrightarrow \int \big[ x-\tilde{s}(x)\big]\tilde{\mu}_t(dx) \geq \int\big[x-\tilde{s}(x)\big]\mu_t^{*}(dx). 
\label{OSG-eq12}
\end{equation}

Combining \eqref{OSG-eq7} and \eqref{OSG-eq12}, we have that $\int \big[ x-\tilde{s}(x)\big]\tilde{\mu}_t(dx)\geq \int\big[x - s_t^{*}(x)\big]\mu_t^{*}(dx)$ for all $t\in T\backslash(\Pi_0\cup \overline{T})$. It indicates that the principal wouldn't get worse off after replacing menu $M$ with $\{\tilde{s}(\cdot),\mu_t^{*}\}_{t\in T}$. Next, we aim to establish that this replacement actually leads to a strict improvement.

According to our assumption, there must be a positive measure set $\Pi_1 \subset T\backslash (\Pi_0\cup \overline{T})$ of types such that $s_t^{*}(\cdot)\neq \tilde{s}(\cdot)$, for all $t\in \Pi_1$. Then for any $t$ from $\Pi_1$, we have that
\begin{equation*}
    \text{either }\tilde{\kappa} > \kappa_t^{*}\text{ or }\tilde{\kappa} = \kappa_t^{*}, b_t^{*} > \tilde{b}. 
\end{equation*}

We revisit Case 2a and Case 2b  and establish the argument in each case\\
\\
\textbf{Case 2a revisited: }In this case, $\tilde{\kappa} = \argmax_{t\in T\backslash (\Pi_0\cup \overline{T})}\kappa_t^{*} \leq 1$.

Take any $t\in \Pi_1$. Suppose it holds that $\tilde{\kappa} > \kappa_t^{*}$. From Lemma \ref{PRE-lm2}, we have that $\int c(y)\tilde{\mu}_t(dy)\in \argmax_{\alpha\in[\underline{c},\overline{c}]} \big[\frac{\tilde{\kappa}}{c}(\alpha-d)+\tilde{b} - K_t(\alpha)\big]$ and $\int c(y)\mu_t^{*}(dy) \in \argmax_{\alpha\in [\underline{c},\overline{c}]} \big[\frac{\kappa_t^{*}}{c}(\alpha-d)+ b_t^{*} - K_t(\alpha)\big]$. According to Assumption \ref{OSG-as1}, we have that $K_t'(\overline{c})>\frac{1}{c} \geq \frac{\tilde{\kappa}}{c} > \frac{\kappa_t^{*}}{c} \geq K_t'(\underline{c})$. It indicates that
\begin{equation}
    \overline{c} > \int c(y)\tilde{\mu}_t(dy) > \int c(y)\mu_t^{*}(dy),
    \label{OSG-eq17}
\end{equation}
which implies that $\mathbb{E}[\mu_t^{*}] > \mathbb{E}[\mu_t]$. If $\tilde{\kappa}<1$, then the inequality in \eqref{OSG-eq12} is strict in this case for type $t$. While if $\tilde{\kappa} = 1$, we can show that \eqref{OSG-eq7} is strict. Suppose this is not the case and we have \begin{equation}
    \int s_t^{*}(x)\mu_t^{*}(dx)- C_t(\mu_t^{*}) = \int \tilde{s}(x)\mu_t^{*}(dx)-C_t(\mu_t^{*}).
\label{OSG-eq16}
\end{equation}  

From \eqref{OSG-eq17} we have that
\begin{equation*}
\int \tilde{s}(x)\tilde{\mu}_t(dx) - C_t(\tilde{\mu}_t) = \max_{\alpha\in[\underline{c},\overline{c}]} \big[\frac{\tilde{\kappa}}{c}(\alpha-d)+ \tilde{b} - K_t(\alpha)\big] > \int \tilde{s}(x)\mu_t^{*}(dx) - C_t(\mu_t^{*}).
\end{equation*}

Combined with \eqref{OSG-eq16}, it implies that $\int \tilde{s}(x)\tilde{\mu}_t(dx) - C_t(\tilde{\mu}_t) > \int s_t^{*}(x)\mu_t^{*}(dx)- C_t(\mu_t^{*})$, which violates incentive compatibility. Hence, the inequality \eqref{OSG-eq7} must be strict when $\tilde{\kappa} = 1$.

Suppose it holds that $\tilde{\kappa} = \kappa_t^{*}$ yet $b_t^{*} > \tilde{b}$. Then \eqref{OSG-eq7} is strict for type $t$. In conclusion, the replacement leads to a strict improvement in the principal's profit.\\
\\
\textbf{Case 2b revisited: }In this case, $\argmax_{t\in T\backslash (\Pi_0)\cup \overline{T}}$ is empty, and $\tilde{s}(\cdot)\notin \{s_t^{*}(\cdot)\}_{t\in T\backslash (\Pi_0)\cup \overline{T}}$. In this case $\Pi_1$ is exactly $T\backslash (\Pi_0\cup \overline{T})$. If $\tilde{\kappa} < 1$, similar to Case 2a, we can show that for all $t\in \Pi_1$, the inequality in \eqref{OSG-eq12} is strict, thereby establishing that the replacement strictly increases the principal's profit. 

If $\tilde{\kappa} = 1$, we aim to show that the inequality in \eqref{OSG-eq7} is strict for all $t\in \Pi_1$. 

Given $t\in \Pi_1$ and $\kappa_t^{*} < \tilde{\kappa} = 1$, suppose that the inequality in \eqref{OSG-eq7} is not strict, then we have that
\begin{equation}
\int s_t^{*}(x)\mu_t^{*}(dx) - C_t(\mu_t^{*}) = \int \tilde{s}(x)\mu_t^{*}(dx) - C_t(
\mu_t^{*}). 
\label{OSG-eq13}
\end{equation}

Similar to Case 2a, we can establish that
\begin{equation}
\begin{aligned}
\int \tilde{s}(x)\tilde{\mu}_t(dx) - C_t(\tilde{\mu}_t) > \int s_t^{*}(x)\mu_t^{*}(dx) - C_t(\mu_t^{*}).
\end{aligned}
\label{OSG-eq14}
\end{equation}

On the other hand, there exists a sequence $\{s^n(\cdot)\}_{n\in \mathbb{N}}\subseteq \{s_t^{*}(\cdot)\}_{t\in T\backslash (\Pi_0\cup \overline{T})}$, where $s^n(\cdot)$ can be written as $s^n(x) = \kappa^n \cdot x + b^n$, with $\kappa_t^{*} < \kappa^n < 1$ for all $n\in \mathbb{N}$, satisfying that $\lim_{n\to \infty} \kappa^n= 1$. As the next step, we show that $\lim_{n\to\infty} b^n = \tilde{b}$. 

By the definition of $\tilde{b}$ in Case 2b, there is $b^n \geq \tilde{b}$ for all $n\in \mathbb{N}$. By the incentive compatibility of the type-$t$ agent, it holds that, for all $n\in \mathbb{N}$,
\begin{equation}
\int s_t^{*}(x)\mu_t^{*}(dx) \geq \int s^n(x)\mu_t^{*}(dx) = \kappa^n \cdot \mathbb{E}[\mu_t^{*}] + b^n \geq \kappa^n \cdot \mathbb{E}[\mu_t^{*}] + \tilde{b}.
\label{OSG-eq15}
\end{equation}
Take $n\to \infty$ on the right-hand side of \eqref{OSG-eq15}, we have that $\lim_{n\to \infty} \kappa^n\cdot \mathbb{E}[\mu_t^{*}] + \tilde{b} = \tilde{\kappa}\cdot \mathbb{E}[\mu_t^{*}] + \tilde{b} = \int \tilde{s}(x)\mu_t^{*}(dx) = \int s_t^{*}(x)\mu_t^{*}(dx)$, where the last term is equal to the left-hand side of \eqref{OSG-eq15}. It indicates that $\lim_{n\to \infty} \kappa^n \mathbb{E}[\mu_t^{*}] + b^n = \lim_{n\to \infty} \kappa^n \mathbb{E}[\mu_t^{*}] + \tilde{b}$, which implies that $\lim_{n\to \infty} b^n = \tilde{b}$. 

Let $\mu_t^n$ denote some best response of the type-$t$ agent to contract $s^n(\cdot)$. According to Lemma \ref{PRE-lm2}, for sufficiently large $n$, we have that 
\begin{equation*}
\begin{aligned}
\int s^n(x) \mu_t^n(dx) - C_t(\mu_t^n) &= \max_{\alpha\in[\underline{c},\overline{c}]} \big[\frac{\kappa^n}{c}\cdot (\alpha-d) + b^n - K_t(\alpha)\big]\\
& = \frac{\kappa^n}{c}\big[k_t^{-1}(\frac{\kappa^n}{c})-d\big] + b^n - K_t(k_t^{-1}(\frac{\kappa^n}{c})),
\end{aligned}
\end{equation*}
and
\begin{equation*}
\begin{aligned}
\int \tilde{s}(x) \tilde{\mu}_t(dx) - C_t(\tilde{\mu}_t) &= \max_{\alpha \in [\underline{c},\overline{c}]}\big[\frac{\tilde{\kappa}}{c}\cdot(\alpha-d) + \tilde{b} - K_t(\alpha)\big]\\
&= \frac{\tilde{\kappa}}{c}\big[k_t^{-1}(\frac{\tilde{\kappa}}{c})-d\big] + \tilde{b} - K_t(k_t^{-1}(\frac{\tilde{\kappa}}{c})).
\end{aligned}
\end{equation*}

Given $\lim_{n\to\infty}\kappa^n = \tilde{\kappa}$, $\lim_{n\to \infty} b^n = \tilde{b}$ and the continuity of $k_t(\cdot)$, we have that 
\begin{equation*}
\lim_{n\to\infty}\int s^n(x) \mu_t^n(dx) - C_t(\mu_t^n) = \int \tilde{s}(x) \tilde{\mu}_t(dx) - C_t(\tilde{\mu}_t) > \int s_t^{*}(x)\mu_t^{*}(dx) - C_t(\mu_t^{*}),
\end{equation*}
where the last inequality follows from \eqref{OSG-eq14}. Then there exists $n\in \mathbb{N}$ such that 
\begin{equation*}
\int s^n(x)\mu_t^n(dx) - C_t(\mu_t^n)>\int s_t^{*}(x)\mu_t^{*}(dx) - C_t(\mu_t^{*}),
\end{equation*}
which is contradictory to the incentive compatibility of the type-$t$ agent. Hence, we have that the inequality in \eqref{OSG-eq7} is strict for type $t\in \Pi_1$ such that $\kappa_t^{*} < \tilde{\kappa} = 1$. In conclusion, the replacement leads to a strict increase in the principal's profit. $\hfill\square$

\end{document}